\documentclass[aps,prx,10pt,twocolumn,nobibnotes,nofootinbib,longbibliography,floatfix,superscriptaddress]{revtex4-2}
\usepackage[T1]{fontenc}
\usepackage{amsmath,amssymb,amsthm,mathtools,mathrsfs,microtype}
\usepackage{graphicx,xcolor,tikz,booktabs}
\usetikzlibrary{arrows.meta,positioning,calc,decorations.pathreplacing}
\usepackage[colorlinks=true,linkcolor=blue!45!black,citecolor=blue!45!black,urlcolor=blue!45!black]{hyperref}
\usepackage[capitalise,nameinlink]{cleveref}
\crefname{section}{Sec.}{Secs.}
\crefname{subsection}{Sec.}{Secs.}
\definecolor{qecnavy}{RGB}{46,64,87}        % fixed QEC operations, quantum wires
\definecolor{qecpetrol}{RGB}{0,121,140}      % classical records and calibration
\definecolor{qecraspberry}{RGB}{209,73,91}   % noise
\definecolor{qecsaffron}{RGB}{237,174,73}    % decoder and final recovery
\newcommand{\Tr}{\operatorname{Tr}}
\newcommand{\C}{\mathcal C}
\newcommand{\Dtr}{D_{\mathrm{tr}}}
\newcommand{\TV}{D_{\mathrm{TV}}}
\newcommand{\Xbar}{\bar X}
\newcommand{\Zbar}{\bar Z}
\newcommand{\FP}[1]{F_{#1}^{\mathrm P,\star}}
\newcommand{\FC}[1]{F_{#1}^{\mathrm{CPTP},\star}}
\newcommand{\id}{\operatorname{id}}
\newcommand{\Fid}[1]{F_{#1}}
\newcommand{\Floss}[1]{\Delta F_{#1}}
\newcommand{\Fmm}[1][]{\Delta F^{\star}\if\relax\detokenize{#1}\relax\else_{#1}\fi}
\newcommand{\RefInf}[2][]{\varepsilon_{#2}\if\relax\detokenize{#1}\relax\else^{#1}\fi}

\theoremstyle{plain}
\newtheorem{theorem}{Theorem}
\newtheorem{proposition}{Proposition}
\newtheorem{lemma}{Lemma}
\newtheorem{corollary}{Corollary}
\theoremstyle{definition}
\newtheorem{definition}{Definition}

\hypersetup{pdftitle={What Must a Quantum-Memory Decoder Know About Temporally Correlated Noise?}}
\begin{document}
\raggedbottom
\title{What Must a Quantum-Memory Decoder Know About Temporally Correlated Noise?}

\author{Danesh Morales-Hashemi}
\email{danesh.moraleshashemi@uwaterloo.ca}
\author{Graeme Smith}
\email{graeme.smith@uwaterloo.ca}
\affiliation{%
Institute for Quantum Computing and Department of Applied Mathematics,
University of Waterloo, Waterloo, Ontario N2L 3G1, Canada
}

\author{Ningping Cao}
\email{ncao@uwaterloo.ca}
\affiliation{Digital Technologies, National Research Council Canada, Canada}
\affiliation{%
Institute for Quantum Computing and Department of Applied Mathematics,
University of Waterloo, Waterloo, Ontario N2L 3G1, Canada
}
\date{\today}

\begin{abstract}
What must a quantum error-correction decoder know about temporally correlated noise, and what does it cost to learn? We study this question for fixed stabilizer memory experiments with known system--environment interactions. We formulate calibration in terms of the fidelity differences that select a Pauli recovery and relate missing information to the resulting fidelity loss.
Our main example uses weak coherent $X$ rotations in a class of CSS codes including Shor, Steane, and odd-distance rotated surface codes. All physical qubits share a sign either fixed throughout a run or independently redrawn each interval. The two models have identical syndrome statistics and the same spatiotemporal Pauli process (SPP), obtained by Pauli-twirling every noise interval. Yet at suitable storage times, they require different logical corrections. In the weak-noise limit, each model is almost perfectly correctable when known. Without model information, even the best final recovery has a worst-case fidelity loss approaching $1/2$ relative to the noise-aware optimum, despite access to the full syndrome record.
For this pair of models, at fixed code and fixed positive target loss below $1/2$, we prove matching calibration bounds. With logical preparation and final readout, complete QEC cycles require a number of noise intervals scaling as the inverse rotation angle raised to the code's $X$ distance. Allowing two uninterrupted intervals before one extraction reduces the optimal cost to inverse-square scaling.
Physical Pauli twirling of every noise interval changes both recovery performance and its information requirements. Discarding the record of applied random Pauli operations makes the SPP sufficient for optimal recovery. Keeping this record can improve the attainable fidelity, but reaching that optimum may require additional noise information.
\end{abstract}
\maketitle

%\tableofcontents
\section{Introduction}\label{sec:intro}

Repeated stabilizer measurements produce the classical record from which the decoder of an error-corrected memory chooses a recovery. Experiments now run long sequences of error-correction cycles and give their decoders increasingly accurate noise models \cite{krinner2022,google2025}. The syndrome record is also a natural calibration resource, because it is generated during storage without measuring the unknown logical state. This raises two questions about temporally correlated noise: which of its properties must calibration determine to select a good recovery, and how much experimentation does that take?

A slowly drifting field and a rapidly fluctuating field illustrate the problem. Both can produce the same rotation strength in each noise interval, yet their effects accumulate differently over a memory run. We compare a fixed-sign model, $\mathsf F$, in which all physical qubits share a single random rotation sign for the entire run, with a resampled-sign model, $\mathsf R$, in which this common sign is redrawn independently in every interval. The two models share the single-interval channel and, for the codes studied here, also the distribution of complete syndrome histories and the multi-time Pauli-twirled description. Nevertheless, they can require different logical corrections after an identical storage circuit. The syndrome determines the branch rotation angle up to the unobserved field sign. For either fixed sign, a nonzero mean logical rotation accumulates under $\mathsf F$, whereas independent resampling removes this drift under $\mathsf R$.

Several lines of work show how error-correction data reveal noise properties that matter for decoding. Pauli channels and logical Pauli noise can be learned from syndrome measurements under structural assumptions \cite{wagner2022,wagner2023}, and work on Pauli-noise learning separates learnable functions of the noise from unidentifiable parameters \cite{chen2023}. Syndrome-based estimators can track drifting error rates and update decoding weights within stochastic Pauli and detector-error models~\cite{huoli2017,spitz2018}. Syndrome-based logical-error estimation \cite{takou2026}, detector error models for specified coherent errors \cite{takoubrown}, and optimization of decoder priors from logical outcomes \cite{sivak} provide complementary approaches.

For temporally correlated noise, process tensors capture how the noise responds to the interventions of a decoding experiment \cite{pollock2018,kobayashi2024}. In this language, $\mathsf F$ is non-Markovian, with a purely classical memory, and $\mathsf R$ is Markovian. Spatiotemporal Pauli processes (SPPs) describe the joint distributions of Pauli histories \cite{kam2026}. Threshold theorems establish scalable fault tolerance under sufficiently weak local non-Markovian noise~\cite{terhal2005,agp2006}. Here we fix the memory circuit and ask a narrower question than full noise characterization: even when the data do not determine the noise model, do they determine a recovery with the desired fidelity?

Our main result concerns long storage times. For a class of CSS codes that includes the Shor, Steane, and odd-distance rotated surface codes, weak physical $X$ rotations with a common sign give identical syndrome statistics and the same SPP under $\mathsf F$ and $\mathsf R$. The logical evolution, however, differs. Let $\theta$ be the rotation parameter of each noise interval and $d_X$ the $X$ distance of the code. For either realization of the fixed sign, a logical rotation accumulates coherently. After a storage time $T_\pi\propto\theta^{-d_X}$, the mean logical-$Z$ readout of a $\bar Z=+1$ input is close to $-1$ under $\mathsf F$ and $+1$ under $\mathsf R$. Each model is then almost perfectly correctable when known. Without model information, the smallest worst-case loss relative to the noise-aware optimum approaches $1/2$ in entanglement fidelity, even when recovery uses the entire syndrome history and an arbitrary final CPTP map. Coherent physical errors can leave residual logical rotations after syndrome measurement~\cite{huang2019,beale2018}, and these rotations can accumulate over repeated cycles~\cite[Sec.~3]{iverson2020}. Related analyses address coherent errors in surface codes~\cite{bravyi2018,pataki2024,behrends_beri_2025_coherence}. Rajmohan et al. derive repeated-cycle logical channels for correlated coherent Pauli rotations and compare dynamical decoupling and Pauli twirling, including their dependence on the encoding stabilizer eigenspace~\cite{rajmohan2026}. Here we use this structure to quantify the recovery-relevant information that syndrome statistics and the SPP leave undetermined.

Independent calibration runs can supply this information while the QEC circuit that stores the data stays unchanged. With logical preparation and a final logical readout, calibration by complete QEC cycles needs of order $\theta^{-d_X}$ noise intervals, and this order is optimal. Allowing two noise intervals before a single syndrome extraction reduces the optimal cost to order $\theta^{-2}$, and allowing longer delays reduces it further, down to order $\theta^{-1}$. The lower bounds hold even when calibration retains quantum outputs and processes them jointly with the data at the final recovery. Concurrent work examines how probe preparation changes the cost of learning spatial noise correlations relevant to QEC control~\cite{deng2026}.

We treat this question as a statistical decision problem for a given memory experiment. When the system--environment interactions are known, the fidelity contribution of each candidate Pauli correction is linear in the unknown initial environment state. Standard tools of observable identification and statistical decision theory \cite{dariano2008,jencova2021,lecam,helstrom} then answer three questions: which fidelity differences calibration determines, whether the remaining ambiguity changes the optimal correction, and how limited distinguishability of the calibration data translates into lost recovery fidelity. The framework builds on channel-adapted and robust recovery optimization \cite{fletcher2007,kosut2008,danageozian2023}. A two-round Shor-code example shows how it pinpoints one missing temporal correlation and guides the design of a calibration experiment that measures it.

Adding control operations to the QEC circuit offers a different route. Randomized compiling can improve QEC performance under coherent errors~\cite{jain2023}. In our setting, ideal logical dynamical decoupling (a midpoint Hahn echo) or independent logical Pauli twirling over $\{I,\bar Z\}$ removes the order-one conflict between the $\mathsf F$ and $\mathsf R$ corrections, and independent physical Pauli twirling of every noise interval can even reduce the leading-order infidelity below that of either noise-aware recovery in the original circuit. What randomization does for decoding depends on the records the decoder keeps. When the record of applied random Paulis is discarded, the SPP suffices for optimal recovery. In a separate exact example, however, discarding the record lowers the noise-aware optimum from one to one half. Keeping the record preserves this optimum, but attaining it can require additional noise information. Retained random choices have also been used to estimate coherent errors \cite{vyas2026}.

Section~\ref{sec:setting} defines memory experiments, final decoding, and calibration. Section~\ref{sec:general} develops the decision framework and the two-round Shor example. Section~\ref{sec:family} establishes the $\mathsf F/\mathsf R$ recovery gap and extends it to a physical process that requires quantum memory. Section~\ref{sec:calibration} proves matching calibration costs for $\mathsf F/\mathsf R$, Sec.~\ref{sec:protection} compares modified QEC circuits, and Sec.~\ref{sec:discussion} discusses the scope of the results and open questions. Appendix~\ref{app:notation} lists the notation.
\section{Preliminaries and setting}
\label{sec:setting}

\subsection{The physical process and the recorded logical instrument}

A memory experiment stores an unknown logical state, which may be entangled with a reference system that the experiment never acts on. Write $\mathcal H_L\coloneqq\mathbb C^D$, with $D\coloneqq 2^k$, for the logical space and $\mathcal H_{\rm phys}\coloneqq(\mathbb C^2)^{\otimes n}$ for the physical space. An encoding isometry $V:\mathcal H_L\to\mathcal H_{\rm phys}$ embeds the input in an $[[n,k]]$ stabilizer code. We define the code space as the simultaneous $+1$ eigenspace of the chosen stabilizer generators. The environment has space $\mathcal H_E$, and its initial state $\rho_E(h)$ is independent of the logical input. Each noise interval $U_j$ acts on $\mathcal H_{\rm phys}\otimes\mathcal H_E$. Noise intervals alternate with prescribed syndrome extraction, correction, and control operations, and the physical noise process can carry classical or quantum memory from one interval to the next. Figure~\ref{fig:setting} shows how these storage operations, the final recovery, and the information supplied by calibration fit together.

We write $\mathcal B(\mathcal H,\mathcal K)$ for linear operators from $\mathcal H$ to $\mathcal K$, with $\mathcal B(\mathcal H)\coloneqq\mathcal B(\mathcal H,\mathcal H)$. The Hermitian operators on $\mathcal H$ form the real vector space $\operatorname{Herm}(\mathcal H)$. A syndrome outcome $s$ is recorded as a classical basis state $|s\rangle$. In ideal extraction the syndrome projectors $\Pi_s$ and table corrections $R_s$ act on $\mathcal H_{\rm phys}$.

\begin{figure*}[t]
\centering
% Circuit body shared by the manuscript and the standalone preview.
\begin{tikzpicture}[x=1cm,y=0.8cm,font=\small,line cap=round,
 wire/.style={line width=0.65pt},
 block/.style={draw=qecnavy,fill=qecnavy!9,rounded corners=2pt,
   line width=0.7pt,align=center,inner sep=4pt,minimum height=0.85cm},
 info/.style={draw=qecpetrol,fill=qecpetrol!11,rounded corners=3pt,
   line width=0.7pt,align=center,inner sep=5pt,minimum height=0.72cm},
 focus/.style={draw=qecsaffron!70!black,fill=qecsaffron!35,line width=1pt},
 classical/.style={draw=qecpetrol,dashed,line width=0.7pt,
   -{Latex[length=1.5mm]}},
 quantum/.style={draw=qecnavy,line width=0.8pt,
   -{Latex[length=1.5mm]}}]

 % The physical data and environment pass through the same T-cycle block.
 \draw[wire] (0.65,0)--(15.75,0);
 \draw[wire] (2.02,1.13)--(9.30,1.13);
 \node[anchor=east] at (0.57,0) {$\rho_L$};
 \node[anchor=east] at (1.97,1.13) {$\rho_E(h)$};
 \node[anchor=west] at (9.34,1.13) {$\operatorname{Tr}_E$};
 \node[anchor=west] at (15.78,0) {$\rho'_L$};
 \draw[black!45,dashed,rounded corners=4pt]
   (2.39,-1.10) rectangle (8.83,1.94);
 \node[fill=white,inner sep=2pt,text=black!75] at (5.61,1.94)
   {fixed QEC cycle $j=1,\ldots,T$};

 \node[block,minimum width=0.85cm] (encode) at (1.38,0)
   {$V$\\encode};
 \node[block,draw=qecraspberry,fill=qecraspberry!12,
   minimum width=1.20cm,minimum height=1.72cm] (noise) at (3.32,0.54)
   {$U_j$\\joint noise};
 \node[block,minimum width=1.55cm] (extract) at (5.34,0)
   {$\mathcal M$\\syndrome};
 \node[block,minimum width=1.70cm] (correct) at (7.65,0)
   {$R_{s_j}$\\fixed table};
 \node[block,minimum width=1.40cm] (decode) at (10.38,0)
   {$\mathcal D$\\logical output};
 \node[block,focus,minimum width=2.30cm] (recover) at (13.69,0)
   {Final recovery\\$L$, $\mathcal R_y$, or $\mathcal G$};

 % The reference correction depends on this cycle's syndrome. The final
 % decoder has access to the entire record, not only the final syndrome.
 \draw[classical] (extract.south)--(5.34,-0.82)--(7.65,-0.82)
   --(correct.south);
 \node[fill=white,inner sep=1pt,text=qecpetrol] at (6.50,-0.82) {$s_j$};
 \fill[qecpetrol] (5.34,-0.82) circle (1.3pt);
 \node[info,minimum width=5.40cm] (calibrate) at (3.93,-2.31)
   {\textbf{Separate calibration runs}\\known inputs, same model $h$};
 \node[info,focus,minimum width=3.35cm] (decoder) at (12.53,-2.31)
   {\textbf{Decoder}\\$L$ from $y$ and $z$};
 \draw[classical] (5.34,-0.82)--(5.34,-1.38)--(11.37,-1.38)
   --([xshift=-1.16cm]decoder.north);
 \node[below,inner sep=2pt,text=qecpetrol] at (8.63,-1.38)
   {record $y=(s_1,\ldots,s_T)$};
 \draw[classical] (calibrate.east)--node[below,inner sep=3pt,text=qecpetrol]
   {classical data $z$}(decoder.west);
 \draw[classical] ([xshift=1.11cm]decoder.north)
   --([xshift=-0.05cm]recover.south);

 % Quantum calibration output is delivered directly to final processing.
 \draw[quantum] (calibrate.south)--(3.93,-3.60)--(15.28,-3.60)
   --(15.28,-0.86)--(14.47,-0.86)
   --([xshift=0.78cm]recover.south);
 \node[fill=white,inner xsep=5pt,text=qecnavy] at (9.60,-3.60)
   {retained quantum output $Q_{\rm cal}$ in state $\sigma_h$ (optional)};
 % The recorded logical instrument of Eq. (1): encoding, all cycles, and D.
 \draw[decorate,decoration={brace,amplitude=5pt},black!60,line width=0.6pt]
   (0.93,2.30)--(11.10,2.30);
 \node[anchor=south,inner sep=2pt,text=black!80] at (6.02,2.42)
   {recorded logical instrument $\Phi_h$: record $y$ and residual logical map $\mathcal N_{h,y}$};
\end{tikzpicture}
\caption{\label{fig:setting} Memory experiment with prescribed cycle corrections $R_{s_j}$ and a final recovery chosen from the full syndrome record $y$ and calibration transcript $z$. Yellow marks the decoder and its final recovery, the subject of this paper. Known interactions $U_j$ couple the encoded system to an unobserved environment initially in $\rho_E(h)$, $\mathcal M$ extracts the syndrome $s_j$ in cycle $j=1,\ldots,T$, and $\mathcal D$ returns the residual logical output before final recovery. The logical input and recovered output are $\rho_L$ and $\rho'_L$. The brace marks the recorded logical instrument $\Phi_h$ [\cref{eq:instrument}]. The decoder applies a Pauli correction $L$ with probability $d(L\mid y)$, or more generally a CPTP recovery $\mathcal R_y$. Calibration runs share the model $h$ with the data run but use independently prepared environments, and a retained quantum output $Q_{\rm cal}$ in state $\sigma_h$ can enter a joint final recovery $\mathcal G$. The drawing shows active corrections and identity logical controls between cycles, and \cref{lem:frame-defer} gives conditions for tracking the corrections in a Pauli frame instead.}
\end{figure*}
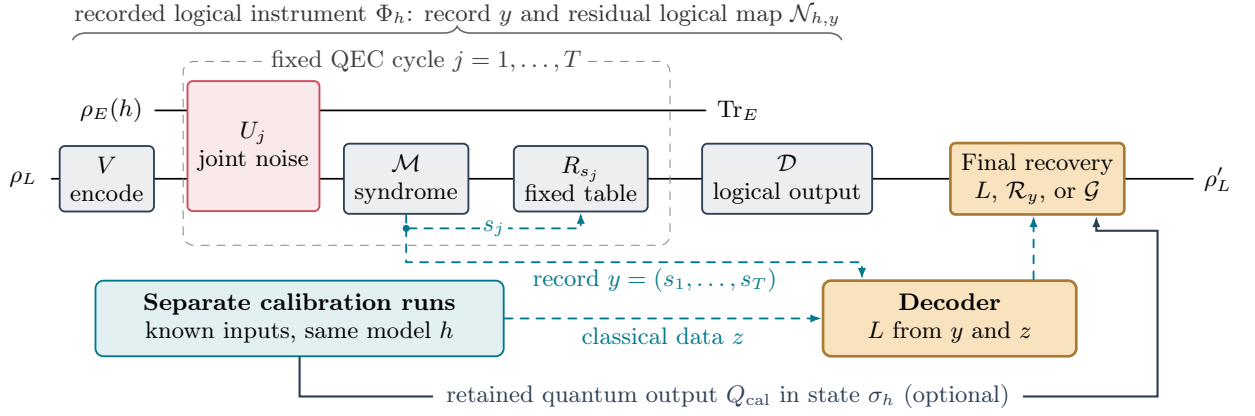

A process tensor, or quantum comb, generalizes a quantum channel to experiments with interventions at several times~\cite{chiribella2009,pollock2018}. It takes the operations applied to the system at the intermediate times as inputs, or slots, and returns the resulting final state and measurement records, with the environment traced out. Without environmental memory it factorizes into independent single-interval channels. The fixed-sign model $\mathsf F$ of \cref{sec:family} is a simple process that does not. Inserting the encoding and the prescribed QEC operations gives a \emph{logical process tensor with syndrome records}: its open slots accept logical interventions, and its classical output contains the reported syndromes. Process-tensor approaches to decoding use this object~\cite{kobayashi2024}. Once the intermediate logical schedule is also fixed, the result is the recorded logical instrument
\begin{equation}
 \Phi_h(\rho)=\sum_{y\in\mathcal Y}|y\rangle\langle y|
                         \otimes\mathcal N_{h,y}(\rho),
 \qquad \sum_y\mathcal N_{h,y}\text{ is CPTP}.
 \label{eq:instrument}
\end{equation}
Here $h$ labels the unknown noise model (below, simply the model), and $y\in\mathcal Y$ is the complete record supplied to the decoder, with $\mathcal Y$ the finite set of possible records. The records are stored in a classical register with space $\mathcal H_Y\coloneqq\operatorname{span}\{|y\rangle:y\in\mathcal Y\}$, so $\Phi_h$ maps $\mathcal B(\mathcal H_L)$ to $\mathcal B(\mathcal H_Y\otimes\mathcal H_L)$. The residual logical map $\mathcal N_{h,y}$ acts on $\mathcal B(\mathcal H_L)$ and returns the unnormalized logical output that accompanies $y$. It describes the logical action left after the circuit's own syndrome extraction and corrections, before the final recovery. Each $\mathcal N_{h,y}$ is completely positive, and $\Tr\mathcal N_{h,y}(\rho)$ is the probability of recording $y$ for input $\rho$. Any decoding map, leakage handling, or return to the code space is fixed as part of the experiment, and failure outcomes remain in the record. The decoder selects only the final recovery. The general framework allows noisy extraction and prescribed feedback, whereas the explicit $\mathsf F/\mathsf R$ construction uses ideal syndrome extraction and an identity logical schedule. Appendix~\ref{app:process-conventions} specifies our Choi conventions.

\subsection{Recovery actions and fidelity loss}

Let $\mathcal P_k$ be a phase-fixed set of $D^2$ Hermitian logical Paulis. A Pauli decoder $d$ applies the final correction $L\in\mathcal P_k$ with probability $d(L\mid y)$ when the record is $y$, where $\sum_Ld(L\mid y)=1$ for every $y$. A deterministic decoder assigns probability one to a single correction for each record. We allow randomized decoders because randomized choices attain the minimax values in \cref{sec:general,sec:family}. For Kraus operators $K_{h,y,\mu}\in\mathcal B(\mathcal H_L)$ of $\mathcal N_{h,y}$, define the unnormalized branch contributions
\begin{equation}
 w_h(y,L)\coloneqq D^{-2}\sum_\mu|\Tr(LK_{h,y,\mu})|^2.
 \label{eq:branch}
\end{equation}
We measure the quality of a recovery by its entanglement fidelity with the identity logical channel~\cite{schumacher1996}: the overlap of the final joint state with a maximally entangled state of the logical system and a reference system of the same dimension. This figure of merit suits a quantum memory for two reasons. It measures how well the memory preserves entanglement with the reference, and it is an affine function of the fidelity averaged over pure input states, $\bar F=(DF_{\mathrm e}+1)/(D+1)$ for entanglement fidelity $F_{\mathrm e}$~\cite{horodecki1999,nielsen2002}. It is also linear in the logical channel, so the fidelity contributed by each candidate recovery is linear in the unknown initial environment state. This linearity underlies the decision framework of \cref{sec:general} and makes entanglement fidelity the standard objective of channel-adapted and robust recovery~\cite{barnum2002,reimpell2005,fletcher2007,kosut2008}. For a Pauli decoder $d$, the entanglement fidelity and the noise-aware Pauli optimum are
\begin{equation}
 \begin{aligned}
 \Fid{h}(d)&=\sum_{y,L}d(L\mid y)w_h(y,L),\\
 \FP{h}&=\sum_y\max_L w_h(y,L).
 \end{aligned}
 \label{eq:fidelity}
\end{equation}
We call $\Floss{h}^{\mathrm P}(d)\coloneqq \FP{h}-\Fid{h}(d)$ the \emph{fidelity loss}. The noise-aware optimum uses the same record and the same recovery actions as the decoder. It knows $h$ but neither the realized environmental trajectory nor the unknown input. A final CPTP recovery $\mathcal R=\{\mathcal R_y\}_{y\in\mathcal Y}$ applies a CPTP map $\mathcal R_y:\mathcal B(\mathcal H_L)\to\mathcal B(\mathcal H_L)$ to the logical output when the record is $y$, and $\Fid{h}(\mathcal R)$ is the entanglement fidelity of $\sum_y\mathcal R_y\circ\mathcal N_{h,y}$. The choice $\mathcal R_y(\rho)=\sum_Ld(L\mid y)L\rho L$ recovers $\Fid{h}(d)$. We write $\FC{h}$ for the optimum over arbitrary record-dependent final CPTP recoveries and define the corresponding loss $\Floss{h}^{\mathrm{CPTP}}(\mathcal R)\coloneqq \FC{h}-\Fid{h}(\mathcal R)$. For coherent errors, worst-case measures such as the diamond distance can scale differently from the infidelity~\cite{gilchrist2005,sanders2016,kueng2016}. Half the diamond distance of a recovered logical channel from the identity is at least one minus its entanglement fidelity, and hence at least the corresponding fidelity loss. Every fidelity loss reported below is therefore also a lower bound on the worst-case error.

Pauli orthogonality gives $\sum_{y,L}w_h(y,L)=1$ and $\sum_Lw_h(y,L)=p_h(y)$, where $p_h(y)\coloneqq \Tr\mathcal N_{h,y}(I/D)$ is the probability of record $y$ for a maximally mixed logical input. Both identities hold even when syndrome probabilities depend on the input. For stochastic Pauli noise, $w_h(y,L)$ is the joint probability of the record and the corrected logical error class. For coherent noise, amplitudes of different physical error paths interfere inside the squared magnitude in Eq.~\eqref{eq:branch}, so the contributions are fidelity weights rather than probabilities of physical error trajectories.

Equivalently, these weights are the Pauli--Choi diagonal entries of $J(\mathcal N_{h,y})/D$, the logical branch Choi operator rescaled by the input dimension. Coherences between different logical Pauli histories can contribute to these diagonal entries when the intermediate operations are inserted (Appendix~\ref{app:choi-weights}).

\subsection{Final decoding and frame tracking}
\label{sec:frame-tracking}

Repeated syndrome extraction followed by a decoding decision based on the full record is an established memory workflow~\cite{krinner2022,google2025}. The recovery may be applied physically, tracked as a Pauli frame, or used to interpret a final destructive readout. Whether active and deferred corrections are equivalent depends on the noise and the circuit. When Pauli errors anticommute with coherent noise, tracking them without physical correction can suppress its accumulation across cycles~\cite[Sec.~IV~B]{witzel2026}. \Cref{lem:frame-defer} in \cref{app:frame-defer} shows that they are equivalent whenever every noise interval commutes with the $X$-type corrections and no logical gates intervene. The $X$-rotation models of \cref{sec:family} satisfy these conditions. Every syndrome-compatible final $X$-Pauli decoder then differs from the accumulated frame only by a record-dependent choice of $I$ or $\Xbar$. Changing the reference correction table therefore changes only a known logical frame, and at a given storage time all tables give the same attainable fidelities, noise-aware optima, and minimax fidelity loss. Quantities computed from a particular table, such as the coefficients and characteristic storage time of \cref{sec:family}, can differ between tables, so we fix one reference table before defining them. The general instrument framework does not require deferral.

\subsection{Calibration}

Calibration consists of experiments with known inputs, used to select the recovery for a later memory experiment on the unknown state, which we call the data run. Calibration runs and the data run share the unknown noise model $h$, while their environment states are prepared independently given $h$. Each access class, defined in Sec.~\ref{sec:access}, specifies the permitted preparations, controls, syndrome-extraction times, environment resets, and retained outputs. A retained output may be a classical transcript or a quantum state. A quantum output interacts with the data only at the final recovery, so the QEC circuit for the data run is the same with or without calibration.

In the identification criterion of Sec.~\ref{sec:general}, the environment dimension and all interactions are known and only the initial environment state is unknown. In the cost results, the code and rotation angle are known, and the only unknown is which of the two temporal models holds. Calibration cost counts physical noise intervals to isolate the role of extraction timing. Section~\ref{sec:access} specifies the free preparations, controls, and readout; runs, measurement settings, postprocessing, and wall-clock time are separate resources.

\subsection{Spatiotemporal Pauli process}

The Pauli twirl of a channel on $n$ qubits averages its conjugation by all physical Paulis:
\begin{equation}
 \mathcal T(\mathcal E)\coloneqq 4^{-n}\sum_{Q\in\mathcal P_n}
        \operatorname{Ad}_{Q}^{-1}\circ\mathcal E\circ
        \operatorname{Ad}_{Q},
 \qquad \operatorname{Ad}_Q(\rho)\coloneqq Q\rho Q^\dagger.
 \label{eq:pauli-twirl}
\end{equation}
An independent twirl at every noise interval removes the coherences between distinct Pauli histories. The surviving diagonal coefficients form the spatiotemporal Pauli process (SPP), the joint distribution over sequences of physical Pauli errors~\cite{kam2026}. The SPP retains the temporal correlations among Pauli errors and carries more information than the single-interval marginals. For a coherent process, it is generally a reduced description. For the physically randomized experiment, it gives a complete description when the randomization record is discarded. Retaining that record gives a conditional experiment that may require additional noise information (Sec.~\ref{sec:protection}).

\section{From calibration data to recovery decisions}
\label{sec:general}

This section asks what calibration must determine to select a recovery. We divide this question into three parts. First, which differences between the fidelities of candidate corrections do the calibration data determine, and which remain undetermined (\cref{sec:identification})? Second, when some differences remain undetermined, does the unknown model change which correction is optimal, or is a single decoder optimal for every model consistent with the data (\cref{sec:decision})? Third, with a finite calibration budget, how much fidelity must any recovery lose when the calibration outputs of different models are hard to distinguish (\cref{sec:calibration-loss})? \Cref{sec:shor-example} illustrates the first two questions, and \cref{app:computation} discusses the computational cost of evaluating these criteria. Sections~\ref{sec:family} and~\ref{sec:calibration} then establish the physical gaps and the matching costs.

\subsection{Identifying the quantities that select a recovery}
\label{sec:identification}

We describe the data circuit by a single known dilation shared by all candidate models: the system--environment interactions, syndrome extraction, and corrections are fixed and known, and the models differ only in the initial environment state $\rho_E$. A finite family of models always admits such a description, because a classical environment register can record which model acts. Let $\mathcal H_{E'}$ denote the final environment space. The corrected joint branches $\widetilde K_{y,\mu}$ map $\mathcal H_L\otimes\mathcal H_E$ to $\mathcal H_L\otimes\mathcal H_{E'}$ and have the Pauli expansion
\begin{equation}
 \begin{aligned}
 \widetilde K_{y,\mu}&=\sum_{L\in\mathcal P_k}L\otimes B_{y,L,\mu},\\
 E_{y,L}&\coloneqq \sum_\mu B_{y,L,\mu}^\dagger B_{y,L,\mu}.
 \end{aligned}
 \label{eq:environment-effects}
\end{equation}
Here $B_{y,L,\mu}\coloneqq D^{-1}\Tr_{\mathcal H_L}[(L\otimes I_{E'})\widetilde K_{y,\mu}]$ belongs to $\mathcal B(\mathcal H_E,\mathcal H_{E'})$. The fidelity effects $E_{y,L}$ act on $\mathcal H_E$. The expansion is unique because $\Tr(LL')=D\delta_{L,L'}$. The index $\mu$ includes unobserved outcomes. Writing $w_{\rho_E}(y,L)$ for $w_h(y,L)$ when the initial environment state is $\rho_E$, orthogonality of the logical Paulis gives
\begin{equation}
 \begin{gathered}
 w_{\rho_E}(y,L)=\Tr(\rho_EE_{y,L}),\\
 E_{y,L}\ge0,\qquad\sum_{y,L}E_{y,L}=I_E.
 \end{gathered}
 \label{eq:score}
\end{equation}
The difference between the scores of two corrections at the same record is therefore the expectation of a known difference operator in $\operatorname{Herm}(\mathcal H_E)$,
\begin{equation}
 \Delta E_{y;L,L'}\coloneqq E_{y,L}-E_{y,L'}.
 \label{eq:contrast}
\end{equation}
A calibration setting runs a known circuit, with known logical inputs, on an independently prepared copy of the environment state and then measures. Its outcome $i$ occurs with probability $\Tr(\rho_EC_i)$, where the environment effect $C_i$, with $0\le C_i\le I_E$, is the measurement operator of that outcome pulled back through the known calibration circuit. The index $i$ runs over all outcomes of all permitted settings, and the effects of each setting sum to $I_E$. Measurements on a retained quantum output are included in the same way. Write $\mathbf p\coloneqq(\Tr(\rho_EC_i))_i$ for the exact calibration probabilities. The fidelity effects need not be accessible within the permitted calibration settings.

\begin{figure}[tbp]
\centering
\begin{tikzpicture}[x=1cm,y=1cm,font=\footnotesize,line cap=round]
 % the calibration span, drawn as a plane through the origin
 \filldraw[draw=qecpetrol,fill=qecpetrol!10,line width=0.7pt]
   (0,0) -- (5.4,0) -- (6.8,1.35) -- (1.4,1.35) -- cycle;
 \node[anchor=south west,text=qecpetrol] at (0.3,0.04)
   {$\mathcal V_{\rm cal}=\operatorname{span}_{\mathbb R}\{I_E,C_i\}$};
 \node[anchor=north east,text=black!60] at (8.2,-0.08)
   {$\operatorname{Herm}(\mathcal H_E)$};
 % two difference operators drawn from the origin
 \coordinate (O) at (1.55,0.62);
 \fill (O) circle (1.1pt);
 \draw[qecsaffron!70!black,line width=1pt,-{Latex[length=1.6mm]}]
   (O) -- (3.6,1.02) node[anchor=west,inner sep=2pt,text=black]
   {$\Delta E\in\mathcal V_{\rm cal}$};
 \draw[qecsaffron!70!black,line width=1pt,-{Latex[length=1.6mm]}]
   (O) -- (2.45,2.15) node[anchor=south east,inner sep=1pt,text=black]
   {$\Delta E\notin\mathcal V_{\rm cal}$};
 \draw[qecraspberry,densely dotted,line width=0.7pt,-{Latex[length=1.3mm]}]
   (2.45,0.93) -- node[anchor=west,inner sep=2pt,text=qecraspberry] {$Q_E$}
   (2.45,2.08);
 % states with the same calibration data
 \coordinate (projection) at (5.35,0.72);
 \draw[qecraspberry,dashed,line width=0.7pt] (projection) -- (5.35,2.45);
 \fill[qecpetrol] (projection) circle (1.3pt);
 \node[anchor=north,text=qecpetrol,inner sep=2pt] at (projection) {data $\mathbf p$};
 \fill[qecraspberry] (5.35,1.62) circle (1.4pt);
 \fill[qecraspberry] (5.35,2.2) circle (1.4pt);
 \node[anchor=west,text=qecraspberry,inner sep=3pt,align=left] at (5.35,1.91)
   {$\rho_E\pm\eta Q_E$\\same $\mathbf p$};
\end{tikzpicture}
\caption{\label{fig:identification}Geometry of \cref{prop:span} in $\operatorname{Herm}(\mathcal H_E)$ with the Hilbert--Schmidt inner product. Since every calibration probability has the form $\Tr(\rho_EC_i)$, the data $\mathbf p$ fix only the orthogonal projection of $\rho_E$ onto the calibration span $\mathcal V_{\rm cal}$ [\cref{eq:span}], marked $\mathbf p$, and all states on the dashed line give the same data. If $\Delta E=\Delta E_{y;L,L'}$ [\cref{eq:contrast}] lies in $\mathcal V_{\rm cal}$, the score difference $\Tr(\rho_E\Delta E)$ is constant along this line, so the data determine it. Otherwise $\Delta E$ has a nonzero orthogonal component $Q_E$, and for full-rank $\rho_E$ and small $\eta$ the states $\rho_E\pm\eta Q_E$ give the same data but different score differences.}
\end{figure}
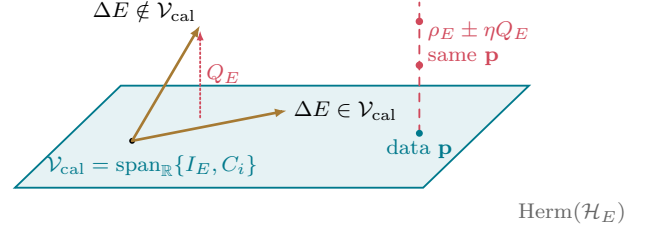

\begin{proposition}[Exact identification of fidelity differences]
\label{prop:span}
Suppose the unknown initial environment state ranges over all states on its specified finite-dimensional Hilbert space. Exact calibration probabilities determine $\Tr(\rho_E\Delta E_{y;L,L'})$ for every such state if and only if
\begin{equation}
 \Delta E_{y;L,L'}\in\mathcal V_{\rm cal},\qquad
 \mathcal V_{\rm cal}\coloneqq \operatorname{span}_{\mathbb R}\{I_E,C_i\}.
 \label{eq:span}
\end{equation}
\end{proposition}

If the difference operator lies in the span, its expectation is a known linear combination of the calibration probabilities~\cite[Sec.~II~A]{dariano2008}; see also Ref.~\cite[Sec.~V]{keith2018}. Otherwise, its component orthogonal to the span in the Hilbert--Schmidt inner product is a nonzero Hermitian $Q_E$ with $\Tr Q_E=\Tr(Q_EC_i)=0$ and $\Tr(Q_E\Delta E_{y;L,L'})\ne0$. For small $\eta$, the states $I_E/\dim\mathcal H_E\pm\eta Q_E$ then give identical calibration data but distinct values of the fidelity difference, which proves necessity. \Cref{fig:identification} illustrates this geometry.

The criterion assumes that the environment state is unrestricted. For a promised family of states, or for data consistent only with boundary states, fewer observables can suffice. Collective measurements on several independently prepared copies of the environment state have probabilities linear in a tensor power of $\rho_E$, but their identifiability analysis must retain the product-state constraint. The criterion concerns exact calibration probabilities. Section~\ref{sec:calibration-loss} addresses finite budgets.

\subsection{When does the missing information change the decision?}
\label{sec:decision}

For exact calibration probabilities $\mathbf p$, let $\Lambda_{\mathbf p}$ be the set of models consistent with those data. 
The recovery choice can be determined even when \cref{prop:span} leaves a fidelity difference undetermined, because the choice depends only on which correction has the largest score at each record.
The following criterion states when a single decoder is optimal for every model in $\Lambda_{\mathbf p}$.

\begin{proposition}[A common optimal Pauli decoder]
\label{prop:common-optimum}
A single Pauli decoder attains $\FP{h}$ for every $h\in\Lambda_{\mathbf p}$ if and only if, for every record $y$,
\begin{equation}
 \bigcap_{h\in\Lambda_{\mathbf p}}
       \operatorname*{arg\,max}_{L\in\mathcal P_k}w_h(y,L)
       \ne\varnothing.
 \label{eq:common-optimum}
\end{equation}
Every common optimal randomized decoder is supported on this intersection at each record.
\end{proposition}

\begin{proof}
The loss is a sum of nonnegative terms,
\[
 \Floss{h}^{\mathrm P}(d)=\sum_{y,L}d(L\mid y)
              [\max_{L'}w_h(y,L')-w_h(y,L)].
\]
It vanishes for all models precisely when every action assigned positive probability maximizes the score of every model consistent with the calibration data at that record. Records with zero probability under every model impose no restriction.
\end{proof}

Determining all fidelity differences is sufficient for a common optimum but not necessary, because a difference can vary without changing sign. The decision changes when an unobserved perturbation passes through a tie: if two actions tie for the maximum and strictly exceed all other actions at a full-rank state, and an unobserved direction $Q_E$ changes their fidelity difference, then small opposite perturbations select different unique actions. Every common decoder consequently has a strictly positive loss on one of them. Appendix~\ref{app:framework-proofs} gives the quantitative bound. For the models consistent with the calibration data, we measure the cost of this ambiguity by
\begin{equation}
 \Delta F^{\mathrm P,\star}(\Lambda_{\mathbf p})
      \coloneqq \inf_d\sup_{h\in\Lambda_{\mathbf p}}\Floss{h}^{\mathrm P}(d).
 \label{eq:general-minimax}
\end{equation}
This criterion compares each model with its own optimum, rather than maximizing the worst-case absolute fidelity, one objective considered in robust recovery~\cite{kosut2008}. Appendix~\ref{app:framework-proofs} gives its dual form for finitely many models.

\subsection{From calibration distinguishability to fidelity loss}
\label{sec:calibration-loss}

The third question requires a quantitative bound. Consider two candidate models, labeled $h=0,1$. To include retained quantum calibration outputs, let the state $\sigma_h$ on $\mathcal H_{\rm cal}$ be the final accessible calibration state under model $h$, carried by $Q_{\rm cal}$ in \cref{fig:setting}. The space $\mathcal H_{\rm cal}$ includes any retained classical registers. Given $h$, the environment for the data run is prepared independently of all calibration runs. A common final CPTP processor $\mathcal G$ maps operators on $\mathcal H_{\rm cal}\otimes\mathcal H_Y\otimes\mathcal H_L$ to operators on $\mathcal H_L$. It may act jointly on $\sigma_h$ and the recorded logical data output. For a calibration state $\sigma$, we write $\Fid{h}(\mathcal G,\sigma)$ for the entanglement fidelity of the data run under model $h$ when $\mathcal G$ acts jointly on $\sigma$ and the recorded output of $\Phi_h$, and abbreviate $\Fid{h}(\mathcal G)\coloneqq \Fid{h}(\mathcal G,\sigma_h)$. We define the equal-prior conflict without calibration by
\begin{equation}
 \Gamma_{01}^{\mathrm{CPTP}}
   \coloneqq \frac{\FC{0}+\FC{1}}{2}
       -\sup_{\mathcal R}\frac{\Fid{0}(\mathcal R)+\Fid{1}(\mathcal R)}2,
 \label{eq:general-conflict}
\end{equation}
where $\mathcal R$ ranges over common full-record final CPTP maps.

\begin{proposition}[Calibration distinguishability bounds recovery]
\label{prop:calibration-loss}
Under these assumptions, every common joint final processor obeys
\begin{equation}
 \begin{aligned}
 &\max_{h=0,1}\{\FC{h}-\Fid{h}(\mathcal G,\sigma_h)\}\\
 &\qquad\ge\left[\Gamma_{01}^{\mathrm{CPTP}}
              -\tfrac12\Dtr(\sigma_0,\sigma_1)\right]_+,
 \end{aligned}
 \label{eq:program-bound}
\end{equation}
where $\Dtr(\sigma_0,\sigma_1)\coloneqq \|\sigma_0-\sigma_1\|_1/2$ and $[x]_+\coloneqq \max\{x,0\}$.
\end{proposition}

For a fixed data model, the entanglement fidelity after $\mathcal G$ is the probability of an effect on the calibration state, so replacing $\sigma_1$ by $\sigma_0$ changes it by at most their trace distance~\cite[Theorem~3.4]{watrous2018}. With the same state inserted under both models, $\mathcal G$ becomes an admissible common recovery, and averaging the two scores gives Eq.~\eqref{eq:program-bound}. Appendix~\ref{app:framework-proofs} contains the complete argument.

Consequently, if every protocol with interval budget $M$ produces calibration states within trace distance $\delta(M)$, the loss is at least $[\Gamma_{01}^{\mathrm{CPTP}}-\delta(M)/2]_+$. A conflict creates the need for information, and the calibration signal sets its cost. When calibration is measured only to select a Pauli decoder, a multiplicative bound relative to the Pauli benchmark also holds (Appendix~\ref{app:framework-proofs}). Equation~\eqref{eq:program-bound} is the version that covers joint quantum processing.

\subsection{A two-round Shor example}
\label{sec:shor-example}

The following small example shows how the first two questions guide the design of a calibration experiment. It complements the long-storage, order-one gap of Sec.~\ref{sec:family}. We use the Shor $[[9,1,3]]$ code, with encoding isometry $V$, and let noise act only on qubits $1,2,3$ of its first repetition block. The physical normalizer $P_{\rm sh}\coloneqq X_1X_2X_3$ induces a nontrivial logical Pauli $L_{\rm sh}\coloneqq V^\dagger P_{\rm sh}V$ (logical $Z$ in the conventional Shor basis). Two environment qubits, with $\mathcal H_E=\mathbb C^2\otimes\mathbb C^2$, interact with the code through
\begin{equation}
 \begin{gathered}
 U_\nu\coloneqq \prod_{j=1}^3e^{-i\vartheta_jX_j\otimes Z_{E_\nu}},\qquad \nu=1,2,\\
 Z_{12}\coloneqq Z_{E_1}Z_{E_2}.
 \end{gathered}
 \label{eq:shor-interactions}
\end{equation}
After each interval we extract the syndrome and apply the standard single-error correction table. We label the four possible outcomes by the error they correct:
\begin{center}
\begin{tabular}{c|cccc}
$(Z_1Z_2,Z_2Z_3)$ & $(+,+)$ & $(-,+)$ & $(-,-)$ & $(+,-)$\\
\hline
$s$ & $0$ & $1$ & $2$ & $3$\\
$R_s$ & $I$ & $X_1$ & $X_2$ & $X_3$
\end{tabular}
\end{center}
The other checks have trivial syndrome. After removal of an environment-only phase that cancels from all effects, each branch takes the form $u_sI-i v_sL_{\rm sh}\otimes Z_{E_\nu}$, with real $u_s,v_s$ and $q_s=u_s^2+v_s^2$. The syndrome effects of complete QEC cycles and all effects determining the SPP are therefore proportional to $I_E$, so these calibration data span only $I_E$, whatever the initial environment state.

The two-round Pauli scores, by contrast, depend on $\kappa\coloneqq \Tr(\rho_EZ_{12})$. For the record $y=(s_1,s_2)$, direct multiplication gives
\begin{align}
 \Delta E_{s_1s_2;I,L_{\rm sh}}
  ={}&(u_{s_1}^2-v_{s_1}^2)(u_{s_2}^2-v_{s_2}^2)I_E\notag\\
     &-4u_{s_1}v_{s_1}u_{s_2}v_{s_2}Z_{12}.
 \label{eq:shor-all-contrasts}
\end{align}
The histories $(I,I)$ and $(L_{\rm sh},L_{\rm sh})$ have the same final logical action $I$, while $(I,L_{\rm sh})$ and $(L_{\rm sh},I)$ both produce $L_{\rm sh}$. Their amplitudes interfere within each final logical class. The correlation-dependent contributions cancel in the sum of the two recovery weights, leaving the syndrome probability unchanged, but add in their difference. Appendix~\ref{app:shor-history-coherence} expresses this mechanism in the logical process tensor.

Every difference operator for Pauli recoveries thus lies in $\operatorname{span}_{\mathbb R} \{I_E,Z_{12}\}$. Choose $\tan\vartheta_2=\tan\vartheta_3=\zeta$ and $\tan\vartheta_1=\sqrt3\zeta^2$ with $0<\zeta\le1/5$. This range ensures that every record other than $(1,1)$ prefers $I$ (Appendix~\ref{app:shor-example}). Then $u_1^2=3v_1^2$, and the exceptional record $(1,1)$ has fidelity difference
\begin{equation}
 \begin{gathered}
 w_{\rho_E}(11,I)-w_{\rho_E}(11,L_{\rm sh})
       =\frac{q_1^2}{4}(1-3\kappa),\\
 q_1=\frac{4\zeta^4}{(1+3\zeta^4)(1+\zeta^2)^2}.
 \end{gathered}
 \label{eq:shor-crossing}
\end{equation}
The optimal correction on this record therefore switches at $\kappa=1/3$, a value that syndrome data and the SPP cannot reveal. In the family $\rho_\kappa\coloneqq (I_E+\kappa Z_{12})/4$ of unbiased signs, the fixed-sign and resampled models are $\kappa=1$ and $\kappa=0$. Between these two models, the exact full-record minimax fidelity loss is $q_1^2/6=O(\zeta^8)$, including arbitrary final CPTP maps (Appendix~\ref{app:shor-example}). Because the conflict occurs only on a rare record, the example illustrates calibration design rather than a large loss.

Now consider a separate calibration run in which extraction is delayed until after both intervals. Opposite signs then cancel, and equal signs double the rotation angles. If $p_*>0$ is the probability of a nonzero syndrome at doubled angles, the nonzero-syndrome event has environment effect
\begin{equation}
 \begin{aligned}
 C_{\rm nz}&=\frac{p_*}{2}(I_E+Z_{12}),\\
 p_*&=1-\prod_{j=1}^3\cos^2(2\vartheta_j)
         -\prod_{j=1}^3\sin^2(2\vartheta_j).
 \end{aligned}
 \label{eq:shor-calibration}
\end{equation}
The calibration span is now $\operatorname{span}_{\mathbb R}\{I_E,Z_{12}\}$, so Proposition~\ref{prop:span} guarantees that this single additional rate determines every two-round Pauli recovery choice without reconstructing the environment state. The data circuit itself still extracts after every interval. In this way the framework identifies an undetermined fidelity difference and a calibration experiment that measures it.
\section{Temporal correlations and incompatible recoveries}
\label{sec:family}

Section~\ref{sec:general} showed that missing calibration information costs fidelity when it leaves undetermined a fidelity difference that changes the optimal correction. We now construct two temporal noise models with identical syndrome statistics and SPPs whose optimal final corrections differ, so that at suitable storage times the worst-case fidelity loss approaches $1/2$, even with the full syndrome record.

\subsection{Fixed and independently resampled sign models}
During one physical noise interval, every qubit undergoes the same coherent rotation. The resulting unitary $U_\xi(\theta)$ is
\begin{equation}
 U_\xi(\theta)\coloneqq \bigotimes_{j=1}^n e^{-i\xi\theta X_j},
 \qquad \xi\in\{+1,-1\}.
 \label{eq:physical}
\end{equation}
Model $\mathsf F$ draws one fair sign at the beginning of a run and keeps it. Model $\mathsf R$ draws an independent fair sign in every noise interval, whether or not a syndrome is extracted. The parameter $\theta>0$ is half the Bloch-sphere rotation angle.

Physically, both models describe a spatially uniform stray field that rotates every qubit by the same angle with a random sign. Model $\mathsf F$ idealizes quasi-static noise, whose correlation time exceeds the storage time~\cite{pataki2024}, and model $\mathsf R$ idealizes noise that decorrelates between successive intervals. In the terminology of multi-time quantum processes~\cite{pollock2018}, model $\mathsf R$ is Markovian, because its process is a product of identical single-interval channels, each an equal mixture of conjugation by $U_+$ and $U_-$. Model $\mathsf F$ is non-Markovian. Its fixed sign acts as a classical memory that correlates all intervals, and its process is an equal mixture of two product processes, which is not itself a product.

Syndrome measurement can suppress logical coherence without eliminating it~\cite{beale2018}. For a fixed sign, coherent physical rotations across the code block can leave a residual logical rotation through interference within a syndrome sector~\cite{huang2019}, and this rotation can accumulate over repeated QEC cycles~\cite[Sec.~3]{iverson2020}. The $\mathsf F/\mathsf R$ comparison isolates the role of temporal correlations in this mechanism by keeping the single-interval averaged channel and spatial correlations fixed.

\begin{definition}[Code and fixed $X$-correction lookup table]
\label{def:family}
The CSS code $\C$ encodes one qubit in the common $+1$ eigenspace of positive-phase stabilizer generators. Its $X$ stabilizers have even weight and its $X$ distance is $d_X=2r+1\ge3$. A fixed $X$-correction lookup table assigns an $X$ correction $R_s$ to every reachable syndrome $s$ and corrects all physical $X$ errors of weight at most $r$. Write $\Pi_s\in\mathcal B(\mathcal H_{\rm phys})$ for the corresponding syndrome projector. Subscripts $\C$ on constants refer to the code together with this lookup table.
\end{definition}
The specified stabilizer eigenvalues fix the relative signs of interfering physical error amplitudes; changing the encoding eigenspace can change the logical channel~\cite{rajmohan2026}. The Shor, Steane, and odd-distance rotated surface codes satisfy these conditions. The code and lookup table stay fixed as $\theta\to0$. In the data run, ideal syndrome extraction and lookup-table correction follow every interval, the logical schedule between cycles is the identity, and the entire syndrome history is available for final recovery. Section~\ref{sec:frame-tracking} explains the equivalence of correction frames at the same storage time. Unless stated otherwise, these conventions apply throughout Secs.~\ref{sec:family}--\ref{sec:protection}. Subscripts on $O$, $o$, and $\Theta$ list the fixed quantities on which the implied constants depend.

The parity conditions of \cref{def:family} make $Z^{\otimes n}$ a physical representative of logical $\Zbar$: it commutes with all stabilizers and anticommutes with the odd-weight logical $\Xbar$. Hence
\begin{equation}
 Z^{\otimes n}\simeq\Zbar,\qquad
 Z^{\otimes n}U_\xi Z^{\otimes n}=U_{-\xi}=U_\xi^*,
 \label{eq:global-z-conjugation}
\end{equation}
where complex conjugation is taken in the computational basis. This parity structure makes the syndromes independent of the sign (\cref{sec:equivalence}) and underlies the echo and randomization protocols of \cref{sec:protection}.

The two models fit the common-interaction framework of \cref{sec:general}. For a given storage time $T$, take one hidden sign register per interval, so that $\mathcal H_E=(\mathbb C^2)^{\otimes T}$ with $Z|\xi\rangle=\xi|\xi\rangle$ for $\xi=\pm1$, and the initial environment states
\begin{align}
 \rho_E^{\mathsf F}&\coloneqq \frac12\sum_{\xi=\pm1}
       (|\xi\rangle\langle\xi|)^{\otimes T},\notag\\
 \rho_E^{\mathsf R}&\coloneqq (I_2/2)^{\otimes T}.
 \label{eq:classical-environment}
\end{align}
In interval $j$, the known interaction reads register $j$ and applies \cref{eq:physical}. Each fresh run uses a new block of registers, drawn independently given the model. Neither the noise-aware decoder nor the common decoder observes the realized signs.

\subsection{Exact equality of syndrome statistics and SPPs}
\label{sec:equivalence}
Each one-cycle corrected branch has a simple form. Up to an irrelevant scalar phase, the branch $K_s(\xi)\coloneqq V^\dagger R_s\Pi_sU_\xi(\theta)V$ for syndrome $s$ is
\begin{equation}
 K_s(\xi)=u_sI-i\xi v_s\Xbar,
 \qquad K_s(\xi)^\dagger K_s(\xi)=q_sI,
 \label{eq:branch-form}
\end{equation}
where $u_s,v_s\in\mathbb R$, $q_s\coloneqq u_s^2+v_s^2$, and $\sum_sq_s=1$. These coefficients do not depend on $\xi$. To derive \cref{eq:branch-form}, divide the physical $X$ errors with syndrome $s$ into those that the correction maps to logical $I$ and those it maps to $\Xbar$. The two groups have opposite weight parity, so, up to a common phase, the amplitudes of the first group are real and those of the second are imaginary. Positive-phase stabilizers introduce no relative signs within a group. The cross terms in $K_s^\dagger K_s$ therefore cancel, and the effect is a scalar. Conditional logical unitarity and input-independent syndrome probabilities have also been established using time-reversal symmetry under suitable code conditions~\cite[Theorem~2 and Corollary~2.1]{cheng2025}; see also Ref.~\cite{behrends_beri_2025_coherence}. Here the coset calculation also establishes independence from the noise sign. \Cref{app:branch} gives the coset formulas.

\begin{proposition}[Identical observed statistics]\label{thm:equivalence}
For every sign trajectory $\boldsymbol\xi=(\xi_1,\ldots,\xi_T)$ and fixed number of cycles,
\begin{equation}
 \Pr(s_1,\ldots,s_T\mid\boldsymbol\xi)=\prod_{j=1}^Tq_{s_j}.
 \label{eq:syndrome}
\end{equation}
This identity remains valid with syndrome-adaptive trace-preserving logical controls and a private quantum register. Models $\mathsf F$ and $\mathsf R$ also have identical SPPs, obtained by Pauli-twirling each noise interval.
\end{proposition}
The first statement follows by applying the scalar branch effects one cycle at a time, including to states entangled with the private register. For the second, twirling gives each physical $X$ error $P$ the probability $(\sin^2\theta)^{|P|}(\cos^2\theta)^{n-|P|}$, which does not depend on the sign, so every Pauli-history probability is the same under both models. Both SPPs are therefore the i.i.d. process of a single twirled interval, although $\mathsf F$ is non-Markovian. The twirl preserves correlations between Pauli errors~\cite{kam2026}, but the memory of $\mathsf F$ resides in the sign of a coherent rotation, which affects no Pauli error probability. The equality holds for the complete multi-time distributions at every angle. Syndromes from complete QEC cycles are thus \emph{blind} to the temporal model: their distribution is the same under $\mathsf F$ and $\mathsf R$. Logical measurements can reveal the model, and Sec.~\ref{sec:calibration} uses them for calibration. When such measurements are inserted, the syndrome distribution summed over their outcomes still obeys \cref{eq:syndrome}.

Physically, each syndrome branch in \cref{eq:branch-form} acts on the logical qubit as $\sqrt{q_s}$ times a unitary rotation about $\Xbar$. The weight $q_s$ depends neither on the sign nor on the logical state, and reversing the sign reverses the syndrome-dependent rotation angle. Rotations accumulated in earlier cycles therefore cannot feed back into later syndrome probabilities. The sign correlations can affect the logical output without appearing in the syndrome probabilities. This blindness relies on extracting after every interval. If one extraction is omitted, the physical rotations of two neighboring intervals compose before projection, doubling for equal signs and cancelling for opposite signs, so the sign correlation becomes visible in the syndrome statistics [\cref{eq:correlation-query}].

\paragraph{Rotation magnitudes and signs.} The same equality of syndrome statistics and SPPs holds for $U_\xi=\bigotimes_j e^{-i\xi\theta_jX_j}$ with arbitrary real angles $\theta_j$, provided the angle profile is the same in both models and the sign is common to all qubits. Each physical $X$ string with support $e$ has a coefficient $(-i\xi)^{|e|}$ times a real number, so the coset parity argument is unchanged. More generally, let nonnegative magnitudes vary across qubits and intervals according to an external classical process independent of the logical input. Conditional on the magnitude trajectory, the same scalar-effect argument applies interval by interval. Models with the same joint distribution of magnitudes therefore have identical syndrome statistics and SPPs, regardless of their common-sign correlations. Magnitude correlations, by contrast, generally affect syndrome statistics; in noise models that tie them to sign correlations, syndrome data may constrain the latter indirectly. The recovery-gap and calibration-cost results below use the uniform fixed angle of \cref{eq:physical}.

\paragraph{Faults that commute with the noise.} Blindness persists with stochastic $X$ faults and incorrect syndrome reports, even when a wrong report selects the wrong physical correction. The classical fault and reporting kernels must be common to the two models and independent of the signs and logical state, with any initial fault memory independent of the source. Boundary controls must preserve syndrome sectors, and private logical measurement outcomes cannot enter these kernels. Under the precise assumptions of \cref{app:reporting-faults}, tracking the residual syndrome sector keeps every branch effect scalar. Thus the reported history and final residual syndrome have the same distribution under both models, including when the fault process has memory.

The fidelity conclusion needs stronger assumptions. For the independent physical-fault and i.i.d.-report-mask model of \cref{prop:reporting-faults}, take distance three, a physical $X$-fault probability $p_X=\kappa_X\theta^2$ per qubit and round, and a probability $p_{\rm rep}=\kappa_{\rm rep}\theta^2$ per round that the reported syndrome is wrong, with fixed $\kappa_X,\kappa_{\rm rep}\ge0$, and a trusted final cleanup round. Then the near-half fidelity loss of \cref{thm:regret} persists at the storage time defined from the fault-averaged logical phase.

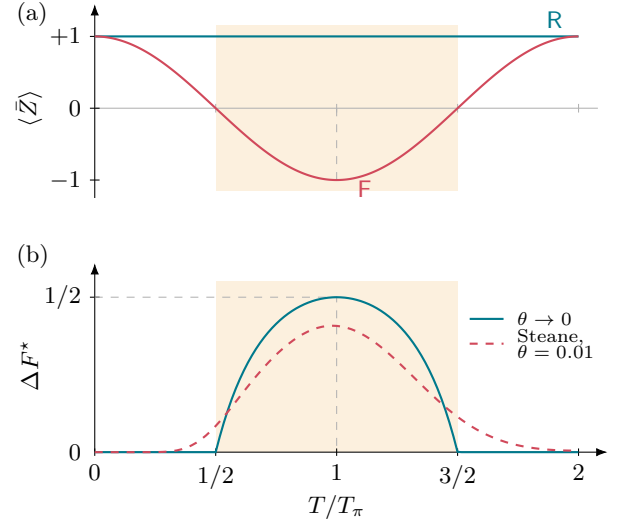
\begin{figure}[tbp]
\centering
\begin{tikzpicture}[x=3.2cm]
 % (a) mean logical readout
 \begin{scope}[y=0.95cm,yshift=4.55cm]
  \fill[qecsaffron!18] (0.5,-1.15) rectangle (1.5,1.15);
  \draw[-{Latex[length=1.5mm]}] (0,-1.25)--(0,1.45);
  \draw[gray!60] (0,0)--(2.08,0);
  \foreach \yy/\lab in {-1/-1,0/0,1/+1}{
   \draw (-0.02,\yy)--(0.02,\yy);
   \node[anchor=east,font=\small] at (-0.02,\yy) {$\lab$};
  }
  \foreach \xx in {0.5,1,1.5,2}{\draw[gray!60] (\xx,-0.06)--(\xx,0.06);}
  \draw[gray!60,dashed] (1,-1)--(1,0);
  \node[rotate=90,font=\small] at (-0.27,0) {$\langle\Zbar\rangle$};
  \draw[qecpetrol,thick] (0,1)--(2,1);
  \node[anchor=south,text=qecpetrol,font=\small] at (1.9,1.0) {$\mathsf R$};
  \draw[qecraspberry,thick,domain=0:2,samples=121,smooth]
   plot (\x,{cos(180*\x)});
  \node[anchor=north west,text=qecraspberry,font=\small] at (1.05,-0.88) {$\mathsf F$};
  \node[anchor=north west,font=\small] at (-0.36,1.6) {(a)};
 \end{scope}
 % (b) full-record minimax loss
 \begin{scope}[y=4.1cm]
  \fill[qecsaffron!18] (0.5,0) rectangle (1.5,0.55);
  \draw[-{Latex[length=1.5mm]}] (0,0)--(2.12,0);
  \draw[-{Latex[length=1.5mm]}] (0,0)--(0,0.61);
  \foreach \xx/\lab in {0/0,0.5/{1/2},1/1,1.5/{3/2},2/2}{
   \draw (\xx,0.012)--(\xx,-0.012);
   \node[below,font=\small] at (\xx,-0.012) {$\lab$};
  }
  \draw (-0.02,0.5)--(0.02,0.5);
  \node[left,font=\small] at (-0.02,0.5) {$1/2$};
  \node[left,font=\small] at (-0.02,0) {$0$};
  \draw[gray!60,dashed] (0,0.5)--(1,0.5)--(1,0);
  \draw[qecpetrol,thick,domain=0:2,samples=121]
   plot (\x,{max(0,-cos(180*\x))/(1+max(0,-cos(180*\x)))});
  \draw[qecraspberry,thick,dashed] plot coordinates {%
   (0.00,0.000000) (0.02,0.000000) (0.04,0.000000) (0.06,0.000000)
   (0.08,0.000000) (0.10,0.000000) (0.12,0.000000) (0.14,0.000000)
   (0.16,0.000000) (0.18,0.000003) (0.20,0.000013) (0.22,0.000046)
   (0.24,0.000146) (0.26,0.000364) (0.28,0.000867) (0.30,0.001741)
   (0.32,0.003361) (0.34,0.005788) (0.36,0.009607) (0.38,0.014734)
   (0.40,0.021773) (0.42,0.030506) (0.44,0.041336) (0.46,0.053976)
   (0.48,0.068474) (0.50,0.084597) (0.52,0.102022) (0.54,0.120648)
   (0.56,0.139871) (0.58,0.159777) (0.60,0.179623) (0.62,0.199631)
   (0.64,0.219070) (0.66,0.238241) (0.68,0.256502) (0.70,0.274299)
   (0.72,0.291117) (0.74,0.307084) (0.76,0.321964) (0.78,0.335846)
   (0.80,0.348758) (0.82,0.360576) (0.84,0.371202) (0.86,0.380600)
   (0.88,0.388707) (0.90,0.395459) (0.92,0.400791) (0.94,0.404646)
   (0.96,0.407102) (0.98,0.408031) (1.00,0.407402) (1.02,0.405384)
   (1.04,0.401843) (1.06,0.396915) (1.08,0.390625) (1.10,0.383009)
   (1.12,0.374146) (1.14,0.364184) (1.16,0.353165) (1.18,0.341233)
   (1.20,0.328406) (1.22,0.314925) (1.24,0.300763) (1.26,0.286224)
   (1.28,0.271218) (1.30,0.256078) (1.32,0.240704) (1.34,0.225415)
   (1.36,0.210122) (1.38,0.195105) (1.40,0.180321) (1.42,0.165950)
   (1.44,0.152030) (1.46,0.138692) (1.48,0.125998) (1.50,0.113859)
   (1.52,0.102516) (1.54,0.091734) (1.56,0.081857) (1.58,0.072587)
   (1.60,0.064256) (1.62,0.056391) (1.64,0.049472) (1.66,0.043054)
   (1.68,0.037475) (1.70,0.032314) (1.72,0.027874) (1.74,0.023876)
   (1.76,0.020433) (1.78,0.017377) (1.80,0.014793) (1.82,0.012556)
   (1.84,0.010658) (1.86,0.009100) (1.88,0.007801) (1.90,0.006790)
   (1.92,0.006004) (1.94,0.005472) (1.96,0.005143) (1.98,0.005057)
   (2.00,0.005163)};
  \draw[qecpetrol,thick] (1.55,0.43)--(1.69,0.43);
  \node[anchor=west,font=\scriptsize] at (1.71,0.43) {$\theta\to0$};
  \draw[qecraspberry,thick,dashed] (1.55,0.35)--(1.69,0.35);
  \node[anchor=west,font=\scriptsize,align=left]
   at (1.71,0.35) {Steane,\\[-2pt]$\theta=0.01$};
  \node[below,font=\small] at (1,-0.11) {$T/T_\pi$};
  \node[rotate=90,font=\small] at (-0.27,0.28) {$\Fmm$};
  \node[anchor=north west,font=\small] at (-0.36,0.70) {(b)};
 \end{scope}
\end{tikzpicture}
\caption{\label{fig:separation}Mean logical readout and minimax fidelity loss versus storage time, where $T_\pi$ [\cref{eq:horizon}] is the time at which the fixed sign accumulates a logical $\pi$ rotation. (a) Mean readout of a $\Zbar=+1$ input before final recovery, in the fixed-code weak-noise limit at fixed $T/T_\pi$. It approaches $\cos(\pi T/T_\pi)$ under $\mathsf F$ and one under $\mathsf R$, although the syndrome statistics and SPPs coincide. (b) Full-record minimax fidelity loss in the same limit (solid) and for the Steane code at $\theta=0.01$ (dashed). In the shaded interval, the $\mathsf F$ readout is negative and the limiting optimal corrections differ. At $T/T_\pi=3/4$ the finite-angle loss is $0.3146$, against the limit $0.4142$. At $T/T_\pi=1/2$ it is about $56\%$ of the leading term in \cref{eq:half-horizon-tightness} ($69\%$ for Shor), which becomes accurate only at much smaller angles.}
\end{figure}

\subsection{Distance-order visibility and coherent accumulation}
\label{sec:visibility}
Define
\begin{equation}
 A\coloneqq \sum_s(u_s^2-v_s^2),\qquad B\coloneqq 2\sum_su_sv_s.
 \label{eq:AB}
\end{equation}
For a given sign, averaging over the syndrome leaves the logical $\Xbar$ component unchanged and multiplies the two perpendicular Bloch components, combined into one complex coordinate, by $A-i\xi B$. Retaining the syndrome gives the one-cycle instrument $\mathcal E_\xi(\rho)\coloneqq \sum_s|s\rangle\langle s|\otimes K_s(\xi)\rho K_s(\xi)^\dagger$. Its sensitivity to the sign is $\beta_{\rm ret}\coloneqq \|\mathcal E_+-\mathcal E_-\|_\diamond/2$.

For every code and lookup table of \cref{def:family}, counting pairs of errors with the same syndrome gives
\begin{align}
 A&=1-a_{\C}\theta^{d_X+1}+O_{\C}(\theta^{d_X+3}),\notag\\
 B&=b_{\C}\theta^{d_X}+O_{\C}(\theta^{d_X+2}),\notag\\
 |b_{\C}|&=2N_{d_X}\binom{2r}{r},\qquad a_{\C}>0,
 \label{eq:orders}\\
 \beta_{\rm ret}&=2N_{d_X}4^r\theta^{d_X}
                  +O_{\C}(\theta^{d_X+2}).
 \label{eq:retained-orders-main}
\end{align}
Here $N_{d_X}$ counts the minimum-weight supports of the nontrivial logical $X$ class. The order-$d_X$ terms pair two errors with the same syndrome whose product is a minimum-weight logical operator. Splitting the support of that operator into a correctable subset and its complement gives the signed coefficient $\binom{2r}{r}$ in $B$ and the absolute coefficient $4^r$ in $\beta_{\rm ret}$. The two coefficients differ because a syndrome whose correctable error has weight $j$ contributes to the logical rotation with sign $(-1)^{r-j}\xi$ at this order. Contributions of different syndromes therefore partly cancel in the syndrome average $B$, whereas retaining the syndrome keeps them apart. The cancellation is incomplete, so $B$ has a nonzero leading term of order $\theta^{d_X}$. Let $m_s$ be the number of weight-$(r+1)$ errors with syndrome $s$ that the table corrects to $\Xbar$. Then $a_{\C}=2\sum_sm_s^2$. \Cref{app:branch} gives the counting proof, including this noncancellation and the norm of the retained output.

For a $\Zbar=+1$ input, the mean readouts before final correction are
\begin{equation}
 \langle\Zbar\rangle_{\mathsf R}(T)=A^T,\qquad
 \langle\Zbar\rangle_{\mathsf F}(T)=\operatorname{Re}[(A-iB)^T].
 \label{eq:readout}
\end{equation}
Let $\alpha\coloneqq \arg(A+iB)$ and, for sufficiently small $\theta$, define
\begin{equation}
 T_\pi\coloneqq \operatorname{round}\frac{\pi}{|\alpha|}
      =\Theta_{\C}(\theta^{-d_X}).
 \label{eq:horizon}
\end{equation}
By this time the mean accumulated logical angle under the fixed sign has magnitude approximately $\pi$, whereas resampling averages the drift away. Both models lose only $O_{\C}(\theta)$ in amplitude, owing to cycle-to-cycle fluctuations of the logical rotation angle (\cref{sec:regret}). \Cref{fig:separation}(a) shows this contrast. Comparisons across correction tables use this same storage time (Sec.~\ref{sec:frame-tracking}).

\subsection{A nearly maximal loss with the full syndrome record}
\label{sec:regret}
For $q_s>0$, write \cref{eq:branch-form} as $K_s(\xi)=\sqrt{q_s}\,e^{-i\xi\phi_s\Xbar/2}$, so that branch $s$ rotates the logical qubit about $\Xbar$ by the angle $\xi\phi_s$. Equivalently, $\phi_s=2\arg(u_s+iv_s)$, defined modulo $2\pi$ because \cref{eq:branch-form} holds up to a scalar phase. For a history $y=(s_1,\ldots,s_T)$, define its accumulated phase $\phi_{\rm tot}(y)\coloneqq\sum_{j=1}^T\phi_{s_j}$. Each history has the same probability $q_y=\prod_jq_{s_j}$ under both models, and its conditional logical multipliers are
\begin{equation}
 g_{\mathsf F}(y)=\cos\phi_{\rm tot}(y),\qquad
 g_{\mathsf R}(y)=\prod_j\cos\phi_{s_j}.
 \label{eq:conditional}
\end{equation}
Under $\mathsf F$ the signed branch angles determined by the record add before averaging over the common sign, whereas under $\mathsf R$ the hidden signs are independent and the cosines multiply. Each conditional channel is the real mixture $(1+g_h(y))\id/2+(1-g_h(y))\mathcal X/2$, with $\mathcal X(\rho)\coloneqq \Xbar\rho\Xbar$. For such channels, a randomized choice between $I$ and $\Xbar$ performs at least as well as any final CPTP recovery, simultaneously for both models: the fidelities of any recovery on the two Pauli components sum to at most one, and a suitably randomized Pauli raises both by the same amount. \Cref{app:regret} proves this reduction for the full record.

Let $\Fmm(T)$ denote the smallest worst-case fidelity loss, relative to each model's full-record CPTP optimum, that a single recovery can guarantee without knowing the model. The noise-aware optimum knows the model but not the realized signs.

\begin{theorem}[Recovery separation and storage time]
\label{thm:regret}\label{cor:scaled-horizon} At the storage time $T_\pi$ defined above,
\begin{gather}
 \frac12-O_{\C}(\theta)\le\Fmm(T_\pi)\le\frac12,
 \notag\\
 \FC{h}(T_\pi)=\FP{h}(T_\pi)\ge1-O_{\C}(\theta).
 \label{eq:regret}
\end{gather}
More generally, let $T_\theta(x)\coloneqq \operatorname{round}(xT_\pi)$ and $\chi(x)\coloneqq [-\cos(\pi x)]_+$. For every fixed finite $x_{\max}>0$,
\begin{equation}
 \sup_{0\le x\le x_{\max}}
 \left|\Fmm(T_\theta(x))-
              \frac{\chi(x)}{1+\chi(x)}\right|
 =O_{\C,x_{\max}}(\sqrt\theta).
 \label{eq:horizon-curve}
\end{equation}
Both statements permit arbitrary full-record final CPTP recoveries, including ancillas. Randomized choices between $I$ and $\Xbar$ attain the minimax value, and deterministic record-dependent choices attain each noise-aware optimum.
\end{theorem}

For the bound at $T_\pi$, the fixed corrections $I$ under $\mathsf R$ and $\Xbar$ under $\mathsf F$ have infidelities
\begin{align}
 \RefInf{\mathsf R}&=\frac{1-\langle\Zbar\rangle_{\mathsf R}(T_\pi)}2,
 &\RefInf{\mathsf F}&=\frac{1+\langle\Zbar\rangle_{\mathsf F}(T_\pi)}2,
 \label{eq:defects}\\
 \RefInf{h}&=\frac{\pi a_{\C}}{2|b_{\C}|}\theta
                  +O_{\C}(\theta^2).
 \label{eq:defect-order}
\end{align}
Because $\mathbb E|g_{\mathsf F}+g_{\mathsf R}| \le2(\RefInf{\mathsf F}+\RefInf{\mathsf R})$, no single full-record recovery can approach both noise-aware optima. At general storage times, the accumulated branch phase concentrates: the conditional multiplier approaches $\cos(\pi x)$ under $\mathsf F$ and one under $\mathsf R$, and optimizing over these two limiting channels gives $\chi/(1+\chi)$ [\cref{fig:separation}(b)]. \Cref{app:regret} proves both statements. For each fixed $x\ge0$, the limiting optimal corrections therefore differ whenever $\cos(\pi x)<0$. The $\mathsf F$ optimum tends to $(1+|\cos\pi x|)/2$, so $\mathsf F$ is almost perfectly correctable only near integer $x$, such as $x=1$. The noise-aware recovery knows the temporal model but not the realized sign. Near $T_\pi$, rotations by $+\pi$ and $-\pi$ about $\Xbar$ coincide up to a global phase. Near $T_\pi/2$, the two signs give rotations by $\pm\pi/2$, which no single recovery undoes, so even the noise-aware $\mathsf F$ fidelity approaches $1/2$. For the Steane code at $\theta=0.01$, it equals $0.643$ at $T_\theta(1/2)$.

\paragraph{Drift and diffusion of the logical angle.} Call a lookup table \emph{minimum-weight} if, for every syndrome, the class it corrects to logical $I$ contains an error of minimum weight. For such tables every branch angle $\phi_s$ is $O_{\C}(\theta)$ (\cref{lem:minimum-weight-small-phase}), and $\phi_{\rm tot}(y)$ performs a biased random walk. The largest angles occur on rare syndromes, while the leading contributions to the mean have alternating signs according to the correctable-error weight (\cref{sec:visibility}). The mean angle per cycle, equal to $B$ at leading order, is therefore much smaller than the per-cycle fluctuation, whose variance is $\Theta_{\C}(\theta^{d_X+1})$. Over $T_\pi$ the mean accumulated angle has magnitude approximately $\pi$, while the standard deviation of $\phi_{\rm tot}$ is $O_{\C}(\sqrt\theta)$. For the Steane code, the leading contributions to the mean angle per cycle are $+42\theta^3$ from the nonzero syndromes and $-14\theta^3$ from the zero syndrome. At $\theta=0.01$, the mean is about $2\%$ of the root-mean-square angle, and the standard deviation at $T_\pi$ is about $0.53$ radians. This spread sets the $O_{\C}(\theta)$ infidelity of the noise-aware $\mathsf F$ recovery, which knows $\phi_{\rm tot}(y)$ but not the sign. Under $\mathsf R$ independent signs multiply the same branch angles and dephase the logical qubit by a comparable amount.

\paragraph{The square-root rate is sharp.} For minimum-weight tables, the square-root order in \cref{eq:horizon-curve} is sharp, as shown by the loss at $T_\theta(1/2)$ (\cref{prop:half-horizon-tightness}). At this time the limiting $\mathsf F$ multiplier vanishes, so phase fluctuations enter the loss at first order rather than at second order as at $T_\pi$. Convergence at small angles is slow, as the dashed curve in \cref{fig:separation}(b) shows.

\begin{table}[tbp]
\caption{\label{tab:finite-angle}Finite-angle comparison with minimum-weight $X$-correction lookup tables, at the storage time $T_\pi$ of the original data circuit. Original circuit: optimal fidelities of the two models over full-record final CPTP recoveries, and the minimax loss $\Fmm$. Randomized circuit: independent physical Pauli randomization with its record discarded, followed by the same lookup table and the final identity correction. Its entanglement fidelity $F_{\rm tw}$ [\cref{eq:twirl-fidelity}] is the same under $\mathsf F$ and $\mathsf R$.}
\small
\begin{ruledtabular}
\setlength{\tabcolsep}{4.5pt}
\begin{tabular}{lrcccc}
 & & \multicolumn{3}{c}{Original circuit} & Randomized\\
\cmidrule(lr){3-5}\cmidrule(l){6-6}
$\theta$ & $T_\pi$ & $\FC{\mathsf F}$ & $\FC{\mathsf R}$ & $\Fmm$ & $F_{\rm tw}$\\
\hline
\multicolumn{6}{l}{\rule{0pt}{2.6ex}Shor $[[9,1,3]]$}\\
0.10 & 267 & 0.8510 & 0.8073 & 0.2593 & 0.8111\\
0.05 & 2\,105 & 0.9003 & 0.8943 & 0.3558 & 0.8949\\
0.01 & 261\,852 & 0.9770 & 0.9770 & 0.4673 & 0.9770\\
\multicolumn{6}{l}{\rule{0pt}{2.9ex}Steane $[[7,1,3]]$}\\
0.10 & 128 & 0.8231 & 0.6111 & 0.0848 & 0.8000\\
0.05 & 929 & 0.8297 & 0.7438 & 0.1956 & 0.8931\\
0.01 & 112\,357 & 0.9346 & 0.9340 & 0.4074 & 0.9770\\
\end{tabular}
\end{ruledtabular}
\end{table}

\Cref{tab:finite-angle} lists finite-angle fidelities and full-record minimax losses. At $\theta=0.05$, the twirled physical error probability is $\sin^2\theta\simeq2.5\times10^{-3}$ per qubit per interval, and $T_\pi$ is 929 QEC cycles for Steane and 2105 for Shor. The last column refers to the physically randomized data circuit analyzed in \cref{sec:protection}. \Cref{app:numerics} describes the numerical methods.

\subsection{A process requiring quantum memory}
\label{sec:quantum-extension}
Classical signs already produce the recovery conflict, and the conflict persists when the physical noise has quantum memory. Replace \cref{eq:physical} by the system--environment unitary $U(\theta,\omega)\in\mathcal B(\mathcal H_{\rm phys}\otimes\mathbb C^2)$, in which coupling and precession act simultaneously,
\begin{equation}
 U(\theta,\omega)\coloneqq \exp\!\left[-i\left(
 \theta\sum_{j=1}^nX_j\otimes Z_E+\omega I\otimes X_E\right)\right].
 \label{eq:quantum-unitary}
\end{equation}
The environment qubit starts in $I_E/2$. Model $\mathrm{Ret}_\omega$ retains it between intervals, while $\mathrm{Res}_\omega$ resets it to $I_E/2$ at every interval boundary. Both use the same interaction. At $\omega=0$ they reduce to $\mathrm{Ret}_0=\mathsf F$ and $\mathrm{Res}_0=\mathsf R$. We call the noise channel of one interval between neighboring intervention slots a \emph{link}.

\paragraph{Identical observations despite quantum memory.} At any fixed $\theta,\omega$, these two models have identical SPPs and syndrome distributions under the controls of \cref{thm:equivalence}. If $\theta>0$, $\omega\ne0$, and $\sqrt{n^2\theta^2+\omega^2}<\pi/2$, the retained-environment process on four physical noise intervals has a negative partial transpose across a partition of whole channel links and therefore requires quantum memory. At fixed $\theta,T$, the recovery separation persists for sufficiently small $|\omega|$ whenever it is positive at $\omega=0$. The equalities of the observed statistics follow from scalar effects for both the physical Pauli coefficients and the corrected joint branches. In each interval, even-weight physical error paths act on the environment through $I_E$ and $X_E$, and odd-weight paths through $Z_E$. Their relative phase and the anticommutation of $X_E$ and $Z_E$ cancel the interference between the two corrected logical classes in every syndrome probability, as the real and imaginary amplitudes did for a classical sign. For $\omega\ne0$, the coupling operator $Z_E$ does not commute with the environment precession, which coherently mixes the $Z_E$ sectors in the joint evolution. The partial-transpose witness establishes that this process requires quantum memory under the conditions above. \Cref{app:quantum} gives these identities, an explicit four-link witness, a common known dilation, and the finite-storage-time continuity estimate. For example, the Shor code at $(\theta,T)=(0.1,267)$ keeps a certified loss of at least $0.0669$ at $\omega=10^{-4}$.

The witness certifies quantum memory of the physical noise process before QEC. The long-time results, namely the near-half loss and the matching calibration costs of the next section, are proved for the pair $\mathsf F/\mathsf R$ at the classical-memory point $\omega=0$. Extending them to fixed nonzero $\omega$ would require an analysis beyond finite-storage-time continuity.

\section{The cost of calibrating the recovery decision}
\label{sec:calibration}

We now keep the data circuit unchanged and ask how independent calibration experiments can supply the missing model information. The data run has storage time $T_\pi$ and the identity logical schedule between cycles, and its complete syndrome record is available to the final recovery. The code, fixed lookup table, angle $\theta$, and the two candidate models $\mathsf F,\mathsf R$ are known. Calibration therefore needs only to select the recovery, not reconstruct a general temporally correlated noise model.

\subsection{Calibration access and resource accounting}
\label{sec:access}

Calibration experiments differ in what the experimenter may prepare, control, and read out, but none of them gives direct access to the environment. A syndrome-only experiment returns stabilizer outcomes. A logical experiment additionally allows a known encoded input and a final logical readout. A third option also lets the experimenter choose when extraction occurs. Comparing these options separates the value of preparation and readout from the value of extraction timing, while the data run itself is unchanged. We compare three nested access classes, summarized in \cref{fig:overview}:
\begin{enumerate}
\item $\mathcal A_{\rm syn}$ uses complete QEC cycles and passes to the recovery only the syndrome records and randomness independent of the source. Logical operations between cycles may be adaptive and may use a private register, but their measurement results, their quantum outputs, and any action records that encode such results are withheld from the recovery.
\item $\mathcal A_{\rm log}$ uses the same complete QEC cycles and may pass on their full classical and quantum output. Ideal logical operations are available between complete QEC cycles, and different calibration runs may share a private quantum register.
\item $\mathcal A_{\rm del}$ additionally permits fresh, independently reset delayed-extraction queries: two consecutive noise intervals with no intervening control or measurement, followed by a single syndrome extraction and the fixed lookup-table correction. Their full output may be retained.
\end{enumerate}
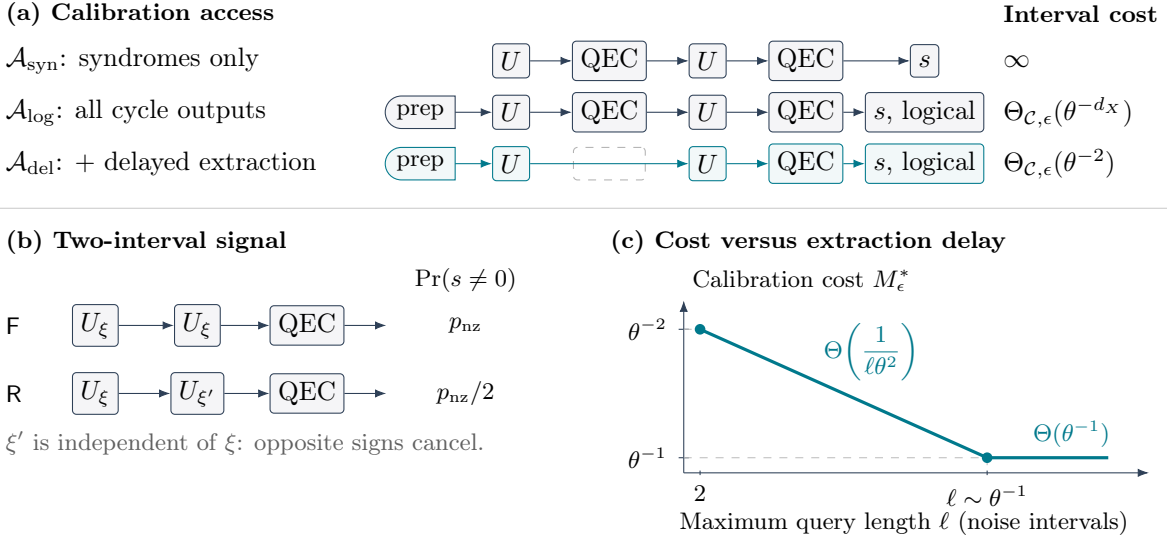
\begin{figure*}[tbp]
\centering
\begin{tikzpicture}[x=1cm,y=1cm,
 box/.style={draw=qecnavy,rounded corners=1.5pt,fill=qecnavy!5,
 minimum height=0.43cm,inner sep=3pt,font=\normalsize},
 gap/.style={draw=gray!60,dashed,rounded corners=1.5pt,
 minimum height=0.43cm,minimum width=0.95cm},
 arrow/.style={-{Latex[length=1.5mm]},draw=qecnavy},
 lab/.style={anchor=west,font=\normalsize}]
 \node[lab,font=\bfseries] at (0,0.62) {(a) Calibration access};
 \node[lab,font=\bfseries] at (13.2,0.62) {Interval cost};
 \node[lab] at (0,0) {$\mathcal A_{\rm syn}$: syndromes only};
 \node[lab] at (0,-0.68) {$\mathcal A_{\rm log}$: all cycle outputs};
 \foreach \yy [count=\idx] in {0,-0.68}{
  \node[box] (u\idx) at (6.8,\yy) {$U$};
  \node[box] (q\idx) at (8.1,\yy) {QEC};
  \node[box] (v\idx) at (9.4,\yy) {$U$};
  \node[box] (r\idx) at (10.7,\yy) {QEC};
  \draw[arrow] (u\idx)--(q\idx);
  \draw[arrow] (q\idx)--(v\idx);
  \draw[arrow] (v\idx)--(r\idx);
 }
 \filldraw[draw=qecnavy,fill=qecnavy!5] (6.06,-0.895)--(5.35,-0.895) arc (270:90:0.215)--(6.06,-0.465)--cycle;
 \node at (5.6,-0.68) {prep};
 \draw[arrow] (6.06,-0.68)--(u2);
 \node[box] (out1) at (12.27,0) {$s$};
 \node[box] (out2) at (12.27,-0.68) {$s$, logical};
 \draw[arrow] (r1)--(out1);
 \draw[arrow] (r2)--(out2);
 \node[lab] at (13.2,0) {$\infty$};
 \node[lab] at (13.2,-0.68) {$\Theta_{\C,\epsilon}(\theta^{-d_X})$};
 \node[lab] at (0,-1.36) {$\mathcal A_{\rm del}$: $+$ delayed extraction};
 \filldraw[draw=qecpetrol,fill=qecpetrol!5] (6.06,-1.575)--(5.35,-1.575) arc (270:90:0.215)--(6.06,-1.145)--cycle;
 \node at (5.6,-1.36) {prep};
 \node[box,draw=qecpetrol,fill=qecpetrol!5] (ud) at (6.8,-1.36) {$U$};
 \node[gap] at (8.1,-1.36) {};
 \node[box,draw=qecpetrol,fill=qecpetrol!5] (vd) at (9.4,-1.36) {$U$};
 \node[box,draw=qecpetrol,fill=qecpetrol!5] (qd) at (10.7,-1.36) {QEC};
 \draw[arrow,draw=qecpetrol] (6.06,-1.36)--(ud);
 \draw[arrow,draw=qecpetrol] (ud)--(vd);
 \draw[arrow,draw=qecpetrol] (vd)--(qd);
 \node[box,draw=qecpetrol,fill=qecpetrol!5] (out3) at (12.27,-1.36) {$s$, logical};
 \draw[arrow,draw=qecpetrol] (qd)--(out3);
 \node[lab] at (13.2,-1.36) {$\Theta_{\C,\epsilon}(\theta^{-2})$};
 \draw[gray!40] (0,-1.95)--(15.7,-1.95);
 \node[lab,font=\bfseries] at (0,-2.4) {(b) Two-interval signal};
 \node at (6.2,-2.9) {$\Pr(s\ne0)$};
 \foreach \yy/\model/\second [count=\idx] in
    {-3.5/\mathsf F/\xi,-4.4/\mathsf R/\xi'}{
  \node[anchor=west] at (0,\yy) {$\model$};
  \node[box] (bfirst\idx) at (1.3,\yy) {$U_\xi$};
  \node[box] (bsecond\idx) at (2.65,\yy) {$U_{\second}$};
  \node[box] (bqec\idx) at (4.1,\yy) {QEC};
  \draw[arrow] (bfirst\idx)--(bsecond\idx);
  \draw[arrow] (bsecond\idx)--(bqec\idx);
  \draw[arrow] (bqec\idx)--(5.15,\yy);
 }
 \node at (6.2,-3.5) {$p_{\rm nz}$};
 \node at (6.2,-4.4) {$p_{\rm nz}/2$};
 \node[lab,font=\bfseries] at (8,-2.4) {(c) Cost versus extraction delay};
 \begin{scope}[shift={(9.3,-5.25)},x=2.7cm,y=1cm]
  \draw[-{Latex[length=1.5mm]},qecnavy] (-0.08,-0.2)--(-0.08,2.08);
  \draw[-{Latex[length=1.5mm]},qecnavy] (-0.08,-0.2)--(2.2,-0.2);
  \node[anchor=west,font=\small] at (-0.08,2.35) {Calibration cost $M_\epsilon^*$};
  \foreach \yy/\ticklabel in {0/{\theta^{-1}},1.69897/{\theta^{-2}}}{
   \draw[qecnavy] (-0.105,\yy)--(-0.055,\yy);
   \node[anchor=east] at (-0.12,\yy) {$\ticklabel$};
  }
  \draw[gray!50,dashed] (-0.08,0)--(1.40654,0)--(1.40654,-0.2);
  \draw[gray!50,dashed] (-0.08,1.69897)--(0,1.69897);
  % Schematic log--log scaling for ell >= 2; multiplicative constants suppressed.
\draw[qecpetrol,very thick]
  (0,1.69897)--(1.40654,0)--(2,0);
\fill[qecpetrol] (0,1.69897) circle (2pt);
\fill[qecpetrol] (1.40654,0) circle (2pt);
\draw[qecnavy] (0,-0.17)--(0,-0.23);
\node[below] at (0,-0.25) {$2$};
\draw[qecnavy] (1.40654,-0.17)--(1.40654,-0.23);
\node[below] at (1.40654,-0.25) {$\ell\sim\theta^{-1}$};
\node[text=qecpetrol,font=\small] at (0.84,1.4)
  {$\displaystyle\Theta\!\left(\frac{1}{\ell\theta^2}\right)$};
\node[text=qecpetrol,font=\small] at (1.82,0.32)
  {$\Theta(\theta^{-1})$};
\node[font=\small] at (1.0,-0.87)
  {Maximum query length $\ell$ (noise intervals)};
 \end{scope}
 \node[anchor=west,font=\small,text=gray!75!black] at (0,-5.05)
   {$\xi'$ is independent of $\xi$: opposite signs cancel.};
\end{tikzpicture}
\caption{\label{fig:overview}Calibration for the fixed data circuit at storage time $T_\pi$, with fixed code and target loss $0<\epsilon<1/2$, as $\theta\to0$. (a) Each $U$ costs one noise interval, and QEC is one syndrome extraction followed by the fixed lookup-table correction. A delayed-extraction query omits the intermediate QEC step (dashed). Outputs are the syndrome $s$ and, where shown, the logical output, which may be measured or retained. The known-state preparations (prep) are those used by the protocols attaining the upper bounds. \Cref{sec:access} specifies the more general operations, retained outputs, and source resets that the classes permit. All classes receive the same SPP. (b) Under $\mathsf R$, independent signs cancel in half of the two-interval queries, halving the nonzero-syndrome probability $p_{\rm nz}$ [\cref{eq:delayed-rate}]. (c) Optimal calibration cost in noise intervals versus maximum query length $\ell\ge2$ [\cref{eq:delay-cost}], shown schematically on logarithmic axes with multiplicative constants suppressed. The cost scales as $(\ell\theta^2)^{-1}$ until $\ell$ is of
order $\theta^{-1}$, then saturates at order $\theta^{-1}$.}
\end{figure*}
All classes receive the SPP at no cost. It is identical for the two models and therefore carries no information about which one holds. The types and lengths of all runs, the source-reset boundaries, and the total interval budget are set in advance, while operations at the permitted boundaries may depend on earlier outcomes. The model does not change between runs, and each source reset supplies fresh hidden randomness, independent of the other runs given the model. A source reset is a separate operation from the preparation of a new logical state. The environment for the data run is prepared independently of the calibration environments, conditional on the model. Retained calibration systems first interact with the data at final recovery.

The cost of a protocol is its number of physical noise intervals: a protocol with $m$ complete QEC cycles and $N_{\rm del}$ two-interval queries costs $M=m+2N_{\rm del}$. Logical preparations, ideal logical controls, readout, and private registers are free in this accounting. Let $M^*_{\mathcal A,\epsilon}$ be the smallest budget, fixed in advance, for which the worst-case expected fidelity loss relative to the noise-aware final CPTP optimum is at most $\epsilon$. The noise-aware recovery knows which model governs the run but is not given its hidden signs. The budget measures interval complexity rather than wall-clock time or adaptive stopping. The lower bounds below also cover calibration procedures that keep a quantum state and never output a classical estimate of the model.

\subsection{Two experiments that reveal the temporal correlation}

With complete QEC cycles, we prepare a logical $\Zbar=+1$ eigenstate, store it for $T_\pi$ cycles, and measure $\Zbar$ once at the end. Because the logical channel is a mixture of $I$ and $\Xbar$, the success probability of this ideal logical-$Z$ memory test equals the entanglement fidelity. The mean readouts, $A^{T_\pi}$ under $\mathsf R$ and $\operatorname{Re}(A-iB)^{T_\pi}$ under $\mathsf F$, approach $+1$ and $-1$, respectively. A single such run therefore identifies the model with conditional error probabilities $\RefInf{\mathsf R}$ and $\RefInf{\mathsf F}$, both $O_{\C}(\theta)$. Choosing the corresponding fixed final correction in the independent data run then incurs a loss of at most $2\max_h\RefInf{h}=O_{\C}(\theta)$.

A delayed-extraction query instead probes interference before extraction [\cref{fig:overview}(b)]. Each query channel returns the syndrome in a classical register together with the logical state. Write $\mathcal E_{\rm av}(2\theta)$ for the channel, with the syndrome retained, of one extraction and correction after the rotation $U_\xi(2\theta)$ with a fair sign $\xi$, and $\mathcal E_{\rm id}$ for the channel that records a zero syndrome and leaves the logical state unchanged. The two-interval channels are
\begin{equation}
 \mathcal Q_{\mathsf F,2}=\mathcal E_{\rm av}(2\theta),\qquad
 \mathcal Q_{\mathsf R,2}
 =\tfrac12\mathcal E_{\rm av}(2\theta)
       +\tfrac12\mathcal E_{\rm id}.
 \label{eq:delayed-mixture}
\end{equation}
Under $\mathsf R$, opposite signs cancel exactly. Consequently the nonzero-syndrome probabilities are $p_{\rm nz}$ under $\mathsf F$ and $p_{\rm nz}/2$ under $\mathsf R$, where
\begin{equation}
 p_{\rm nz}=4\kappa_{\C}\theta^2+O_{\C}(\theta^4),\qquad
 \kappa_{\C}\coloneqq \sum_{s\ne0}\bigl(n_s^{(1)}\bigr)^2>0.
 \label{eq:delayed-rate}
\end{equation}
Here $n_s^{(1)}$ is the number of weight-one physical $X$ errors with syndrome $s$. Weight-one errors with the same syndrome differ by an $X$ stabilizer, which acts as $+1$ on the code space, so their amplitudes add before squaring and $\kappa_{\C}$ can exceed $n$. Independent repetitions of this binary experiment distinguish the two rates with fixed error using $O_{\C,\epsilon}(\theta^{-2})$ noise intervals.

For any joint distribution of two fair signs, the same measurement gives
\begin{equation}
 \Pr(s\ne0)=\frac{p_{\rm nz}}2
       \bigl(1+\mathbb E[\xi_1\xi_2]\bigr).
 \label{eq:correlation-query}
\end{equation}
The query thus directly measures the nearest-neighbor sign correlation, a two-time correlation of the field direction, which equals one under $\mathsf F$ and zero under $\mathsf R$. For a general temporal process, recovery can also depend on higher-order correlations.

\subsection{Matching costs, including quantum calibration output}

\begin{theorem}[Optimal calibration interval costs]
\label{thm:costs}
Fix a code and lookup table as in \cref{def:family}, and fix $0<\epsilon<1/2$. For the data run at $T_\pi$, as $\theta\to0^+$,
\begin{align}
 M^*_{\mathcal A_{\rm syn},\epsilon}&=\infty,\notag\\
 M^*_{\mathcal A_{\rm log},\epsilon}
   &=\Theta_{\C,\epsilon}(T_\pi)
     =\Theta_{\C,\epsilon}(\theta^{-d_X}),\notag\\
 M^*_{\mathcal A_{\rm del},\epsilon}
   &=\Theta_{\C,\epsilon}(\theta^{-2}).
 \label{eq:costs}
\end{align}
The first equality holds for all sufficiently small $\theta$. The lower bounds allow an arbitrary private quantum register and a common final CPTP processor acting jointly on the retained calibration state and the complete data output.
\end{theorem}

The two experiments above give the upper bounds. The lower bounds limit the information that any retained calibration state can carry, before any measurement. Recall the one-cycle instruments $\mathcal E_\pm$ of \cref{sec:visibility}. The diamond norm allows extension by the identity on an arbitrary ancillary register, and their sign sensitivity is
\begin{equation}
 \beta_{\rm ret}
 =\left\|\frac{\mathcal E_+-\mathcal E_-}{2}\right\|_\diamond
 =2\sum_s|u_sv_s|=\Theta_{\C}(\theta^{d_X}),
 \label{eq:beta}
\end{equation}
where the syndrome blocks make the absolute contributions add. Although $\mathsf F$ and $\mathsf R$ have exactly the same averaged single-cycle channel, this sensitivity limits how fast information about the sign, and hence about the model, accumulates over a longer run. We bound the output distance by replacing one noise call at a time, following the hybrid method~\cite{bbbv1997} and using trace-norm contractivity,
\begin{equation}
 \Dtr(\sigma_{\mathsf F},\sigma_{\mathsf R})
 \le \frac{m\beta_{\rm ret}}2
              +2N_{\rm del}n^2\theta^2.
 \label{eq:hybrid}
\end{equation}
Combining this bound with \cref{prop:calibration-loss} and the bound $\Gamma^{\mathrm{CPTP}}_{\mathsf F\mathsf R}\ge1/2 -\RefInf{\mathsf F}-\RefInf{\mathsf R}$ on the conflict of \cref{eq:general-conflict} for the pair $\mathsf F,\mathsf R$ gives
\begin{equation}
 \begin{split}
 &\max_h\{\FC{h}-\Fid{h}(\mathcal G)\}\\
 &\quad\ge\left[\frac{1-\Dtr(\sigma_{\mathsf F},\sigma_{\mathsf R})}{2}
       -\RefInf{\mathsf F}-\RefInf{\mathsf R}\right]_+.
 \end{split}
 \label{eq:joint}
\end{equation}
Here $\sigma_h$ is the full retained calibration state, including all classical and quantum outputs, and $\Dtr(\rho,\tau)=\|\rho-\tau\|_1/2$. The joint processor $\mathcal G$ of \cref{sec:calibration-loss} acts on this state and the recorded logical data output, so the bound covers strategies that never measure the calibration output before recovery. \Cref{app:calibration} gives the proofs. Syndrome-only outputs have the same state under both models for every budget, so the loss stays at its uncalibrated value, which exceeds $\epsilon$ for small $\theta$.

The leading sensitivity also yields explicit constants relative to the storage time: $M^*_{\mathcal A_{\rm log},\epsilon}/T_\pi$ lies asymptotically between $2(1-2\epsilon)\binom{2r}{r}/(\pi4^r)$ and $1$, which for $d_X=3$ gives the interval $[(1-2\epsilon)/\pi,1]$ (\cref{app:cost-prefactors}).

\subsection{Longer extraction delays}

Two-interval delayed-extraction queries give the first step of a graded tradeoff between extraction delay and cost. Let $\mathcal A_{\rm del}(\ell)$ contain protocols using complete QEC cycles together with independent fresh delayed-extraction queries of any chosen length $2\le \lambda\le\ell$. Each query is an uninterrupted run of $\lambda$ noise intervals followed by a single syndrome extraction and the fixed lookup-table correction. Its channel $\mathcal Q_{h,\lambda}$ has the same input and output systems as the two-interval query, with arbitrary ancillary extensions allowed. All other permissions and the fixed-budget convention are unchanged. In particular, $\mathcal A_{\rm del}(1)=\mathcal A_{\rm log}$ and $\mathcal A_{\rm del}(2)=\mathcal A_{\rm del}$.

\begin{proposition}[Maximum-delay interpolation]
\label{thm:delay-cost}
For every code and fixed lookup table as in \cref{def:family} and every $0<\epsilon<1/2$, the smallest calibration interval budget $M^*_\epsilon(\theta,\ell)\coloneqq M^*_{\mathcal A_{\rm del}(\ell),\epsilon}$ satisfies
\begin{equation}
 M^*_\epsilon(\theta,\ell)
 =\Theta_{\C,\epsilon}\!\left(
   \frac{1}{\theta^{d_X}+\min\{(\ell-1)\theta^2,\theta\}}
                           \right).
 \label{eq:delay-cost}
\end{equation}
The constants and the range of sufficiently small angles are uniform in all integers $\ell\ge1$, including choices of $\ell$ that depend on the angle. The data circuit, the storage time $T_\pi$, and the noise-aware final CPTP benchmark do not depend on $\ell$.
\end{proposition}

For a length-$\lambda$ delayed-extraction query, the net rotation is $\lambda\theta$ under $\mathsf F$ and $S_\lambda\theta$ under $\mathsf R$, where $S_\lambda$ is the sum of $\lambda$ independent fair signs. Recall that $\mathsf R$ resamples the sign in every interval, including within a query. While $\lambda\theta$ is small, the nonzero-syndrome probabilities are of order $\lambda^2\theta^2$ and $\lambda\theta^2$, respectively. A test with fixed confidence therefore needs $O_{\C,\epsilon}((\lambda^2\theta^2)^{-1})$ queries, that is, $O_{\C,\epsilon}((\lambda\theta^2)^{-1})$ noise intervals. Choosing $\lambda$ of order $\theta^{-1}$ reaches inverse-linear cost, and longer blocks cannot improve this order [\cref{fig:overview}(c)]. Longer queries amplify a coherent rotation, as in phase-estimation-based gate calibration~\cite{kimmel2015}. Here the angle is known, and cost is measured against a fixed target recovery loss. \Cref{app:cadence} proves the matching lower bound for the complete output and gives a test that is uniform at small angles.

Allowing every length up to $\ell$, rather than exactly $\ell$, matters because, as $\lambda\theta$ approaches $\pi$, the fixed-sign rotation returns close to the identity up to a global phase, so the best query can be shorter than $\ell$. Extraction timing thus determines which interference is observable, while the benchmark remains the original protected memory experiment. Stopping rules and source-reset times chosen adaptively lie outside these results.

\section{Changing the QEC circuit changes the decoding problem}
\label{sec:protection}

Calibration adds information about a given data run. Additional physical or logical controls change the experiment itself. We compare the original circuit with ideal physical Pauli randomization, logical Pauli twirling over $\{I,\Zbar\}$, and logical dynamical decoupling (a midpoint echo). The same three controls have been compared for correlated coherent noise~\cite{rajmohan2026}, and Pauli twirling can turn noise correlations into a resource for circuit fidelity under Gaussian noise~\cite{brillant2026}. Here we ask how these controls change the recovery ambiguity and the need for calibration.

\subsection{Physical Pauli randomization}

We conjugate each noise interval by an independent, uniformly random physical Pauli, using ideal controls and compensating the known frame. This implements the ideal noise-tailoring step of randomized compiling~\cite{wallman2016}. Noisy or gate-dependent implementations require additional assumptions.

Discarding the randomization record gives the interval-by-interval Pauli twirl, and the SPP describes the averaged experiment exactly~\cite{kam2026}. For any $[[n,k]]$ stabilizer memory with ideal syndrome extraction and Pauli corrections, the SPP therefore determines the joint distribution $p(y,L)$ of syndrome records and residual logical Paulis $L\in\mathcal P_k$. Choosing the most probable logical class for each record maximizes entanglement fidelity, even among final CPTP recoveries, because, for each record, the fidelities of any recovery against the orthogonal Pauli branches sum to at most one~\cite{fletcher2008}:
\begin{equation}
 \begin{aligned}
 d^\star(y)&\in\operatorname*{arg\,max}_L p(L\mid y),\\
 \FP{\mathrm{twirl}}=\FC{\mathrm{twirl}}
   &=\sum_y\max_L p(y,L).
 \end{aligned}
 \label{eq:pauli-optimum}
\end{equation}
This sufficiency follows directly from the exact SPP description and requires neither the CSS nor the parity assumptions of \cref{def:family}. It settles how to decode the twirled circuit, but not how that circuit performs relative to the original circuit decoded with full knowledge of the noise model.

For $\mathsf F/\mathsf R$, twirling gives independent physical $X$ errors with probability $\sin^2\theta$ under either model, which removes the recovery ambiguity without calibration. Let $F_{\rm tw}(T)$ be the fidelity of the twirled circuit with the original lookup table and a final identity correction [\cref{eq:twirl-fidelity}]. We compare it with the strongest benchmark of the original circuit, its full-record, noise-aware CPTP optimum.

\begin{proposition}[Physical-randomization performance for $\mathsf F/\mathsf R$]
\label{prop:protection-comparison}
Fix a code in \cref{def:family} and a minimum-weight $X$ lookup table. Let $m_s$ be the number of weight-$(r+1)$ physical $X$ errors with syndrome $s$ that the table corrects to $\Xbar$. At the original storage time $T_\pi$,
\begin{align}
 1-F_{\rm tw}(T_\pi)
    &=\frac{\pi\theta}{|b_{\C}|}\sum_s m_s+o_{\C}(\theta),\notag\\
 1-\FC{h}(T_\pi)
    &=\frac{\pi\theta}{|b_{\C}|}\sum_s m_s^2+o_{\C}(\theta),
       \quad h\in\{\mathsf F,\mathsf R\}.
 \label{eq:protection-leading}
\end{align}
The first leading coefficient is at most the second, with equality exactly when every $m_s\in\{0,1\}$.
\end{proposition}

Coherent amplitudes add before squaring, giving $m_s^2$, whereas twirling adds probabilities, giving $m_s$. For Shor every $m_s\in\{0,1\}$, so the two coefficients agree. For Steane each nonzero syndrome has $m_s=3$, and twirling reduces the leading infidelity threefold. Jain et al.\ derived the same factor for a single round of Steane-code correction~\cite[Sec.~IV~A]{jain2023}. Our comparison differs in two respects: it is made at the storage time $T_\pi$, after many cycles, and it credits the original circuit with arbitrary record-dependent final CPTP recoveries and knowledge of the model (\cref{app:protection-comparison}). At finite angles the ordering can reverse, as for Steane at $\theta=0.1$ (\cref{tab:finite-angle}).

\subsection{What the randomization record changes}
\label{sec:retained-randomizers}

Keeping the record of applied Paulis cannot lower the attainable fidelity and can raise it, because recovery may then depend on the record. The SPP describes only the averaged experiment and need not determine this record-dependent optimum. Randomization records have also been used to estimate coherent errors~\cite{vyas2026}.

An exact two-interval example shows that the SPP can indeed fall short. Consider any stabilizer code encoding at least one qubit and a Hermitian physical Pauli $P$ in its normalizer that represents a nontrivial logical Pauli $\bar P$. In each fresh two-interval run, a fair classical sign $\xi$ is drawn, and the noise applies $W_\xi$ in the first interval and $W_{c\xi}$ in the second, where
\begin{equation}
 W_\xi\coloneqq e^{-i\xi\pi P/4},\qquad c\in\{+1,-1\}.
 \label{eq:normalizer-construction}
\end{equation}
The two values of $c$ define the two models. Each interval is dressed with an independent uniform physical Pauli $Q_j$, and its frame is compensated. Since $Q_jPQ_j^\dagger=\eta_jP$ with $\eta_j\in\{+1,-1\}$, the retained record determines each $\eta_j$. The net logical channel is conjugation by $\bar P$ when $c\eta_1\eta_2=+1$, and the identity otherwise.

The noise syndrome is always zero. Each interval has Pauli weights $1/2$ on $I$ and $P$, independent of the sign, so both models have the same SPP, uniform on $\{I,P\}^2$. Changing $c$ exchanges the required correction for every record, so the noise-aware record-dependent recovery attains fidelity one. For each record, the two possible logical channels have orthogonal Choi vectors, and the fidelities of any CPTP recovery under the two models sum to at most one. An equal mixture of the $I$ and $\bar P$ corrections attains $1/2$ under both models, so the best common recovery has worst-case fidelity $1/2$. Discarding the record gives the same channel $(\id+\operatorname{Ad}_{\bar P})/2$ in both models, with optimal fidelity $1/2$ even when the model is known.

Here the SPP suffices only after the record is discarded, at the cost of lowering the noise-aware optimum from one to $1/2$. Keeping the record permits fidelity one when $c$ is known, but the SPP does not determine $c$. At leading order, twirling the Steane $\mathsf F/\mathsf R$ circuit simplified the information requirement and improved protection at once, whereas here the two goals conflict. An accuracy-dependent variant that independently redraws $c$ with opposite biases of magnitude $2\epsilon$ requires $\Theta(\epsilon^{-2})$ independent two-interval queries to attain target loss $\epsilon$ (\cref{sec:accuracy-dependence}). This example uses classical memory and a finite-angle normalizer rotation, outside the weak local noise setting of \cref{sec:family}.

\subsection{Logical controls and the scope of calibration}
\label{sec:echo}

Logical dynamical decoupling using code normalizers has also been analyzed together with QEC in effective Pauli-noise models~\cite{kasatkin2026}. By \cref{eq:global-z-conjugation}, the physical Pauli $Z^{\otimes n}$ represents $\Zbar$ and reverses the sign. Logical Pauli twirling over $\{I,\Zbar\}$ is implemented with ideal logical controls: independently dress each interval with $I$ or $Z^{\otimes n}$, compensate the frame, and discard the record of applied Paulis. This averaging turns $\mathsf F$ exactly into $\mathsf R$, including the syndrome output. For logical dynamical decoupling, we use a midpoint echo: an ideal logical $\Zbar$ pulse at the midpoint and compensation of the known final frame. This changes the syndrome-averaged $\mathsf F$ multiplier to $|A-iB|^T$ for even $T$, or $A|A-iB|^{T-1}$ for odd $T$, while the $\mathsf R$ multiplier stays $A^T$. The midpoint echo is a logical Hahn echo~\cite{hahn1950}. Because $\Zbar$ anticommutes with $\Xbar$, it reverses all later branch angles. For even $T$ this cancels the mean drift under $\mathsf F$, leaving syndrome-induced fluctuations, while the independent signs of $\mathsf R$ are statistically unchanged. Both controls remove the order-one recovery conflict at $T_\pi$ and attain the original $\mathsf R$ optimum to first order.

For the generalized Shor code, echo is exactly optimal. This code has odd distance $d_X$ and consists of $d_X$ rows of $d_X$ qubits, with $Z$ checks within rows and products of adjacent full-row $X$ operators as $X$ checks. Suppose $0<\theta<\pi/(4d_X)$ and $T$ is even, the data circuit may include ideal syndrome-adaptive logical Pauli gates, and the recovery is a full-record final Pauli correction. Then a midpoint echo attains the optimal worst-case absolute fidelity $(1+A^T)/2$ over $\mathsf F$, $\mathsf R$, and their mixtures. The $\mathsf R$ model supplies the matching upper bound, proved in \cref{app:logical-echo}. General intermediate logical CPTP maps and physical Pauli randomization lie outside this class of controls. Physical Pauli randomization can perform strictly better, as the Steane comparison above shows.

Calibration is therefore the relevant tool when the data circuit is given, or when additional physical or logical controls are costly. With the ideal physical controls assumed above, the leading-order advantage of physical randomization over the original noise-aware optimum is $\pi\theta\sum_s(m_s^2-m_s)/|b_{\C}|$, which vanishes for Shor. The additional infidelity caused by imperfect controls should be weighed against this advantage and against the performance of the uncalibrated original circuit. The ideal comparison thus shows how the permitted QEC circuit and the retained records determine what the decoder must know.
\section{Discussion and outlook}\label{sec:discussion}

The information a decoder needs is set jointly by the storage circuit, the recovery actions, and the records available to them. In the circuit studied here, complete syndrome statistics and the SPP leave a temporal ambiguity that changes the best recovery. The $\mathsf F/\mathsf R$ separation arises from coherent logical evolution driven by a non-Markovian noise source with purely classical memory, which syndromes and the SPP cannot distinguish from a Markovian one. Extraction after every interval leaves the field sign unobserved, allowing its temporal correlations to alter the logical output without altering the syndrome record. At suitable storage times this ambiguity costs nearly half of the entanglement fidelity, even though each model alone is almost perfectly correctable. Identifying the relevant fidelity differences suffices to select an optimal Pauli recovery. Reconstructing the full noise process is unnecessary.

The resource results identify three ways to resolve the ambiguity. Logical preparation and readout expose the information carried by the protected logical output at the original extraction rate. Delayed extraction exposes a stronger physical signal in separate calibration runs. Logical dynamical decoupling and Pauli twirling instead change the QEC circuit itself. In the ideal $\mathsf F/\mathsf R$ comparison they remove the need for calibration, and their fidelity and implementation costs become the relevant measures. The two randomization examples show that a simpler noise description and better protection can pull in opposite directions. Keeping an available record never lowers the achievable fidelity. It can raise the optimum, but attaining that higher optimum can require more model information.

For an experiment with a prescribed data circuit, extraction can be delayed in separate calibration runs while the data run keeps its original QEC timing. The distance-three examples at $\theta=0.05$ have storage times $T_\pi$ of about $930$ and $2100$ QEC cycles for Steane and Shor, respectively (\cref{tab:finite-angle}). At this angle, the minimax losses are $0.20$ (Steane) and $0.36$ (Shor), increasing to $0.41$ and $0.47$ at $\theta=0.01$. The stated storage times are not calibration budgets, which also depend on the target loss and on the code-dependent constants in the bounds.

The matching costs hold for a fixed code in the weak-angle limit, with known interactions, noise realizations that are independent across runs given the model, trusted logical preparation and readout, and ideal syndrome extraction. They count noise intervals rather than gates or laboratory time. Two extensions go beyond these assumptions. Commuting $X$ faults and incorrect syndrome reports preserve the exact blindness of the reported histories under the conditions of \cref{app:reporting-faults}. For independent physical $X$ faults and i.i.d. report masks, distance-three codes retain the near-half loss at rates proportional to $\theta^2$, with trusted final syndrome cleanup. With a precessing environment qubit, the observed statistics remain identical, the physical noise requires quantum memory, and the recovery conflict persists at finite parameters. Extending the cost results to these settings is an open problem.

Once the effects or recorded logical channels are given explicitly, the span criterion is a rank test, the finite-model minimax Pauli loss is a linear program, and the CPTP optima are semidefinite programs. Constructing these inputs from a QEC circuit can, however, scale exponentially with storage time, and the cost results do not give a scalable algorithm for general correlated-noise decoding (\cref{app:computation}).

Several further directions remain. Unknown couplings and preparation or readout errors add unknowns beyond the initial environment state. A source that cannot be reset independently requires a joint analysis of calibration and storage instead of independent training runs. Adaptive choices of reset times, stopping rules, and control sequences define resource classes beyond the prescribed timing considered here. For scalable decoding, a further goal is to combine the identified recovery observables with controlled approximations of the recorded logical process that preserve the recovery choice or bound the resulting fidelity loss.

\begin{acknowledgments}

We acknowledge helpful discussions with Zachary Mann, Raymond Laflamme, and Shilin Huang. 
Part of this work was carried out while N.C. was at the Perimeter Institute and the Institute for Quantum Computing, and was partially supported by funding from Mike and Ophelia Lazaridis.
This work was supported by the NUCC Quantum Computing Seed Fund.

\end{acknowledgments}

\section*{Statement of AI Usage}

The quantum model presented in \cref{sec:quantum-extension} was proposed by the authors, whereas the corresponding classical limitations were identified through iterative interactions with AI tools. AI tools also suggested the delayed-calibration procedure and helped identify its associated cost, while the authors verified the proof. 

They were further used to assist with literature searches, which was subsequently checked by the authors. The code used for the numerical calculations was developed with AI assistance. AI tools were additionally used to assist with drafting and editing the manuscript.

\section*{Data and code availability}
The code and numerical data supporting this work are available at \url{https://github.com/DaneshMorales/what-must-a-decoder-know} and in the \texttt{anc/} directory of the accompanying source package. The README provides instructions for reproducing the finite-angle table and the curve in Fig.~\ref{fig:separation}. All figure definitions and plotting data are also embedded in the manuscript source.

\label{end:maintext}
\appendix
\crefalias{section}{appendix}
\crefalias{subsection}{subappendix}
\section{Process conventions and frame tracking}
\label{app:process-conventions}

\subsection{Choi conventions}
\label{app:choi-weights}

We use unnormalized Choi operators. For a map $\mathcal M:\mathcal B(\mathcal H_{\rm in})\to\mathcal B(\mathcal H_{\rm out})$, in a fixed orthonormal input basis,
\begin{equation}
 \begin{aligned}
 J(\mathcal M)&\coloneqq \sum_{j,l}\mathcal M(|j\rangle\langle l|)
                                    \otimes|j\rangle\langle l|,\\
 \mathcal M(X)&=\Tr_{\rm in}[J(\mathcal M)(I_{\rm out}\otimes X^{\mathsf T})].
 \end{aligned}
 \label{eq:choi-convention}
\end{equation}
The Choi operator acts on $\mathcal H_{\rm out}\otimes\mathcal H_{\rm in}$. Complete positivity is $J\ge0$, and trace preservation is $\Tr_{\rm out}J=I_{\rm in}$. Connecting matching output and input wires is the link product: tensor the unshared identity factors, transpose the shared subsystem of one Choi operator in the chosen basis, multiply, and trace over that subsystem. In particular, this contraction gives $J(\mathcal N\circ\mathcal M)$ when connecting two channels. The partial transpose is a convention for contraction, not a physical transpose operation.

A deterministic $T$-interval comb with successive input and output spaces $\mathcal H_{{\rm in},j},\mathcal H_{{\rm out},j}$ has positive Choi operator $\Upsilon^{(T)}$, and there exist positive operators $\Upsilon^{(j)}$ on the first $j$ input and output spaces, $j=1,\ldots,T-1$, such that
\begin{equation}
 \begin{gathered}
 \Tr_{{\rm out},j}\Upsilon^{(j)}
       =I_{{\rm in},j}\otimes\Upsilon^{(j-1)},\\
 j=1,\ldots,T,\qquad\Upsilon^{(0)}\coloneqq 1.
 \end{gathered}
 \label{eq:comb-causality}
\end{equation}
These constraints express the absence of signaling from later interventions to earlier outputs~\cite{chiribella2009,pollock2018}. An environment coupled sequentially to the system realizes such a comb and can carry memory between intervals. Contracting fixed encoding, extraction, and correction maps, while retaining the logical intervention wires, produces the logical process tensor. Classical syndrome wires are diagonal in their designated record bases. Once all logical interventions are inserted, the remaining Choi blocks are $J(\mathcal N_{h,y})$, with $\sum_y\Tr_{\rm out}J(\mathcal N_{h,y})=I_{\rm in}$.

For a logical map $\mathcal M:\mathcal B(\mathcal H_L)\to\mathcal B(\mathcal H_L)$, entanglement fidelity uses the normalized state $|\Omega_D\rangle\coloneqq D^{-1/2}\sum_j|j,j\rangle$. For any Kraus operators $K_\mu$ of $\mathcal M$,
\begin{equation}
 \begin{aligned}
 F_{\mathrm e}(\mathcal M)
 &\coloneqq \langle\Omega_D|(\mathcal M\otimes\id)
   (|\Omega_D\rangle\langle\Omega_D|)|\Omega_D\rangle\\
 &=D^{-2}\sum_\mu|\Tr K_\mu|^2.
 \end{aligned}
 \label{eq:kraus-fidelity}
\end{equation}
The last formula extends linearly to trace-nonincreasing branches. All branch fidelities in the paper are unnormalized contributions to the unconditional memory fidelity.

For a logical Pauli $P$, define the unit Pauli--Choi vector $|P\rangle\!\rangle\coloneqq(P\otimes I)|\Omega_D\rangle$ on $\mathcal H_L\otimes\mathcal H_L$, and the rescaled branch operator $\widehat J_{h,y}\coloneqq J(\mathcal N_{h,y})/D$ on the same space. Pauli orthogonality and Eq.~\eqref{eq:kraus-fidelity} give
\begin{equation}
 \begin{aligned}
 w_h(y,L)&=\langle\!\langle L|\widehat J_{h,y}|L\rangle\!\rangle,\\
 \Tr\widehat J_{h,y}&=p_h(y)=\Tr\mathcal N_{h,y}(I/D).
 \end{aligned}
 \label{eq:choi-recovery-weights}
\end{equation}
Thus $\widehat J_{h,y}$ retains the branch probability for a maximally mixed logical input; it is not conditioned to have unit trace. Keeping the record gives $\sum_y|y\rangle\langle y|\otimes\widehat J_{h,y}$, whereas summing the blocks without the record register discards it.

To see how logical histories contribute to these weights, consider ideal syndrome extraction, a fixed lookup table, and identity logical operations between intervals. The corrected joint branch in cycle $j$ acts on $\mathcal H_L\otimes\mathcal H_E$ as
\begin{equation}
 \begin{aligned}
 M_{j,s}&\coloneqq(V^\dagger R_s\Pi_s\otimes I_E)
                     U_j(V\otimes I_E)\\
        &=\sum_{P\in\mathcal P_k}P\otimes B^{(j)}_{s,P},
 \end{aligned}
 \label{eq:logical-cycle-expansion}
\end{equation}
with environment operators $B^{(j)}_{s,P}\in\mathcal B(\mathcal H_E)$. For $y=(s_1,\ldots,s_T)$ and a logical Pauli history $\boldsymbol P=(P_1,\ldots,P_T)$, define its environment amplitude $A_{y,\boldsymbol P}$. Multiplying the cycle branches gives $M_y=\sum_LL\otimes B_{y,L}$, where
\begin{equation}
 \begin{aligned}
 A_{y,\boldsymbol P}&\coloneqq
 B^{(T)}_{s_T,P_T}\cdots B^{(1)}_{s_1,P_1},\\
 B_{y,L}&=\sum_{\boldsymbol P}
 \frac{\Tr(L^\dagger P_T\cdots P_1)}{D}\,A_{y,\boldsymbol P}.
 \end{aligned}
 \label{eq:logical-history-sum}
\end{equation}
The trace coefficient retains the phase whenever $P_T\cdots P_1$ equals $L$ up to $\pm1$ or $\pm i$. Accordingly, $w_h(y,L)=\Tr[\rho_E(h)B_{y,L}^\dagger B_{y,L}]$ contains cross terms between distinct histories with the same final Pauli class. These history coherences contribute to the diagonal entries in Eq.~\eqref{eq:choi-recovery-weights}; discarding them before connecting the intervals can change the Pauli recovery scores. The expansion is over logical histories after the QEC operations, whose coefficients already include interference within each cycle. Their diagonal data are generally distinct from the physical Pauli-history probabilities defining the SPP.

The same history amplitudes give the logical process tensor before the intermediate logical operations are inserted. Write $|\boldsymbol P\rangle\!\rangle=|P_T\rangle\!\rangle\otimes\cdots\otimes|P_1\rangle\!\rangle$. Under the ideal extraction and correction assumptions above, the rescaled Choi operator for the syndrome record $y$, acting on $(\mathcal H_L\otimes\mathcal H_L)^{\otimes T}$, is
\begin{equation}
 \begin{aligned}
 \widehat\Upsilon_{h,y}^{(T)}
 &\coloneqq D^{-T}\Upsilon_{h,y}^{(T)}\\
 &=\sum_{\boldsymbol P,\boldsymbol Q}
 \Tr_E\!\left[
 \rho_E(h)A_{y,\boldsymbol Q}^{\dagger}
 A_{y,\boldsymbol P}
 \right]\\
 &\qquad\times
 |\boldsymbol P\rangle\!\rangle
 \langle\!\langle\boldsymbol Q|.
 \end{aligned}
 \label{eq:logical-history-choi}
\end{equation}
This expression makes the coherences between logical histories explicit. The blocks are positive and satisfy
\begin{equation}
 \begin{aligned}
 0&\leq\Tr\widehat\Upsilon_{h,y}^{(T)}\leq1,\\
 \sum_y\Tr\widehat\Upsilon_{h,y}^{(T)}&=1.
 \end{aligned}
\end{equation}
Their traces give the syndrome probabilities in the Choi construction~\cite{pollock2018}, equivalently when each interval receives a fresh maximally mixed logical input, with the preceding logical output discarded and the environment retained. For the memory experiment, the prescribed intermediate operations must first be contracted with the process tensor to obtain $J(\mathcal N_{h,y})$. Its rescaled trace gives the record probability $p_h(y)$ for a maximally mixed initial logical input, as in Eq.~\eqref{eq:choi-recovery-weights}.

\subsection{Deferred commuting corrections}
\label{app:frame-defer}

\begin{lemma}[Deferred commuting corrections]
\label{lem:frame-defer}
Consider ideal syndrome projectors $\Pi_s$ of a one-logical-qubit CSS code and $X$-type corrections $R_s$ with syndrome $s$. Suppose every joint system--environment noise interval $U_j$ commutes with every $R_s$, and there are no intervening logical gates. For an actively corrected record $y=(s_1,\ldots,s_T)$, set
\begin{equation}
 \begin{gathered}
 \bar s_0\coloneqq 0,\qquad \bar s_j\coloneqq \bigoplus_{i=1}^j s_i,\\
 P_{\rm fr}(y)\coloneqq R_{s_T}\cdots R_{s_1}.
 \end{gathered}
 \label{eq:frame-record-map}
\end{equation}
Then, with identities on the environment suppressed,
\begin{align}
 &R_{s_T}\Pi_{s_T}U_T\cdots R_{s_1}\Pi_{s_1}U_1 V\notag\\
 &\hspace{8mm}=P_{\rm fr}(y)\Pi_{\bar s_T}U_T\cdots\Pi_{\bar s_1}U_1 V.
 \label{eq:frame-identity}
\end{align}
Hence correcting actively is equivalent to recording the raw syndromes and applying only the final frame $P_{\rm fr}(y)$: the two recorded instruments coincide after the invertible relabeling $s_j=\bar s_j\oplus\bar s_{j-1}$.
\end{lemma}

\begin{proof}
Use $\Pi_sR_c=R_c\Pi_{s\oplus c}$ and $[R_c,U_j]=0$ to move each correction past all later intervals and projectors. The accumulated syndrome shift at step $j$ is $\bar s_j$. The $X$ corrections commute, and $\operatorname{syn}(P_{\rm fr}(y))=\bar s_T$, where $\operatorname{syn}(P)$ denotes the syndrome of a Pauli $P$, so the final frame restores the code space. The identity holds at the operator level, before the environment is traced out, and therefore also for inputs entangled with a reference.
\end{proof}

Every syndrome-compatible final $X$-Pauli decoder on the raw record equals $P_{\rm fr}(y)$ followed by a record-dependent choice of $I$ or $\Xbar$, up to an $X$ stabilizer: its product with $P_{\rm fr}(y)^\dagger$ is an $X$-type element of the normalizer, and such elements form only two classes modulo the $X$ stabilizers. Changing the reference correction table therefore changes only a known, record-dependent logical frame, which the final recovery can undo. This gives the correction-table equivalence stated in Sec.~\ref{sec:frame-tracking}.

The optimality of Pauli recovery among CPTP maps for the $\mathsf F/\mathsf R$ example is proved separately (Appendix~\ref{app:regret}).

\section{Fidelity observables and statistical decisions}
\label{app:framework-proofs}

\subsection{Derivation of the environment effects}

Let $\widetilde K_{y,\mu}$ be a known joint branch from $\mathcal H_L\otimes\mathcal H_E$ to $\mathcal H_L\otimes\mathcal H_{E'}$. A common final environment can be obtained by padding if needed. Take orthonormal bases of $\mathcal H_E,\mathcal H_{E'}$ and first let $\rho_E=|\psi\rangle\langle\psi|$. The induced logical Kraus operators are $\langle e'|\widetilde K_{y,\mu}|\psi\rangle$. The Pauli expansion in Eq.~\eqref{eq:environment-effects} and $\Tr(LL')=D\delta_{L,L'}$ give
\begin{align}
 w_{\psi}(y,L)
 &=D^{-2}\sum_{\mu,e'}
       |\Tr[L\langle e'|\widetilde K_{y,\mu}|\psi\rangle]|^2\notag\\
 &=\sum_{\mu,e'}|\langle e'|B_{y,L,\mu}|\psi\rangle|^2
 =\langle\psi|E_{y,L}|\psi\rangle.
\end{align}
Linearity extends the identity to every $\rho_E$. Taking the logical partial trace of $\sum_{y,\mu}\widetilde K_{y,\mu}^\dagger \widetilde K_{y,\mu}=I$ and using the same orthogonality gives $\sum_{y,L}E_{y,L}=I_E$. Independently, for an arbitrary recorded instrument, the Pauli Parseval identity yields
\begin{align}
 \sum_L w_h(y,L)
 &=D^{-1}\sum_\mu\Tr(K_{h,y,\mu}^\dagger K_{h,y,\mu})\notag\\
 &=\Tr[\mathcal N_{h,y}(I/D)]=p_h(y).
 \label{eq:score-normalization}
\end{align}
Fix a reference action $L_0\in\mathcal P_k$. If $p_h(y)$ and all differences $\Delta w_y(L)\coloneqq w_h(y,L)-w_h(y,L_0)$ are known, then
\begin{equation}
 w_h(y,L_0)=\frac{p_h(y)-\sum_L\Delta w_y(L)}{D^2}.
 \label{eq:recover-scores}
\end{equation}
This reconstruction uses only the maximally mixed input distribution and does not require all conditional branches to be channels.

\subsection{Invisible perturbations and decision crossings}

For completeness, let $\Delta E$ be a Hermitian difference operator and $\mathcal V_{\rm cal}$ the span in Eq.~\eqref{eq:span}. The Hilbert--Schmidt orthogonal decomposition $\Delta E=\Delta E_\parallel+Q_E$ has $Q_E=Q_E^\dagger$, $\Tr Q_E=0$, and $\Tr(Q_EC_i)=0$. If $Q_E\ne0$, then $\Tr(Q_E\Delta E)=\Tr(Q_E^2)>0$. For any $0<\eta<1/(\dim\mathcal H_E\,\|Q_E\|_\infty)$, the two states $I_E/\dim\mathcal H_E\pm\eta Q_E$ are positive, normalized, and calibration-indistinguishable, whereas their $\Delta E$ expectations differ by $2\eta\Tr(Q_E^2)$. This proves necessity in Proposition~\ref{prop:span}. Sufficiency follows by linearity.

A numerical difference need not change a maximizing action. To establish such a change, suppose that at a full-rank $\rho_*$ two actions $L,L'$ tie at record $y$ and strictly exceed every other action. Let $Q_E\perp\mathcal V_{\rm cal}$ be Hermitian and put $\upsilon\coloneqq \Tr(Q_E\Delta E_{y;L,L'})\ne0$. For sufficiently small $\eta>0$, positivity of $\rho_*\pm\eta Q_E$ and the strict margin to the other actions persist. The fidelity difference becomes $\pm\eta\upsilon$, so the two states have different unique maximizing actions. Each losing action at that record has a fidelity deficit of at least $\eta|\upsilon|$, after decreasing $\eta$ if necessary. If a common rule assigns probabilities $p_L,p_{L'}$ to the two actions, the sum of its losses on this record is at least $\eta|\upsilon|(2-p_L-p_{L'})\ge\eta|\upsilon|$. The contributions of all other records are nonnegative. Hence every common Pauli decoder has worst-case loss at least $\eta|\upsilon|/2$. This gives a sufficient condition for an undetermined fidelity difference to change the decision.

\subsection{Finite-model minimax formula}

For a finite set $\Lambda$ of models consistent with the calibration data, let $\pi$ be one prior on models, shared by all records. The decoder set is a finite product of simplices, and the loss is bilinear in the decoder and prior. Sion's minimax theorem~\cite[Theorem~3.4]{sion1958} therefore gives
\begin{align}
 \Delta F^{\mathrm P,\star}(\Lambda)
  &=\max_{\pi\in\Delta(\Lambda)}\mathcal B_{\mathrm P}(\pi),\notag\\
 \mathcal B_{\mathrm P}(\pi)
  &\coloneqq \sum_h\pi_h\FP{h}
       -\sum_y\max_L\sum_h\pi_hw_h(y,L).
 \label{eq:finite-minimax}
\end{align}
Here $\Delta(\Lambda)$ is the set of probability distributions on $\Lambda$. The inner optimization separates over records. In particular, the equal-prior Pauli conflict is
\begin{equation}
 \Gamma_{01}^{\mathrm P}
 \coloneqq \frac{\FP{0}+\FP{1}}{2}
   -\frac12\sum_y\max_L[w_0(y,L)+w_1(y,L)].
 \label{eq:pauli-conflict}
\end{equation}
It vanishes exactly when a common optimal Pauli decoder exists and lower-bounds the two-model minimax loss. It need not equal that minimax loss: the least favorable prior need not be balanced.

\subsection{A quantum calibration state and joint final recovery}

We prove Proposition~\ref{prop:calibration-loss} without a common-syndrome-distribution or Pauli-noise assumption. Let
\[
 \tau_h\coloneqq (\id_R\otimes\Phi_h)
       (|\Omega_D\rangle\langle\Omega_D|)
\]
be the normalized joint state of the reference, the record register, and the logical output. Conditional independence of the environment preparations makes the joint input to the recovery $\sigma_h\otimes\tau_h$. The final processor acts on the calibration system and the accessible recorded data, not on the reference. For fixed $h$ and $\mathcal G$, its fidelity is a linear functional of the calibration state. The entire experiment, including the final maximally entangled-state test, is an effect on that state:
\begin{equation}
 \Fid{h}(\mathcal G,\sigma)=\Tr(M_h\sigma),\qquad 0\le M_h\le I_{\rm cal}.
 \label{eq:program-effect}
\end{equation}
Here $M_h$ is an operator on $\mathcal H_{\rm cal}$, and the inequalities follow because this probability lies between zero and one for every density operator $\sigma$ on $\mathcal H_{\rm cal}$.

For a trace-zero Hermitian $X$ and an effect $M$, $|\Tr(MX)|\le\|X\|_1/2$~\cite[Theorem~3.4]{watrous2018}. With $\delta\coloneqq \Dtr(\sigma_0,\sigma_1)$, replacing just the state in the model-1 score gives
\begin{align}
 &\frac{\Fid{0}(\mathcal G,\sigma_0)+\Fid{1}(\mathcal G,\sigma_1)}2\notag\\
 &\quad\le
 \frac{\Fid{0}(\mathcal G,\sigma_0)+\Fid{1}(\mathcal G,\sigma_0)}2
           +\frac\delta2\notag\\
 &\quad\le\sup_{\mathcal R}
                  \frac{\Fid{0}(\mathcal R)+\Fid{1}(\mathcal R)}2
           +\frac\delta2.
 \label{eq:program-replacement}
\end{align}
The last step uses the admissible common map $\mathcal R(X)\coloneqq \mathcal G(\sigma_0\otimes X)$. Subtracting from the average noise-aware fidelity proves the bound for average loss, and worst-case loss is at least average loss. Each individual loss is nonnegative because insertion of any fixed normalized state gives a CPTP recovery included in its benchmark. This proves the positive-part form of Eq.~\eqref{eq:program-bound}.

The argument extends to a smaller recovery class only if it is closed under insertion of every fixed program state into every permitted processor. A Pauli-only benchmark lacks that closure for an arbitrary joint CPTP processor. This is why the two action classes are kept distinct. The factor $1/2$ in front of the trace distance is optimal in general. Take identity and logical-$X$ data channels, and let calibration return the model label flipped with probability $e$. Then $\Gamma_{01}^{\mathrm{CPTP}}=1/2$, the trace distance is $1-2e$, and the optimal worst-case loss is $e$.

\subsection{The bound when calibration selects only a Pauli action}

Suppose calibration instead produces a classical transcript with distributions $P_0,P_1$, and the transcript $z$ selects a Pauli decoder $d_z$. Take densities $f_h$ of $P_h$ relative to a common dominating measure and let $m(z)\coloneqq \min\{f_0(z),f_1(z)\}$. The common-mass argument for two-point testing bounds~\cite{yu1997}, together with the nonnegativity of the loss, gives
\begin{align}
 &\frac12\sum_h\int f_h(z)\Floss{h}^{\mathrm P}(d_z)\,dz\notag\\
 &\quad\ge\int m(z)\,
             \frac{\Floss{0}^{\mathrm P}(d_z)+\Floss{1}^{\mathrm P}(d_z)}2\,dz\notag\\
 &\quad\ge\Gamma_{01}^{\mathrm P}[1-\TV(P_0,P_1)].
 \label{eq:classical-testing-bound}
\end{align}
The left-hand side is no larger than the worst-case expected loss. If a quantum calibration state is measured to select a Pauli, possibly using a measurement chosen after the data record $y$ arrives, apply this common-mass argument separately at each record. The trace distance bounds the total variation distance of each measurement's action distributions~\cite[Proposition~3.5]{watrous2018}. Each record contributes a nonnegative share of the equal-prior conflict, and these shares sum to $\Gamma_{01}^{\mathrm P}$. Consequently
\begin{equation}
 \max_h \Floss{h}^{\mathrm P}
 \ge\Gamma_{01}^{\mathrm P}
                  [1-\Dtr(\sigma_0,\sigma_1)].
 \label{eq:pauli-quantum-testing-bound}
\end{equation}
Here the protocol accesses the data only through its classical record when selecting the Pauli. Proposition~\ref{prop:calibration-loss} covers arbitrary joint processing of calibration and data systems.

\subsection{Computational aspects}
\label{app:computation}

The quantities in Sec.~\ref{sec:general} have a computational cost distinct from the cost of acquiring calibration data. Given explicit matrices $E_{y,L}$ and $C_i$, the span criterion is a rank test in a real vector space of dimension $(\dim\mathcal H_E)^2$. For a finite list of models and explicit scores $w_h(y,L)$, common optimality is checked by recordwise comparisons, and the minimax Pauli loss is a linear program over $d(L\mid y)$ and a worst-case loss variable [\cref{eq:finite-minimax}]. For all environment states consistent with fixed calibration probabilities, testing a candidate common correction reduces to semidefinite programs: minimize the difference between its score and each competing score over positive, unit-trace states satisfying those calibration constraints, and check that every minimum is nonnegative~\cite[Eq.~(17)]{keith2018}. The CPTP optima entering $\Gamma_{01}^{\mathrm{CPTP}}$ are also recovery semidefinite programs once the recorded logical channels are supplied~\cite{fletcher2007,kosut2008}.

These formulations do not guarantee efficient computation from a compact QEC circuit description. Constructing the effects or branch channels requires evaluating the coherent dynamics. The environment dimension and the number of syndrome histories can grow exponentially with storage time, and there are $4^k$ candidate logical Paulis. Even with one logical qubit, a small recovery optimization for each record does not remove the cost of representing all records. Numerical tests of span membership and ties require explicit tolerances or certified error bounds. Moreover, the finite-model linear program does not directly give a comparably sized program for an arbitrary continuous model family.

For two explicitly given calibration states, $\Dtr(\sigma_0,\sigma_1)$ is obtained from the spectrum of $\sigma_0-\sigma_1$. Bounding the distinguishability attainable by every allowed protocol at a given budget is a different task. Our calibration lower bounds use channel-sensitivity estimates and hybrid arguments to avoid constructing all such outputs. The examples also have special reductions: the two-round Shor effects lie in $\operatorname{span}_{\mathbb R}\{I_E,Z_{12}\}$, and the $\mathsf F/\mathsf R$ analysis uses scalar syndrome effects and conditional logical multipliers. These structures make the fixed-code results tractable, but the calibration-cost results do not establish a scalable algorithm for general correlated-noise decoding.

\section{The two-round Shor calibration example}
\label{app:shor-example}

\subsection{Branches, effects, and Pauli fidelity differences}
\label{app:shor-history-coherence}

Use stabilizers $Z_1Z_2,Z_2Z_3,Z_4Z_5,Z_5Z_6,Z_7Z_8,Z_8Z_9$ and $X_1X_2X_3X_4X_5X_6,X_4X_5X_6X_7X_8X_9$. In the conventional logical basis, consisting of three GHZ-plus or three GHZ-minus blocks, $P_{\rm sh}=X_1X_2X_3$ acts as logical $Z$. Only the first block is noisy. Write $c_j\coloneqq \cos\vartheta_j$, $t_j\coloneqq \sin\vartheta_j$. Expanding the three physical rotations and applying the lookup table of Sec.~\ref{sec:shor-example} gives
\begin{align}
 u_0&=c_1c_2c_3,&v_0&=-t_1t_2t_3,\notag\\
 u_j&=t_j\prod_{l\ne j}c_l,&
 v_j&=c_j\prod_{l\ne j}t_l,\qquad j=1,2,3.
 \label{eq:shor-uv}
\end{align}
The joint corrected branch in interval $\nu$ is $\widetilde K_{\nu,s}=u_sI-i v_sL_{\rm sh}\otimes Z_{E_\nu}$ up to an environment-only diagonal unitary. For a nonzero syndrome this omitted factor can be chosen as $-iZ_{E_\nu}$, while the zero branch needs no such factor. These factors commute with both interval operators and cancel from every squared Pauli coefficient, so omitting them leaves all effects and fidelities unchanged.

Since $\widetilde K_{\nu,s}^\dagger\widetilde K_{\nu,s}=q_sI$, the probability of successive syndromes $(s_1,s_2)$ is $q_{s_1}q_{s_2}$ for every initial environment state and every encoded input. The same trace argument permits trace-preserving logical controls and a private controller register between complete QEC cycles, provided their additional measurement outputs are not included among the syndrome-only observations. Expanding $U_\nu=\sum_{e\subseteq\{1,2,3\}}X^e\otimes B_{\nu,e}$ over physical $X$ strings $e$ gives the interval environment coefficients
\begin{equation}
 B_{\nu,e}=z_e Z_{E_\nu}^{|e|},\qquad
 z_e\coloneqq (-i)^{|e|}\prod_{j\in e}t_j\prod_{j\notin e}c_j.
\end{equation}
Thus the two-interval SPP probability of strings $e$ and $f$ in intervals 1 and 2 has environment effect $|z_ez_f|^2I_E$. Both the syndrome statistics and the SPP are therefore independent of the initial environment state, not only within the one-parameter fair family.

For $y=(s_1,s_2)$, multiplication gives
\begin{align}
 \widetilde K_{2,s_2}\widetilde K_{1,s_1}
 &=I\otimes A_{s_1s_2}-iL_{\rm sh}\otimes B_{s_1s_2},\notag\\
 A_{s_1s_2}&\coloneqq u_{s_1}u_{s_2}I_E-v_{s_1}v_{s_2}Z_{12},\notag\\
 B_{s_1s_2}&\coloneqq u_{s_2}v_{s_1}Z_{E_1}+u_{s_1}v_{s_2}Z_{E_2}.
 \label{eq:shor-two-branch}
\end{align}
These Hermitian environment coefficients commute. Consequently
\begin{align}
 E_{s_1s_2,I}&=(u_{s_1}^2u_{s_2}^2+v_{s_1}^2v_{s_2}^2)I_E
                         -2u_{s_1}v_{s_1}u_{s_2}v_{s_2}Z_{12},\notag\\
 E_{s_1s_2,L_{\rm sh}}&=(v_{s_1}^2u_{s_2}^2+u_{s_1}^2v_{s_2}^2)I_E
                         +2u_{s_1}v_{s_1}u_{s_2}v_{s_2}Z_{12}.
 \label{eq:shor-score-effects}
\end{align}
The remaining two logical Pauli effects vanish, and $E_{s_1s_2,I}+E_{s_1s_2,L_{\rm sh}}=q_{s_1}q_{s_2}I_E$. This proves Eq.~\eqref{eq:shor-all-contrasts} and shows that every Pauli score depends only on $\kappa=\Tr(\rho_EZ_{12})$, including for a nondiagonal environment state. It does not prove that all optimal CPTP recoveries depend only on $\kappa$: other environment moments can produce a mean logical rotation. The CPTP statement below is restricted to the fair family $\rho_\kappa=(I_E+\kappa Z_{12})/4$.

The same effect has a direct interpretation in the logical process tensor. Restrict to the fair family $\rho_\kappa$, and use the unit Pauli--Choi vectors of Appendix~\ref{app:choi-weights}, with logical dimension $D=2$. On $\mathcal H_L\otimes\mathcal H_L$, set
\begin{equation}
 \begin{aligned}
 \mathsf D_s&\coloneqq u_s^2|I\rangle\!\rangle\langle\!\langle I|
       +v_s^2|L_{\rm sh}\rangle\!\rangle\langle\!\langle L_{\rm sh}|,\\
 \mathsf C_s&\coloneqq iu_sv_s\bigl(
       |I\rangle\!\rangle\langle\!\langle L_{\rm sh}|
       -|L_{\rm sh}\rangle\!\rangle\langle\!\langle I|\bigr).
 \end{aligned}
 \label{eq:shor-history-diagonal-coherent}
\end{equation}
The rescaled one-cycle block for sign $\xi$ is $\mathsf D_s+\xi\mathsf C_s$. Averaging the two fair signs with $\mathbb E(\xi_1\xi_2)=\kappa$ gives the two-interval record block, with the later interval written first,
\begin{equation}
 \begin{aligned}
 \widehat\Upsilon_{\kappa,y}^{(2)}
 &\coloneqq\frac{\Upsilon_{\kappa,y}^{(2)}}{4}\\
 &=\mathsf D_{s_2}\otimes\mathsf D_{s_1}
       +\kappa\mathsf C_{s_2}\otimes\mathsf C_{s_1},\\
 \Tr\widehat\Upsilon_{\kappa,y}^{(2)}&=q_{s_1}q_{s_2}.
 \end{aligned}
 \label{eq:shor-history-choi}
\end{equation}
Here $\Upsilon_{\kappa,y}^{(2)}$ uses the unnormalized convention of Eq.~\eqref{eq:comb-causality}. Retaining the syndromes gives $\sum_y|y\rangle\langle y|\otimes\widehat\Upsilon_{\kappa,y}^{(2)}$. The omitted environment-only phases do not affect this expansion: for a fixed record they multiply every history amplitude by the same final environment unitary, which cancels from all coefficients $\Tr(\rho_E A_{y,\boldsymbol Q}^\dagger A_{y,\boldsymbol P})$.

The term $\kappa\mathsf C_{s_2}\otimes\mathsf C_{s_1}$ is off-diagonal in the logical Pauli-history basis. Composing the corrected branches, as in Eq.~\eqref{eq:shor-two-branch}, combines $(I,I)$ with $(L_{\rm sh},L_{\rm sh})$ and $(I,L_{\rm sh})$ with $(L_{\rm sh},I)$. Their interference changes the rescaled endpoint block $\widehat J_{\kappa,y}=J(\mathcal N_{\kappa,y})/2$ by
\begin{equation}
 \begin{aligned}
 \widehat J_{\kappa,y}-\widehat J_{0,y}
 &=2\kappa u_{s_1}v_{s_1}u_{s_2}v_{s_2}\\
 &\quad\times\bigl(
 |L_{\rm sh}\rangle\!\rangle\langle\!\langle L_{\rm sh}|
 -|I\rangle\!\rangle\langle\!\langle I|\bigr).
 \end{aligned}
 \label{eq:shor-history-endpoint}
\end{equation}
Its trace vanishes, and its two Pauli--Choi diagonal entries are equal and opposite. Equation~\eqref{eq:choi-recovery-weights} therefore recovers the correlation-dependent terms in Eq.~\eqref{eq:shor-score-effects} and their doubled contribution to the fidelity difference.

\subsection{Decision crossing and exact full-record loss}

For the angles specified in Sec.~\ref{sec:shor-example}, let $\phi_s\coloneqq 2\arctan(v_s/u_s)$, taking the displayed continuous branches. Then
\begin{equation}
 \begin{aligned}
 \phi_0&=-2\arctan(\sqrt3\zeta^4),\qquad \phi_1=\pi/3,\\
 \phi_2&=\phi_3=2\arctan(\sqrt3\zeta^2).
 \end{aligned}
 \label{eq:shor-phases}
\end{equation}
On the fair family, the normalized residual logical map is the real mixture of $I$ and $L_{\rm sh}$ with multiplier
\begin{equation}
 g_\kappa(s_1,s_2)=\cos\phi_{s_1}\cos\phi_{s_2}
                       -\kappa\sin\phi_{s_1}\sin\phi_{s_2}.
 \label{eq:shor-multiplier}
\end{equation}
In particular, $g_\kappa(1,1)=(1-3\kappa)/4$. For $0<\zeta\le1/5$, all sums and differences of two phases other than $2\phi_1$ lie strictly between $-\pi/2$ and $\pi/2$. Indeed $\phi_1+\phi_2<\pi/2$, $2\phi_2<\pi/2$, and $|\phi_0|<\phi_2$. At $\kappa=\pm1$ the multipliers are the corresponding cosines. Since they are affine in $\kappa$, they are positive throughout $[-1,1]$. Thus $(1,1)$ is the only record whose optimal Pauli changes.

For $\kappa=1$ versus $\kappa=0$, let $\eta_{\rm mix}$ be the probability of applying $L_{\rm sh}$ on that record. The identity is optimal on the other records. The respective fidelity losses are
\begin{equation}
 \begin{aligned}
 \Floss{1}(\eta_{\rm mix})&=\frac{q_1^2}{2}(1-\eta_{\rm mix}),\\
 \Floss{0}(\eta_{\rm mix})&=\frac{q_1^2}{4}\eta_{\rm mix}.
 \end{aligned}
\end{equation}
Their maximum is minimized at $\eta_{\rm mix}=2/3$, which gives the minimax fidelity loss $\Fmm[\mathrm{sh}]$ of $\kappa=1$ versus $\kappa=0$:
\begin{equation}
 \Fmm[\mathrm{sh}]=\frac{q_1^2}{6}
                   =\frac83\zeta^8+O(\zeta^{10}).
 \label{eq:shor-minimax}
\end{equation}
The simultaneous domination argument of Appendix~\ref{app:regret}, with $L_{\rm sh}$ replacing $\Xbar$, shows that general final CPTP recovery has the same minimax value. More generally, for $\kappa>1/3$ versus $0$, balancing the two affine losses gives $q_1^2(3\kappa-1)/(12\kappa)$. For $\kappa\le1/3$ a common identity decoder is optimal. This example thus separates undetermined fidelity differences from undetermined recovery decisions.

The crossing also realizes the perturbation argument of \cref{app:framework-proofs}. The state $\rho_{1/3}$ is full rank, and $\rho_{1/3}\pm\eta Z_{12}/4$ are physical for $0<\eta<2/3$. Their calibration data are identical, while their fidelity differences on $(1,1)$ are $\mp3q_1^2\eta/4$. Their exact minimax loss is $3q_1^2\eta/8$.

\subsection{Delayed extraction measures the missing direction}

In a separate calibration run, omit the extraction and correction between the two fixed interactions. In the joint $Z$-eigenvalue sector $(\xi_1,\xi_2)$, the product is
\begin{equation}
 U_2U_1=\prod_{j=1}^3
                e^{-i(\xi_1+\xi_2)\vartheta_jX_j}.
\end{equation}
Opposite signs produce the identity and no syndrome. Equal signs produce doubled angles. The zero-syndrome branch contains the identity and $P_{\rm sh}$ paths with relative imaginary phase, so its probability is $\prod_j\cos^2(2\vartheta_j)+\prod_j\sin^2(2\vartheta_j)$, independent of the encoded input. Both equal-sign sectors therefore give the nonzero probability $p_*$ of Eq.~\eqref{eq:shor-calibration}. Tracing out the environment removes its off-diagonal sign blocks, giving the operator identity $C_{\rm nz}=p_*(I_E+Z_{12})/2$ for an arbitrary environment input.

In particular, the measured nonzero rate $\Tr(\rho_EC_{\rm nz})=p_*(1+\kappa)/2$ selects the identity on $(1,1)$ below $2p_*/3$ and selects $L_{\rm sh}$ above it, because
\begin{equation}
 \Delta E_{11;I,L_{\rm sh}}
      =q_1^2\left(I_E-\frac{3}{2p_*}C_{\rm nz}\right).
\end{equation}
The full nonzero syndrome label has a conditional distribution independent of $\kappa$, so it gives no additional information about this parameter. Since $\sin^2(2\vartheta_1)=12\zeta^4/(1+3\zeta^4)^2$ and $\sin^2(2\vartheta_2)=\sin^2(2\vartheta_3)=4\zeta^2/(1+\zeta^2)^2$, \cref{eq:shor-calibration} gives $p_*=8\zeta^2+O(\zeta^4)$. Repeated runs with independent noise realizations are Bernoulli experiments of mean $p_*(1+\kappa)/2$. A fixed-error test of $\kappa=1$ versus $0$ therefore uses $\Theta(\zeta^{-2})$ such runs. This count refers to discriminating the two models with this experiment. For the short data task itself, the uncalibrated loss in Eq.~\eqref{eq:shor-minimax} is already $O(\zeta^8)$.
\section{CSS branch structure and noncancellation}\label{app:branch}
For each syndrome, physical $X$ errors occupy two corrected logical classes of opposite weight parity. After removing a common branch phase, one amplitude is real and the other purely imaginary. When both are nonzero, their relative phase is $\pm\pi/2$, giving \cref{eq:branch-form}. Changing $\xi$ conjugates the amplitudes and leaves their squared magnitudes unchanged. The following calculation gives both the signed coefficient $b_{\C}$ and the retained-channel norm. The remaining coefficient is $a_{\C}=2\sum_s m_s^2>0$: correcting every $X$ error of weight at most $r$ eliminates lower-weight logical errors, and every partition of a minimum logical support into $r$ and $r+1$ qubits supplies a nonzero contribution at weight $r+1$.

\subsection{Retained leading coefficient}\label{app:retained-coefficient}
\begin{proposition}[Leading coefficient with the syndrome retained]
\label{prop:retained-coefficient}
Fix a CSS code encoding one logical qubit, with the code space defined by the $+1$ eigenspaces of positive-phase stabilizer generators, even-weight $X$ stabilizers, and $X$ distance $d_X=2r+1\ge3$. Fix any $X$-correction lookup table that corrects every $X$ error of weight at most $r$. For the uniform rotation $U_\xi(\theta)=\bigotimes_{j=1}^n e^{-i\xi\theta X_j}$, let $N_{d_X}$ count the weight-$d_X$ supports of the nontrivial logical-$X$ class. Then, as $\theta\to0^+$,
\begin{align}
 \beta_{\rm ret}
 &=2N_{d_X}4^r\theta^{d_X}
       +O_{\C}(\theta^{d_X+2}),
       \label{eq:retained-leading-coefficient}\\
 \frac{\beta_{\rm ret}}{|B|}
 &=\frac{4^r}{\binom{2r}{r}}+O_{\C}(\theta^2).
       \label{eq:retained-leading-ratio}
\end{align}
The leading ratio is the same for every lookup table allowed by \cref{def:family}. The remainder constants may depend on the fixed code and table.
\end{proposition}

\begin{proof}
Write $C_X\subseteq C_Z^\perp\subseteq\mathbb F_2^n$ for the binary $X$-stabilizer and zero-$X$-syndrome spaces. Since the code encodes one qubit, $C_Z^\perp/C_X$ has dimension one. Choose $\mathbf{x}_{\mathrm L}\in C_Z^\perp\setminus C_X$ and orient the logical Pauli as $\Xbar=V^\dagger X^{\mathbf{x}_{\mathrm L}}V$. Every vector in $C_X$ has even weight, while every vector in $\mathbf{x}_{\mathrm L}+C_X$ has odd weight. If the correction for syndrome $s$ is $X^{e_s}$, its two corrected logical classes are
\begin{equation}
 \mathcal S_{s,0}=e_s+C_X,\qquad
 \mathcal S_{s,1}=e_s+\mathbf{x}_{\mathrm L}+C_X.
 \label{eq:retained-two-cosets}
\end{equation}
Errors in these classes act after correction as $I$ and $\Xbar$, respectively. Every $X$ stabilizer acts as $+1$ on the code space, giving the same sign to all representatives of either class.

For $\xi=+1$, before removing the common branch phase, the two amplitudes for $b\in\{0,1\}$ are
\begin{equation}
 h_{s,b}(\theta)
 =\sum_{p\in\mathcal S_{s,b}}
   (-i)^{|p|}(\sin\theta)^{|p|}(\cos\theta)^{n-|p|}.
 \label{eq:retained-coset-amplitudes}
\end{equation}
Here $|p|$ denotes Hamming weight. Define the minimum weights and their multiplicities by
\begin{equation}
 w_{s,b}\coloneqq \min_{p\in\mathcal S_{s,b}}|p|,\qquad
 c_{s,b}\coloneqq \#\{p\in\mathcal S_{s,b}:|p|=w_{s,b}\}.
 \label{eq:retained-coset-minima}
\end{equation}
Each coset has fixed parity, and equal-weight amplitudes have equal phase. Consequently
\begin{equation}
 h_{s,b}(\theta)
 =(-i)^{w_{s,b}}c_{s,b}\theta^{w_{s,b}}
     +O_{\C}(\theta^{w_{s,b}+2}).
 \label{eq:retained-coset-leading}
\end{equation}
The minimum-weight terms therefore add with a common phase and survive in the coset amplitude.

For $p\in\mathcal S_{s,0}$ and $p'\in\mathcal S_{s,1}$, their binary sum belongs to the nontrivial logical class. Hence
\begin{equation}
 d_X\le |p+p'|\le |p|+|p'|,
 \qquad w_{s,0}+w_{s,1}\ge d_X.
 \label{eq:retained-distance-constraint}
\end{equation}
The two minimum weights have opposite parity. A syndrome therefore contributes at degree $d_X$ only if $w_{s,0}+w_{s,1}=d_X$. Otherwise its first possible degree is at least $d_X+2$.

Suppose equality holds. One minimum weight is at most $r$, and the table corrects all such errors to $I$. Thus $w_{s,0}=j_s\le r$ and $w_{s,1}=d_X-j_s$. Every pair attaining this total weight consists of minimum-weight representatives of its two cosets. Equivalently, a lighter representative in either coset, paired with the opposite-class member, would contradict \cref{eq:retained-distance-constraint}. All degree-$d_X$ pairs in a fixed syndrome therefore have the same two weights and the same relative phase.

After removing a common phase, the branch is $u_sI-i v_s\Xbar$, so $u_sv_s=ih_{s,0}^*h_{s,1}$. For a contributing syndrome, \cref{eq:retained-coset-leading} gives
\begin{equation}
 u_sv_s=(-1)^{r-j_s}c_{s,0}c_{s,1}\theta^{d_X}
             +O_{\C}(\theta^{d_X+2}).
 \label{eq:retained-branch-leading}
\end{equation}
In particular, its absolute leading coefficient is exactly $c_{s,0}c_{s,1}$.

These products count a concrete set of physical-error pairs. At total weight $d_X$, equality in \cref{eq:retained-distance-constraint} forces $p$ and $p'$ to have disjoint supports, and $p+p'$ is a minimum-weight logical support. Conversely, take any such support and partition it into a subset $p$ of size $j\le r$ and its complement $p'$. The two errors have the same syndrome, so the table, which corrects $p$ to $I$, corrects $p'$ to $\Xbar$. Their weights sum to $d_X$, so both are minimum representatives of the corresponding corrected cosets. This is a bijection: the pair determines its logical support and its unique smaller subset. Consequently
\begin{equation}
 \sum_{s:\,w_{s,0}+w_{s,1}=d_X}c_{s,0}c_{s,1}
 =N_{d_X}\sum_{j=0}^{r}\binom{2r+1}{j}
 =N_{d_X}4^r.
 \label{eq:retained-partition-count}
\end{equation}

It remains to connect this count to the retained-channel norm. With syndrome output retained, the sign-odd map has classical blocks
\begin{equation}
 \frac{\mathcal E_+-\mathcal E_-}{2}(\rho)
 =\sum_s|s\rangle\langle s|\otimes
       \bigl(-iu_sv_s[\Xbar,\rho]\bigr).
 \label{eq:retained-sign-odd-blocks}
\end{equation}
The commutator map $\rho\mapsto-i[\Xbar,\rho]$ has diamond norm two: the triangle inequality gives the upper bound, and a $\Zbar$ eigenstate attains it. The same input attains the bound in every block, and trace norms add over the orthogonal syndrome blocks. Thus $\beta_{\rm ret}=2\sum_s|u_sv_s|$. Equations \eqref{eq:retained-branch-leading} and \eqref{eq:retained-partition-count} prove \eqref{eq:retained-leading-coefficient}.

Summing the same contributions with their signs in \cref{eq:retained-branch-leading} yields
\begin{equation}
 \begin{split}
 B=2\sum_s u_sv_s
 &=2N_{d_X}\theta^{d_X}
      \sum_{j=0}^{r}(-1)^{r-j}\binom{2r+1}{j}\\
 &\quad+O_{\C}(\theta^{d_X+2})\\
 &=2N_{d_X}\binom{2r}{r}\theta^{d_X}
       +O_{\C}(\theta^{d_X+2}).
 \end{split}
 \label{eq:retained-signed-count}
\end{equation}
The last identity follows by substituting Pascal's identity into the alternating sum, after which all terms except the boundary term cancel. Reversing the logical-axis convention can reverse $B$, without changing its magnitude. Dividing the two nonzero leading expansions proves \eqref{eq:retained-leading-ratio}.
\end{proof}

\section{Full-record recovery optimization}\label{app:regret}
The conditional channels in \cref{eq:conditional} follow by multiplying the normalized single-axis branch rotations and averaging their signs. The proof below reduces full-record CPTP recovery to randomized Pauli decisions. Equation~\eqref{eq:scaled-recovery-scores} in that proof gives the exact finite-angle noise-aware optima.

At $T_\pi$, the fixed choices $I$ under $\mathsf R$ and $\Xbar$ under $\mathsf F$ imply $\FC{h}\ge1-\RefInf{h}$. For any common decision, its summed fidelity is at most $1+\mathbb E|g_{\mathsf F}+g_{\mathsf R}|/2$, and
\begin{equation}
 \mathbb E|g_{\mathsf F}+g_{\mathsf R}|
 \le\mathbb E(1+g_{\mathsf F})+\mathbb E(1-g_{\mathsf R})
 =2(\RefInf{\mathsf F}+\RefInf{\mathsf R}).
\end{equation}
Subtracting from the sum of the two noise-aware lower bounds proves $\Fmm\ge[1/2-\RefInf{\mathsf F}-\RefInf{\mathsf R}]_+$. An equal mixture of $I$ and $\Xbar$ achieves fidelity $1/2$ under both models, proving the upper bound in \cref{thm:regret}.

\subsection{Uniform limit at scaled storage times}\label{app:scaled-horizon}

\begin{proof}[Proof of \cref{eq:horizon-curve} in \cref{thm:regret}]
At storage time $T$, write
\begin{equation}
 \Fmm(T)\coloneqq \inf_{\mathcal R}
 \max_{h\in\{\mathsf F,\mathsf R\}}
 \bigl[\FC{h}(T)-\Fid{h}(\mathcal R;T)\bigr],
 \label{eq:scaled-regret-definition}
\end{equation}
where $\mathcal R$ ranges over common final CPTP recoveries whose inputs are the logical output and full syndrome record. The common record law is $q_y=\prod_{j=1}^Tq_{s_j}$, and the two conditional channels are
\begin{equation}
 \mathcal L_{g_h(y)}
 \coloneqq \frac{1+g_h(y)}2\id+\frac{1-g_h(y)}2\mathcal X,
 \qquad \mathcal X(\rho)=\Xbar\rho\Xbar,
 \label{eq:scaled-real-channels}
\end{equation}
where $g_{\mathsf F}(y)=\cos\phi_{\rm tot}(y)$ and $g_{\mathsf R}(y)=\prod_j\cos\phi_{s_j}$. Both multipliers lie in $[-1,1]$. We first establish a stability bound that permits optimization after approximating these conditional channels.

For any qubit CPTP recovery $\mathcal R_y$, set $f_{I,y}\coloneqq F_{\mathrm e}(\mathcal R_y)$ and $f_{X,y}\coloneqq F_{\mathrm e}(\mathcal R_y\circ\mathcal X)$. Hilbert--Schmidt orthogonality of $I$ and $\Xbar$ gives, for any Kraus operators $W_{y,\mu}$ of $\mathcal R_y$,
\begin{align}
 f_{I,y}+f_{X,y}
 &=\frac14\sum_\mu\left(
   |\Tr W_{y,\mu}|^2+|\Tr(W_{y,\mu}\Xbar)|^2\right)\notag\\
 &\le\frac12\sum_\mu\Tr(W_{y,\mu}^\dagger W_{y,\mu})=1.
 \label{eq:scaled-recovery-bessel}
\end{align}
Choose $I$ with probability $(1+\eta_y)/2$ and $\Xbar$ with probability $(1-\eta_y)/2$, where $\eta_y\coloneqq f_{I,y}-f_{X,y}\in[-1,1]$. For every real $g$, this randomized Pauli recovery increases the fidelity after $\mathcal L_g$ by $(1-f_{I,y}-f_{X,y})/2\ge0$. The replacement works simultaneously under both models, record by record. Thus arbitrary final CPTP recovery and randomized Pauli recovery have the same noise-aware and minimax optima, with
\begin{align}
 \Fid{h}(\eta;T)&=\frac12\left(1+\sum_yq_y\eta_yg_h(y)\right),\notag\\
 \FC{h}(T)=\FP{h}(T)&=\frac12\left(1+\sum_yq_y|g_h(y)|\right).
 \label{eq:scaled-recovery-scores}
\end{align}

Consider another pair of real multipliers $\widetilde g_h(y)$ with the same record law, and denote their fidelity losses by $\widetilde{\Floss{h}}(\eta)$ and their minimax loss by $\Fmm[\widetilde g](T)$. Put $\delta_h\coloneqq \sum_yq_y|g_h(y)-\widetilde g_h(y)|$. For every decision $\eta$, its fidelity and the noise-aware optimum each change by at most $\delta_h/2$. More explicitly, its fidelity loss is
\begin{equation}
 \Floss{h}(\eta)=\frac12\sum_yq_y
        \bigl(|g_h(y)|-\eta_yg_h(y)\bigr),
\end{equation}
so $|\Floss{h}(\eta)-\widetilde{\Floss{h}}(\eta)|\le\delta_h$ uniformly over $\eta$. Taking first the maximum over models and then the infimum over the common decision set preserves this uniform bound. Therefore
\begin{equation}
 \left|\Fmm(T)-\Fmm[\widetilde g](T)\right|
 \le\max_{h\in\{\mathsf F,\mathsf R\}}\delta_h.
 \label{eq:scaled-minimax-stability}
\end{equation}
The bounds on the noise-aware optimum and on the minimax value hold even though the record alphabet grows with $T$.

To control the conditional multipliers, omit zero-probability outcomes, draw $s$ with law $q_s$, and set
\begin{equation}
 Z_s\coloneqq e^{i\phi_s},\qquad
 \mathbb E Z_s=A+iB=\lvert A+iB\rvert e^{i\alpha}.
 \label{eq:scaled-phase-moment}
\end{equation}
For sufficiently small $\theta$, $0<A\le\lvert A+iB\rvert\le1$ and $\alpha=\arctan(B/A)\ne0$. Under the product record law, $S_T\coloneqq \prod_{j=1}^TZ_{s_j}$ obeys $|S_T|=1$ and $\mathbb E S_T=(\lvert A+iB\rvert e^{i\alpha})^T$. A direct calculation gives
\begin{equation}
 \mathbb E|S_T-e^{iT\alpha}|^2=2(1-\lvert A+iB\rvert^T).
 \label{eq:scaled-phase-concentration}
\end{equation}
Taking real parts and using Cauchy--Schwarz yields
\begin{equation}
 \mathbb E|g_{\mathsf F}(y)-\cos(T\alpha)|
 \le\sqrt{2(1-\lvert A+iB\rvert^T)}.
 \label{eq:scaled-frozen-lone}
\end{equation}
For the resampled model, $g_{\mathsf R}(y)\le1$ and independence gives $\mathbb E g_{\mathsf R}(y)=A^T$. Hence
\begin{equation}
 \mathbb E|g_{\mathsf R}(y)-1|=1-A^T.
 \label{eq:scaled-resampled-lone}
\end{equation}
These expectations directly control the dependence of the conditional channels on the full syndrome record.

The fixed-code orders in \cref{eq:orders} imply $1-A=O_{\C}(\theta^{d_X+1})$ and $|\alpha|=\Theta_{\C}(\theta^{d_X})$. Fix $0<x_{\max}<\infty$. Since $T_\pi=\operatorname{round}(\pi/|\alpha|)$, uniformly over $0\le x\le x_{\max}$ and $T=T_\theta(x)$,
\begin{equation}
 0\le1-\lvert A+iB\rvert^T\le1-A^T\le T(1-A)
 =O_{\C,x_{\max}}(\theta).
 \label{eq:scaled-modulus-bound}
\end{equation}
The two roundings, in $T_\pi$ and in $T_\theta(x)$, together contribute only
\begin{equation}
 |T_\theta(x)\alpha-\operatorname{sgn}(B)\pi x|
 \le\frac{1+x}{2}|\alpha|
 =O_{\C,x_{\max}}(\theta^{d_X}).
 \label{eq:scaled-rounding}
\end{equation}
The cosine is even and has Lipschitz constant one. Combining \cref{eq:scaled-frozen-lone,eq:scaled-resampled-lone,eq:scaled-modulus-bound,eq:scaled-rounding} therefore proves the uniform conditional-channel estimates
\begin{align}
 \mathbb E|g_{\mathsf F}(y)-\cos(\pi x)|
 &=O_{\C,x_{\max}}(\sqrt\theta),\notag\\
 \mathbb E|g_{\mathsf R}(y)-1|
 &=O_{\C,x_{\max}}(\theta).
 \label{eq:scaled-conditional-limits}
\end{align}
At $T=0$ the record is empty and all products equal one, so the same bounds include $x=0$.

It remains to optimize the pair with constant multipliers $\widetilde g_{\mathsf F}\coloneqq \cos(\pi x)$ and $\widetilde g_{\mathsf R}\coloneqq 1$, retaining the same independent record. Any rule $\eta_y$ acts through its average $\overline\eta\coloneqq \sum_yq_y\eta_y$. If $\cos(\pi x)\ge0$, the common choice $\overline\eta=1$ is optimal under both models and the minimax fidelity loss is zero. Otherwise write $\cos(\pi x)=-\chi$ with $\chi=\chi(x)>0$. The two fidelity losses are
\begin{equation}
 \widetilde{\Floss{\mathsf F}}(\overline\eta)=\frac{\chi(1+\overline\eta)}2,
 \qquad
 \widetilde{\Floss{\mathsf R}}(\overline\eta)=\frac{1-\overline\eta}2.
\end{equation}
Their maximum is minimized by $\overline\eta=(1-\chi)/(1+\chi)$ and has value $\chi/(1+\chi)$. Applying \cref{eq:scaled-minimax-stability} to \cref{eq:scaled-conditional-limits} proves the uniform minimax limit in \cref{cor:scaled-horizon}. Equation~\eqref{eq:scaled-recovery-scores} for the noise-aware optimum, together with the same conditional estimates and $||u|-|v||\le|u-v|$, also gives, uniformly for $0\le x\le x_{\max}$,
\begin{align}
 \left|\FC{\mathsf F}(T_\theta(x))-
       \frac{1+|\cos(\pi x)|}{2}\right|
 &=O_{\C,x_{\max}}(\sqrt\theta),\notag\\
 \bigl[1-\FC{\mathsf R}(T_\theta(x))\bigr]
 &=O_{\C,x_{\max}}(\theta).
 \label{eq:scaled-oracle-limits}
\end{align}
\end{proof}

At $x=1$, \cref{thm:regret} gives the sharper endpoint estimates $\Fmm(T_\pi)=1/2-O_{\C}(\theta)$ and $\FC{h}(T_\pi)=1-O_{\C}(\theta)$ for both models. Across scaled storage times, the $\mathsf F$ noise-aware optimum tends to $(1+|\cos(\pi x)|)/2$, while the $\mathsf R$ noise-aware optimum tends to one.

\subsection{Sharpness of the square-root rate}\label{app:half-horizon-tightness}

\begin{proposition}[Fidelity loss at $T_\pi/2$]\label{prop:half-horizon-tightness}
Fix a code and $X$-correction lookup table satisfying \cref{def:family}, and assume in addition that the table is minimum-weight in the sense of \cref{sec:regret}. Then
\begin{equation}
 \Fmm\bigl(T_\theta(1/2)\bigr)
 =\sqrt{\frac{a_{\C}}{2|b_{\C}|}}\,\sqrt\theta
     +o_{\C}(\sqrt\theta).
 \label{eq:half-horizon-tightness}
\end{equation}
With their minimum-weight $X$-correction lookup tables, the coefficient is $3/2$ for Steane and $\sqrt3/2$ for Shor. In particular, on every fixed interval with $x_{\max}\ge1/2$, the uniform error in \cref{cor:scaled-horizon} is of order $\sqrt\theta$.
\end{proposition}

\begin{proof}
Use the principal branch angles $\phi_s\in(-\pi,\pi)$ supplied by \cref{lem:minimum-weight-small-phase}. They give the same conditional channels as any angles differing by multiples of $2\pi$. Let
\begin{equation}
 \mu_\theta\coloneqq \mathbb E_q\phi_s,\qquad
 \sigma_\theta^2\coloneqq \operatorname{Var}_q(\phi_s).
\end{equation}
The phase-moment estimates in \cref{eq:small-phase-moments}, with $a_{\C}=2\sum_s m_s^2$, give
\begin{align}
 \max_s|\phi_s|&=O_{\C}(\theta),\notag\\
 \mu_\theta&=B+O_{\C}(\theta^{d_X+2}),\notag\\
 \sigma_\theta^2&=2a_{\C}\theta^{d_X+1}
                    +O_{\C}(\theta^{d_X+3}).
 \label{eq:half-horizon-phase-moments}
\end{align}
Write $T\coloneqq T_\theta(1/2)$ and $\phi_{\rm tot}(y)=\sum_{j=1}^T\phi_{s_j}$ for records of this length. The branch records are independent with law $q_s$, while $T\sim\pi/(2|b_{\C}|\theta^{d_X})$. Since $\alpha=\arctan(B/A)=B+O_{\C}(\theta^{d_X+2})$, the moment estimates and rounding of the storage time imply
\begin{align}
 T\mu_\theta&=\operatorname{sgn}(b_{\C})\frac\pi2
                 +O_{\C}(\theta^2),\notag\\
 \frac{T\sigma_\theta^2}{\theta}
 &\longrightarrow \frac{\pi a_{\C}}{|b_{\C}|}.
 \label{eq:half-horizon-mean-variance}
\end{align}
Each centered phase increment is $O_{\C}(\theta)$, whereas the standard deviation of their sum is of order $\sqrt\theta$. Thus the triangular array of centered increments satisfies the Lindeberg condition, and
\begin{equation}
 \frac{\phi_{\rm tot}(y)-\operatorname{sgn}(b_{\C})\pi/2}{\sqrt\theta}
 \;\Longrightarrow\;
 N\!\left(0,\frac{\pi a_{\C}}{|b_{\C}|}\right).
 \label{eq:half-horizon-phase-clt}
\end{equation}
At this storage time $g_{\mathsf F}=\cos\phi_{\rm tot}(y)$. Expanding the cosine about $\operatorname{sgn}(b_{\C})\pi/2$ therefore gives
\begin{equation}
 \frac{g_{\mathsf F}}{\sqrt\theta}
 \;\Longrightarrow\;
 N\!\left(0,\frac{\pi a_{\C}}{|b_{\C}|}\right).
 \label{eq:half-horizon-multiplier-clt}
\end{equation}
Furthermore, $|g_{\mathsf F}|\le |\phi_{\rm tot}(y)-\operatorname{sgn}(b_{\C})\pi/2|$ and \cref{eq:half-horizon-mean-variance} show that $\mathbb E g_{\mathsf F}^2=O_{\C}(\theta)$. The scaled multipliers consequently have uniformly bounded second moments, so their negative parts are uniformly integrable. Defining $\gamma_{\mathsf F}(y)\coloneqq [-g_{\mathsf F}(y)]_+$ gives
\begin{equation}
 \mathbb E\gamma_{\mathsf F}(y)
 =\sqrt{\frac{a_{\C}}{2|b_{\C}|}}\,\sqrt\theta
      +o_{\C}(\sqrt\theta),\qquad
 \mathbb E\gamma_{\mathsf F}(y)^2=O_{\C}(\theta).
 \label{eq:half-horizon-negative-part}
\end{equation}

For sufficiently small $\theta$, all $\cos\phi_s$ are positive, so $\gamma_{\mathsf R}(y)\coloneqq g_{\mathsf R}(y)$ lies in $(0,1]$, and \cref{eq:scaled-resampled-lone} gives $\mathbb E[1-\gamma_{\mathsf R}(y)]=O_{\C}(\theta)$. The final identity is optimal under $\mathsf R$ and has fidelity loss $\mathbb E\gamma_{\mathsf F}(y)$ under $\mathsf F$, proving $\Fmm(T)\le\mathbb E\gamma_{\mathsf F}(y)$.

For the matching lower bound, put $\tau\coloneqq \theta^{1/4}$ and assign prior probability $p=1-\tau$ to $\mathsf F$ and $\tau$ to $\mathsf R$. On records where $g_{\mathsf F}\ge0$, the models share the optimal identity correction. On every other record, their opposite preferences make the minimum prior-averaged fidelity loss on that record equal to $\min\{p\,\gamma_{\mathsf F}(y),\tau\gamma_{\mathsf R}(y)\}$. The worst-case fidelity loss of any common decoder is at least its prior-averaged fidelity loss, so
\begin{align}
 \Fmm(T)
 &\ge\mathbb E\min\{p\,\gamma_{\mathsf F}(y),\tau\gamma_{\mathsf R}(y)\}\notag\\
 &\ge p\,\mathbb E\gamma_{\mathsf F}(y)
       -\frac{p^2\mathbb E\gamma_{\mathsf F}(y)^2}{\tau}
       -\tau\mathbb E[1-\gamma_{\mathsf R}(y)]\notag\\
 &=\mathbb E\gamma_{\mathsf F}(y)-O_{\C}(\theta^{3/4}).
 \label{eq:half-horizon-bayes-lower}
\end{align}
Here the second inequality uses the Lipschitz dependence of $\min\{u,v\}$ on $v$ and $(u-\tau)_+\le u^2/\tau$ for $u\ge0$. Combining this bound with \cref{eq:half-horizon-negative-part} and the identity-decoder upper bound proves \cref{eq:half-horizon-tightness}.

For Steane, each of the seven nonzero syndromes has three weight-two $X$ errors in its corrected logical-$X$ class. Therefore $a_{\C}=2\cdot7\cdot3^2=126$. Its seven weight-three logical supports give $|b_{\C}|=2\cdot7\binom21=28$, and the coefficient in \cref{eq:half-horizon-tightness} is $\sqrt{126/(2\cdot28)}=3/2$. Since $\chi(1/2)=0$, this positive coefficient also supplies the asserted lower bound on the uniform error.
\end{proof}

\section{Stochastic \texorpdfstring{$X$}{X} faults and incorrect reports}
\label{app:reporting-faults}

\subsection{Exact blindness with residual syndrome sectors}

We use the code and fixed $X$-correction lookup table of \cref{def:family}. Write $\operatorname{syn}(e)$ for the syndrome of $X^e$, and represent the physical sector with syndrome $s_{\rm res}$ by $R_{s_{\rm res}}V$, where $V$ encodes the logical qubit. Each round applies $U_{\xi_j}$ and a stochastic Pauli $X^{e_j}$, performs ideal complete syndrome extraction with true outcome $s_j$, and physically applies $R_{y_j}$ selected by its reported outcome $y_j$.

The fault and reporting distributions are common to the two models. A classical fault memory may select $e_j$ using earlier faults, reports, and external classical settings. The reporting kernel may also use the current true syndrome and fault. Conditional on those classical arguments, all kernels are independent of the sign trajectory and the logical quantum state. Their initial classical state is independent of the sign source and has the same law in both models. External settings may depend on earlier reports and source-independent randomness. Boundary CPTP controllers may use a private quantum register and the reported history, but preserve each physical syndrome sector. Private logical measurement outcomes are not inputs to the fault or reporting kernels.

\begin{lemma}[Blind reported histories with stochastic $X$ faults]
\label{lem:reporting-history}
Under these assumptions, starting in the code space, the joint distribution of the entire reported history and final residual syndrome is independent of the sign trajectory. In particular, these observations are blind to the temporal model in the sense of \cref{sec:equivalence}: their distributions coincide under $\mathsf F$ and $\mathsf R$. The result allows classical fault memory and incorrect physical corrections.
\end{lemma}
\begin{proof}
Fix an incoming sector $s_{\rm res}$, a fault $e$, a true syndrome $s$, and a report $y$. Put $w\coloneqq s\oplus s_{\rm res}\oplus\operatorname{syn}(e)$, and recall that all $X$ operators commute with $U_\xi$. After the applied $R_y$, the outgoing sector is $s_{\rm res}'=s\oplus y$. In canonical code coordinates the unweighted branch is
\begin{align}
 &V^\dagger R_{s_{\rm res}'}R_y\Pi_sX^eU_\xi R_{s_{\rm res}}V
       =L_{s_{\rm res},s,y,e}K_w(\xi),\notag\\
 &L_{s_{\rm res},s,y,e}\coloneqq V^\dagger R_{s_{\rm res}'}R_yX^eR_{s_{\rm res}}R_wV
       \in\{I,\Xbar\},
 \label{eq:reporting-sector-branch}
\end{align}
up to a phase. Here $L_{s_{\rm res},s,y,e}$ is the logical restriction of an $X$-type normalizer with zero syndrome. Its branch effect is therefore $q_wI$ by \cref{eq:branch-form}, even for an input entangled with a private register. Multiplying this probability by the stipulated classical kernels gives a sign-independent update of the fault memory, physical sector, and reported history. Induction over rounds proves the claim. The residual sector is retained in this update, so the reported process need not have independent increments.
\end{proof}

\begin{proposition}[Incorrect reports and physical $X$ faults]
\label{prop:reporting-faults}
Use the identity data schedule. Physical $X$ faults occur independently across qubits and rounds, with probability $p_X$ per qubit per round. Syndrome-report masks are i.i.d., with probability $p_{\rm rep}$ of a nonzero mask, and are independent of the physical faults. Both fault processes are independent of the sign source. Reports select the actual physical corrections. A trusted final syndrome extraction and correction returns the block to the code space, and both the noise-aware and the common recovery receive its syndrome. The complete reported history has the same distribution under $\mathsf F$ and $\mathsf R$.

Let $\RefInf[(p_X)]{h}(T)$ be the infidelity of the final correction $I$ under $\mathsf R$, or $\Xbar$ under $\mathsf F$, in the correctly reported noisy experiment, and let $q_{\rm nz}^{(p_X)}$ be its one-cycle nonzero-syndrome probability. Set
\begin{equation}
\epsilon_{\rm rep}(T)\coloneqq \min\{1,
 \begin{aligned}[t]
 &(2T-1)q_{\rm nz}^{(p_X)}p_{\rm rep}\\[-2pt]
 &\quad +(T-1)p_{\rm rep}^2\}.
 \end{aligned}
 \label{eq:reporting-budget-main}
\end{equation}
Let $\widetilde{\FC{h}}$ and $\widetilde{\Fmm}$ be the noise-aware full-record CPTP optimum and the minimax fidelity loss of the actual noisy instrument. Then
\begin{align}
 \widetilde{\FC{h}}&\ge1-\RefInf[(p_X)]{h}-\epsilon_{\rm rep},\notag\\
 \widetilde{\Fmm}
 &\ge[\tfrac12-\RefInf[(p_X)]{\mathsf F}-\RefInf[(p_X)]{\mathsf R}
                   -2\epsilon_{\rm rep}]_+.
 \label{eq:reporting-main}
\end{align}
For fixed $d_X=3$, rates $p_X=\kappa_X\theta^2$ and $p_{\rm rep}=\kappa_{\rm rep}\theta^2$ with fixed $\kappa_X,\kappa_{\rm rep}\ge0$ give $\widetilde{\FC{h}}=1-O(\theta)$ and $\widetilde{\Fmm}=1/2-O(\theta)$ at the storage time $T_{\pi,p_X}=\Theta_{\C,\kappa_X}(\theta^{-3})$ of \cref{eq:reporting-noisy-horizon}, defined by the noisy logical phase.
\end{proposition}

\subsection{A sparse-event bound for the actual noisy data task}

We first analyze the physical faults with correct reports, under the independence assumptions of \cref{prop:reporting-faults}. With $\nu_{p_X}(e)\coloneqq p_X^{|e|}(1-p_X)^{n-|e|}$, averaging the fixed-fault branches gives effects $q_s^{(p_X)}I$ and subnormalized transverse multipliers $c_s^{(p_X)}-i\xi d_s^{(p_X)}$, with real coefficients and $|c_s^{(p_X)}-i d_s^{(p_X)}|\le q_s^{(p_X)}$. Define
\begin{equation}
 \begin{gathered}
 A_{p_X}\coloneqq \sum_s c_s^{(p_X)},\qquad B_{p_X}\coloneqq \sum_s d_s^{(p_X)},\\
 q_{\rm nz}^{(p_X)}\coloneqq 1-q_0^{(p_X)}.
 \end{gathered}
 \label{eq:reporting-noisy-coefficients}
\end{equation}
At any integer storage time $T\ge1$, the correctly reported, record-discarded multipliers are $A_{p_X}^T$ under $\mathsf R$ and $\operatorname{Re}[(A_{p_X}-iB_{p_X})^T]$ under $\mathsf F$. Set
\begin{align}
 \RefInf[(p_X)]{\mathsf R}(T)&\coloneqq \frac{1-A_{p_X}^T}{2},\notag\\
 \RefInf[(p_X)]{\mathsf F}(T)&\coloneqq \frac{1+\operatorname{Re}[(A_{p_X}-iB_{p_X})^T]}{2}.
 \label{eq:reporting-noisy-defects}
\end{align}
These quantities refer to the fixed final corrections $I$ under $\mathsf R$ and $\Xbar$ under $\mathsf F$.

Now add i.i.d. report masks $b_j$, independent of the physical faults and the sign source, with $b_0=0$ and $\Pr(b_j\ne0)=p_{\rm rep}$. The distribution of a nonzero mask is arbitrary and fixed. The report is $y_j=s_j\oplus b_j$ and triggers the actual correction $R_{y_j}$. After round $T$, one trusted complete syndrome extraction measures the residual sector and applies its $X$ correction. Both the noise-aware and the common recovery receive the reported history and this final syndrome and use the same final CPTP recovery class. This cleanup costs one additional syndrome extraction. Tildes below refer to this noisy experiment and its own noise-aware benchmark.

\begin{proof}[Proof of \cref{prop:reporting-faults}]
After averaging the independent physical faults, scalar branch effects couple the experiment to i.i.d. increments $w_j$ of law $q^{(p_X)}$, independent of the masks and signs. The true syndrome, report, and outgoing residual sector obey
\begin{equation}
 s_j=w_j\oplus b_{j-1},\qquad
 y_j=w_j\oplus b_{j-1}\oplus b_j,\qquad s_{{\rm res},j}=b_j.
 \label{eq:reporting-increments}
\end{equation}
Relative to the correctly reported branch, the canonical logical map has the additional physical $X$-type normalizer $C_j$,
\begin{equation}
 C_j=R_{b_j}R_{w_j\oplus b_{j-1}\oplus b_j}R_{b_{j-1}}R_{w_j}.
 \label{eq:reporting-cocycle}
\end{equation}
Its restriction to the code space acts as $I$ or $\Xbar$ and commutes with the conditional logical phase channels. Final cleanup removes the final sector representative. The logical maps therefore differ only by the product of the $C_j$.

The factor $C_j$ is the identity if both adjacent masks vanish, or if $w_j=0$ and at most one adjacent mask is nonzero. A nontrivial factor consequently requires a nonzero increment next to a wrong report, or two adjacent wrong reports. Independence and a union bound give \cref{eq:reporting-budget-main}. Outside these events, the cleaned logical channels agree for every input, including an entangled input. Half the diamond distance between their record-discarded channels is at most $\epsilon_{\rm rep}$, so each fixed-action fidelity changes by at most that amount. This proves the first inequality in \cref{eq:reporting-main}.

The actual record $\widetilde y=(y_1,\ldots,y_T,b_T)$ has the same distribution under both models by \cref{eq:reporting-increments}. Conditioned on it, the cleaned channel is a mixture of single-axis phase channels and logical $X$ factors. Global sign symmetry makes its transverse multiplier $g_h(\widetilde y)$ real, with $|g_h|\le1$. The CPTP-to-Pauli reduction of \cref{app:regret} therefore applies to this actual instrument. If its fixed-correction infidelities are $\widetilde{\RefInf{h}}$, then
\[
 \mathbb E|g_{\mathsf F}+g_{\mathsf R}|
 \le\mathbb E(1+g_{\mathsf F})+\mathbb E(1-g_{\mathsf R})
 =2(\widetilde{\RefInf{\mathsf F}}+\widetilde{\RefInf{\mathsf R}}).
\]
Every common recovery has summed fidelity at most $1+\widetilde{\RefInf{\mathsf F}}+\widetilde{\RefInf{\mathsf R}}$, whereas the two noise-aware optima sum to at least $2-\widetilde{\RefInf{\mathsf F}}-\widetilde{\RefInf{\mathsf R}}$. Their maximum fidelity loss is at least $[1/2-\widetilde{\RefInf{\mathsf F}}-\widetilde{\RefInf{\mathsf R}}]_+$. Using $\widetilde{\RefInf{h}}\le\RefInf[(p_X)]{h}+\epsilon_{\rm rep}$ proves the second inequality in \cref{eq:reporting-main}.
\end{proof}

\subsection{A distance-three weak-noise regime}

For any fixed code and lookup table in \cref{def:family} with $d_X=3$, let $p_X=\kappa_X\theta^2$ and $p_{\rm rep}=\kappa_{\rm rep}\theta^2$ for fixed $\kappa_X,\kappa_{\rm rep}\ge0$. Then
\begin{equation}
 \begin{aligned}
 1-A_{p_X}&=O(\theta^4),\\
 B_{p_X}&=b_{\C}\theta^3+O(\theta^5),\\
 q_{\rm nz}^{(p_X)}&=O(\theta^2).
 \end{aligned}
 \label{eq:reporting-distance-three-orders}
\end{equation}
To see the first order, note that zero physical faults give logical error $O(\theta^4)$. One fault has probability $O(p_X)$ and needs at least one coherent error to escape correction, which contributes $O(p_X\theta^2)$. Two or more faults have probability $O(p_X^2)$. For $B_{p_X}$, each fixed-fault cross term pairs coherent supports in different logical classes, whose total weight is at least three. Nonzero fault patterns change its leading coefficient only at $O(p_X\theta^3)$, and \cref{eq:orders} supplies the nonzero zero-fault coefficient. A nonzero syndrome requires a coherent error or a physical fault, proving the final order.

For sufficiently small $\theta$, take
\begin{equation}
 T_{\pi,p_X}\coloneqq \operatorname{round}\frac{\pi}{|\arg(A_{p_X}-iB_{p_X})|}
           =\Theta(\theta^{-3}).
 \label{eq:reporting-noisy-horizon}
\end{equation}
Since $A_{p_X}\le|A_{p_X}-iB_{p_X}|\le1$, the same phase-rounding estimate as for \cref{eq:defects} gives $\RefInf[(p_X)]{h}(T_{\pi,p_X})=O(\theta)$. Equation~\eqref{eq:reporting-budget-main} gives $\epsilon_{\rm rep}(T_{\pi,p_X})=O(\theta)$ as well. Hence $\widetilde{\FC{h}}=1-O(\theta)$ and $\widetilde{\Fmm}=1/2-O(\theta)$. Complete logical $X$ dephasing supplies the universal upper bound $1/2$ on the latter. The noise-aware benchmark throughout knows F or R and the specified fault laws, while receiving the same noisy record and cleanup as the common recovery.

\section{A simultaneous environment-qubit extension}
\label{app:quantum}

\subsection{Common interaction and scalar effects}
Set $H\coloneqq \sum_jX_j$ and $\Omega\coloneqq (\theta^2H^2+\omega^2I)^{1/2}$. Functional calculus gives
\begin{equation}
 U(\theta,\omega)=\cos\Omega\otimes I_E
 -i\omega\operatorname{sinc}(\Omega)\otimes X_E
 -i\theta H\operatorname{sinc}(\Omega)\otimes Z_E,
 \label{eq:quantum-functional}
\end{equation}
where $\operatorname{sinc}x\coloneqq \sin x/x$ and $\operatorname{sinc}0\coloneqq 1$. In the generator of \cref{eq:quantum-unitary}, coupling and precession act simultaneously throughout each noise interval.

In the physical $X$-string expansion $U(\theta,\omega)=\sum_PP\otimes M_P$, each $P$ acts on $\mathcal H_{\rm phys}$ and $M_P\in\mathcal B(\mathbb C^2)$ acts on the environment qubit. With real coefficients $\eta_{0,P},\eta_{x,P},\eta_{z,P}$, even and odd strings have
\begin{equation}
 M_P=\begin{cases}
 \eta_{0,P}I_E-i\eta_{x,P}X_E,& |P|\text{ even},\\
 -i\eta_{z,P}Z_E,& |P|\text{ odd},
 \end{cases}
 \label{eq:quantum-word}
\end{equation}
Indeed, real functions of $H^2$ contain only even strings, whereas multiplication by $H$ gives odd strings. Thus $M_P^\dagger M_P=\varpi_PI_E$, with $\varpi_P\coloneqq \eta_{0,P}^2+\eta_{x,P}^2$ for even $|P|$ and $\varpi_P\coloneqq\eta_{z,P}^2$ for odd $|P|$, and $\sum_P\varpi_P=1$. The retained model's Pauli-history probability is
\begin{equation}
 \Tr\!\left(M_{P_T}\cdots M_{P_1}\rho_E
             M_{P_1}^\dagger\cdots M_{P_T}^\dagger\right)
 =\prod_{t=1}^T\varpi_{P_t},
 \label{eq:quantum-history}
\end{equation}
which equals the reset model's probability.

For syndrome $s$, the opposite parities of the two corrected logical classes give the one-cycle joint branch $\widetilde K_s\in\mathcal B(\mathcal H_L\otimes\mathbb C^2)$,
\begin{equation}
 \widetilde K_s=I_L\otimes C_{s,0}+\Xbar\otimes C_{s,1},
 \label{eq:quantum-branch}
\end{equation}
where $C_{s,0},C_{s,1}\in\mathcal B(\mathbb C^2)$ act only on the environment, and $(C_{s,0},C_{s,1})$ is either $(\eta_{0,s}I-i\eta_{x,s}X_E,-i\eta_{z,s}Z_E)$ or the reversed pair, with real coefficients. Anticommutation of $X_E$ and $Z_E$ cancels the cross term, so
\begin{equation}
 \widetilde K_s^\dagger\widetilde K_s=(\eta_{0,s}^2+\eta_{x,s}^2+\eta_{z,s}^2)I_{LE}=q_sI_{LE}.
 \label{eq:quantum-scalar}
\end{equation}
Here $I_{LE}$ is the identity on $\mathcal H_L\otimes\mathbb C^2$. The scalar effect holds even after the environment has become entangled with the logical system or a private controller register. Applying it successively proves the equality of syndrome distributions stated in \cref{sec:quantum-extension}. These equalities compare retention and reset at the same parameters. Their common probabilities $q_s$ can depend on $\omega$ and reduce to those of \cref{eq:branch-form} at $\omega=0$.

To express both models with one known dilation for a fixed storage time, append a hidden classical flag, an initially maximally mixed active environment register, and fresh maximally mixed environment registers. At every interval boundary a known controlled swap exchanges the active register with the next fresh register only in the reset flag state. Used registers remain inaccessible. The initial flag is the only model-dependent preparation. This realizes the two recorded experiments with identical known global interactions and different initial environment states, as required by the framework.

\subsection{An explicit quantum-memory witness}
The Choi operator of a classical-memory process is separable across whole channel links: it is a convex combination of tensor products of positive operators, one for each link, together with an initial-state factor. A negative partial transpose across groups of whole links is therefore a sufficient witness of quantum memory \cite{giarmatzi2021}. We use the normalized Pauli-Choi vectors $|P\rangle\!\rangle\coloneqq (P\otimes I)|\Omega_{2^n}\rangle$ of physical Paulis $P$ and omit the independent initial system state. Projections within individual links preserve this separability. The projection below serves only to test the mathematical process tensor and is not performed in the data circuit.

Label two physical Paulis by $0$ and $1$ (below, $I$ and $X_1$), and suppose their coefficients acting on the environment are
\begin{equation}
 M_0=\eta_0I-i\eta_xX_E,\qquad M_1=-i\eta_zZ_E,
 \quad \varpi_0\coloneqq \eta_0^2+\eta_x^2.
 \label{eq:quantum-witness-coefficients}
\end{equation}
After projecting each of four links onto these two labels, the retained-environment process has coefficient matrix
\begin{equation}
 Q_{\mathbf b,\mathbf c}=\frac12\Tr\!\left[
 M_{b_4}\cdots M_{b_1}
 (M_{c_4}\cdots M_{c_1})^\dagger\right].
 \label{eq:quantum-gram}
\end{equation}
Here $\mathbf b,\mathbf c\in\{0,1\}^4$. Its trace is $(\varpi_0+\eta_z^2)^4$, and normalization is immaterial to the sign of a partial transpose. On the three histories $0000,0101,1111$, the partial transpose of the first two links is the matrix
\begin{equation}
 \begin{pmatrix}
 \varpi_0^4 & \eta_z^2\varpi_0 d_0 & \eta_z^4\varpi_0^2\\
 \eta_z^2\varpi_0 d_0 & \eta_z^4\varpi_0^2 & \eta_z^6\varpi_0\\
 \eta_z^4\varpi_0^2 & \eta_z^6\varpi_0 & \eta_z^8
 \end{pmatrix},
 \quad d_0\coloneqq \eta_0^4-6\eta_0^2\eta_x^2+\eta_x^4.
 \label{eq:quantum-minor}
\end{equation}
This follows directly from $M_0M_0^\dagger=\varpi_0I$, $M_1^2=-\eta_z^2I$, and $M_1M_0=M_0^\dagger M_1$. For example, the first off-diagonal entry is $\eta_z^2\varpi_0\Tr(M_0^4)/2=\eta_z^2\varpi_0d_0$. The vector
\begin{equation}
 |v\rangle\coloneqq \eta_z^4|0000\rangle
       +\eta_z^2\varpi_0|0101\rangle-2\varpi_0^2|1111\rangle
\end{equation}
therefore satisfies
\begin{equation}
 \langle v|Q^{T_{12}}|v\rangle
       =-16\eta_0^2\eta_x^2\eta_z^8\varpi_0^2<0
 \label{eq:quantum-negative-witness}
\end{equation}
whenever $\eta_0\eta_x\eta_z\ne0$. A classical-memory process would remain separable under these linkwise projections, so the original physical process requires quantum memory.

For \cref{eq:quantum-unitary}, choose physical Pauli strings $I$ and $X_1$. Let $x_j$ be independent uniform signs, $m\coloneqq \sum_jx_j$, and $\Omega_m\coloneqq (\theta^2m^2+\omega^2)^{1/2}$, the eigenvalue of $\Omega$ on the sector with $\sum_jx_j=m$. The coefficients of $I$ and $X_1$ are
\begin{align}
 \eta_0&=\mathbb E\cos\Omega_m,&
 \eta_x&=\omega\mathbb E\operatorname{sinc}\Omega_m,\notag\\
 \eta_z&=\frac\theta n\mathbb E[m^2\operatorname{sinc}\Omega_m].
 \label{eq:quantum-witness-values}
\end{align}
The last identity uses permutation symmetry, $\mathbb E[x_1m\mid m]=m^2/n$. If $\theta>0$, $\omega\ne0$, and $\sqrt{n^2\theta^2+\omega^2}<\pi/2$, then $\eta_0,\eta_z>0$ and $\eta_x$ has the sign of $\omega$. This proves the quantum-memory witness stated in \cref{sec:quantum-extension}. The reset-environment process has a product realization across links, whereas the retained-environment process has no classical-memory realization, even one with history-dependent instruments.

\subsection{Recorded logical channels and continuity of recovery loss}
For $q_s>0$ define the environment operators $W_{s,\pm},W_\pm(y)\in\mathcal B(\mathbb C^2)$ by
\begin{equation}
 \begin{gathered}
 W_{s,\pm}\coloneqq (C_{s,0}\pm C_{s,1})/\sqrt{q_s},\\
 W_\pm(y)\coloneqq W_{s_T,\pm}\cdots W_{s_1,\pm}.
 \end{gathered}
\end{equation}
Both $W_{s,\pm}$ are unitary with determinant one. For an even-weight correction from the lookup table they have the form $(\eta_{0,s}I-i(\eta_{x,s}X_E\pm\eta_{z,s}Z_E))/\sqrt{q_s}$. For an odd-weight correction, one acquires an overall minus sign, which preserves determinant one in dimension two. The maximally mixed initial environment state therefore gives real multipliers
\begin{align}
 g_{\mathrm{Ret}_\omega}(y)
 &=\tfrac12\Tr[W_+(y)W_-(y)^\dagger],\notag\\
 g_{\mathrm{Res}_\omega}(y)
 &=\prod_{t=1}^T\tfrac12\Tr[W_{s_t,+}W_{s_t,-}^\dagger].
 \label{eq:quantum-logical-multipliers}
\end{align}
Thus the CPTP-to-Pauli reduction of \cref{app:regret} applies here as well. The reality of these multipliers uses the maximally mixed initial environment state, whereas the scalar-effect argument applies to any initial environment state. All $X$ corrections in the fixed lookup table commute with the joint unitary, so \cref{lem:frame-defer} also continues to apply.

Duhamel's formula gives the operator-norm estimate
\begin{equation}
 \|U(\theta,\omega)-U(\theta,0)\|_\infty\le|\omega|.
 \label{eq:quantum-duhamel}
\end{equation}
For either model, telescoping over $T$ intervals bounds half the diamond distance between its complete recorded instruments at $\omega$ and zero by $T|\omega|$. Fixed recovery fidelities and their optimized values each change by at most this amount. Hence the minimax fidelity losses of the two model pairs, $\Fmm[\mathrm{Ret}_\omega,\mathrm{Res}_\omega](T)$ and $\Fmm[\mathsf F,\mathsf R](T)=\Fmm(T)$, with the same full-record final CPTP action set, satisfy
\begin{equation}
 \left|\Fmm[\mathrm{Ret}_\omega,\mathrm{Res}_\omega](T)
             -\Fmm[\mathsf F,\mathsf R](T)\right|
 \le 2T|\omega|.
 \label{eq:quantum-continuity}
\end{equation}
The operations within a run may correlate the environment and the system. Contractivity and the hybrid argument include those correlations.

For a concrete conservative bound, the fixed-correction infidelities of Shor at $\theta=0.1$, $T=267$ give $[1/2-\RefInf{\mathsf F}-\RefInf{\mathsf R}]_+ =0.120386047734\ldots$. Equation~\eqref{eq:quantum-continuity} thus gives at least $0.066986047734\ldots$ at $\omega=10^{-4}$. These parameters also satisfy the four-link witness condition. The argument establishes the coexistence of quantum memory, identical observed statistics, and nonzero recovery loss at a fixed storage time. Since the continuity bound grows with $T$, extending the weak-noise cost results to fixed nonzero $\omega$ requires a separate argument. The witness applies to the physical process before QEC. Contraction with QEC may remove its temporal entanglement.

\section{Calibration with retained quantum output}
\label{app:calibration}

This appendix proves \cref{thm:costs} for the access classes in \cref{sec:access}. Every channel norm includes an arbitrary reference system. CPTP maps on enlarged classical--quantum registers represent measurement outcomes and feed-forward. The bounds therefore cover calibration procedures that retain quantum output until recovery.

\subsection{One-cycle and two-interval sensitivities}

For the logical branch $K_s(\xi)=u_sI-i\xi v_s\Xbar$, the sign-odd part of the corrected instrument is
\begin{equation}
 \frac{\mathcal E_+-\mathcal E_-}{2}(\rho)
 =\sum_s|s\rangle\langle s|\otimes
                  \bigl(-iu_sv_s[\Xbar,\rho]\bigr).
 \label{eq:cal-sign-odd}
\end{equation}
\Cref{app:retained-coefficient} shows that this map has diamond norm $2\sum_s|u_sv_s|$, which proves the equality in \cref{eq:beta}, and evaluates its leading coefficient.

For a two-interval query, let $\mathcal U_\pm\coloneqq \operatorname{Ad}_{U_\pm(\theta)}$ be the physical one-interval channels for the two signs, and write $\mathcal U_{\rm av}\coloneqq (\mathcal U_++\mathcal U_-)/2$ and $\mathcal U_{\rm odd}\coloneqq (\mathcal U_+-\mathcal U_-)/2$. Before the common final extraction, the $\mathsf F/\mathsf R$ channel difference is exactly
\begin{equation}
 \tfrac12(\mathcal U_+^2+\mathcal U_-^2)-\mathcal U_{\rm av}^2
       =\mathcal U_{\rm odd}^2.
 \label{eq:cal-two-difference}
\end{equation}
With $H=\sum_{j=1}^nX_j$, the physical generator obeys $\|{-i}[H,\cdot]\|_\diamond\le2n$, and hence $\|\mathcal U_{\rm odd}\|_\diamond\le2n\theta$. Preparation, extraction, correction, and disposal of inaccessible systems are common CPTP maps, so
\begin{equation}
 \tfrac12\|\mathcal Q_{\mathsf F,2}
                 -\mathcal Q_{\mathsf R,2}\|_\diamond
       \le2n^2\theta^2.
 \label{eq:cal-two-bound}
\end{equation}

\begin{lemma}[Whole calibration output]
\label{lem:hybrid}
The full retained calibration states of every protocol in \cref{sec:access} satisfy \cref{eq:hybrid}.
\end{lemma}
\begin{proof}
For each ordinary $\mathsf F$ run, sample the fixed sign in advance and store it in an inaccessible classical register. Replace its complete QEC cycles one at a time by $\mathcal E_0\coloneqq (\mathcal E_++\mathcal E_-)/2$. Immediately before a replacement, the joint state, including the hidden registers, has the form $\sum_{\boldsymbol \xi}|\boldsymbol \xi\rangle\langle\boldsymbol \xi|\otimes\varrho_{\boldsymbol \xi}$, where $\boldsymbol\xi$ lists the stored signs, $\varrho_{\boldsymbol \xi}\succeq0$, and $\sum_{\boldsymbol \xi}\Tr\varrho_{\boldsymbol \xi}=1$. Let $\xi_j$ be the stored sign of the run containing the call being replaced. Changing the call from $\mathcal E_{\xi_j}$ to $\mathcal E_0$ changes the trace norm by at most
\begin{equation}
 \sum_{\boldsymbol \xi}\left\|
   [ (\mathcal E_{\xi_j}-\mathcal E_0)\otimes\id](\varrho_{\boldsymbol \xi})
                      \right\|_1
 \le\beta_{\rm ret}.
 \label{eq:cal-hidden-hybrid}
\end{equation}
All subsequent operations are CPTP on the enlarged state and contract trace distance, including operations controlled by a hidden sign. Thus the $m$ ordinary calls contribute at most $m\beta_{\rm ret}/2$. The proof allows correlations between the accessible system and the fixed sign accumulated during a run.

Each delayed-extraction query has a fresh source, independent of its input conditional on the model. It is therefore a well-defined channel on that input. Replacing its $\mathsf F$ channel by its $\mathsf R$ channel and applying \cref{eq:cal-two-bound} costs at most $2n^2\theta^2$ in trace distance per query. Common adaptive boundary operations and a private quantum register do not increase any replacement cost. After all replacements the accessible output is exactly the $\mathsf R$ output. The triangle inequality proves \cref{eq:hybrid}.
\end{proof}

\subsection{Joint recovery as a consequence of the decision bound}

\begin{corollary}[Joint calibration-assisted recovery]
\label{lem:joint}
If the retained calibration state $\sigma_h$ is independent of the environment prepared for the data run, conditional on $h$, every common final CPTP processor obeys \cref{eq:joint}.
\end{corollary}
\begin{proof}
It suffices to bound the uncalibrated conflict in \cref{eq:general-conflict}. The two fixed endpoint corrections give
\begin{equation}
 \FC{\mathsf F}+\FC{\mathsf R}
 \ge 2-\RefInf{\mathsf F}-\RefInf{\mathsf R}.
 \label{eq:cal-oracle-score}
\end{equation}
For any common full-record recovery, the scores on the ideal endpoints $\operatorname{Ad}_{\Xbar}$ and $\id$ sum to at most one, by the orthogonality of their Choi vectors. The syndrome law is common to both models, and half the history-averaged diamond distance of each actual conditional channel from its endpoint is exactly $\RefInf{h}$. Thus
\begin{equation}
 \Fid{\mathsf F}(\mathcal R)+\Fid{\mathsf R}(\mathcal R)
 \le 1+\RefInf{\mathsf F}+\RefInf{\mathsf R}.
 \label{eq:cal-common-score}
\end{equation}
Subtracting these bounds and dividing by two gives $\Gamma_{\mathsf F\mathsf R}^{\mathrm{CPTP}} \ge 1/2-\RefInf{\mathsf F}-\RefInf{\mathsf R}$. Applying \cref{prop:calibration-loss} proves \cref{eq:joint}, including joint processors that retain quantum calibration output.
\end{proof}

\subsection{Achievability and completion of the cost theorem}

No weight-one $X$ error has zero syndrome: an $X$ stabilizer of that weight is excluded by parity, and a logical $X$ of that weight is excluded by $d_X\ge3$. If two distinct weight-one errors have the same syndrome, their product is a weight-two $X$ stabilizer and acts as $+1$ on the specified code sector. The probability of a nonzero syndrome after the rotation $U_+(x)$ acts on the code space is therefore an even analytic function
\begin{equation}
 p_{\C}(x)=\kappa_{\C}x^2+O_{\C}(x^4),\qquad
 \kappa_{\C}=\sum_{s\ne0}(n_s^{(1)})^2>0.
 \label{eq:cal-event-onset}
\end{equation}
Commutativity of the physical rotations proves \cref{eq:delayed-mixture}, hence $p_{\rm nz}=p_{\C}(2\theta)$ and \cref{eq:delayed-rate}.

For Bernoulli probabilities $p_{\rm nz}$ and $p_{\rm nz}/2$, the single-query Bhattacharyya coefficient is
\begin{equation}
 \begin{split}
 \operatorname{BC}
 &=\frac{p_{\rm nz}}{\sqrt2}
       +\sqrt{(1-p_{\rm nz})(1-p_{\rm nz}/2)}\\
 &=1-\left(\frac34-\frac1{\sqrt2}\right)p_{\rm nz}
       +O(p_{\rm nz}^2).
 \end{split}
 \label{eq:cal-binary-overlap}
\end{equation}
For $N$ independent queries, the sum of the two conditional errors of the equal-prior likelihood-ratio test is the distributional overlap, which is at most $\operatorname{BC}^N$. Thus each conditional error probability $e_h$, the probability that the test selects the other model when $h$ holds, is at most $\exp[-c_{\C}N\theta^2]$ for some fixed $c_{\C}>0$ and all sufficiently small $\theta$. Applying the fixed correction of the selected model in the independent data run then gives a loss of at most $\max_h\RefInf{h}+\max_h e_h$. Taking $e_h\le\epsilon/2$ proves the upper bound for $\mathcal A_{\rm del}$. The logical experiment using complete QEC cycles has $e_h=\RefInf{h}$, proving its upper bound with cost $T_\pi$.

Conversely, \cref{eq:joint} requires
\begin{equation}
 \Dtr(\sigma_{\mathsf F},\sigma_{\mathsf R})
       \ge1-2\epsilon-2\RefInf{\mathsf F}-2\RefInf{\mathsf R}
 \label{eq:cal-required-distance}
\end{equation}
for any procedure attaining the target. Since $\RefInf{h}=O_{\C}(\theta)$, the right-hand side is bounded below by a positive constant for small enough $\theta$. \Cref{eq:hybrid} and $\beta_{\rm ret}=\Theta_{\C}(\theta^{d_X})$ then give the lower bounds in \cref{eq:costs}. For syndrome-only calibration the distance is exactly zero at every budget, which proves the infinite-cost assertion at every sufficiently small angle.

\subsection{Explicit calibration prefactors}\label{app:cost-prefactors}

For calibration using complete QEC cycles, each cycle costs one physical noise interval, so its interval budget $M$ is also its cycle count. The data run keeps the identity schedule at $T_\pi$, and its environment is prepared independently of calibration given $h\in\{\mathsf F,\mathsf R\}$. Combining the full-output hybrid bound \eqref{eq:hybrid} with the joint-processor bound \eqref{eq:joint} gives
\begin{equation}
 \max_h[\FC{h}-\Fid{h}(\mathcal G)]
 \ge\left[\frac12-\frac{M\beta_{\rm ret}}4
              -\RefInf{\mathsf F}-\RefInf{\mathsf R}\right]_+.
 \label{eq:retained-cycle-loss}
\end{equation}
The fixed-correction infidelities are $\RefInf{\mathsf R}=(1-A^{T_\pi})/2$ and $\RefInf{\mathsf F}=(1+\operatorname{Re}[(A-iB)^{T_\pi}])/2$. They upper-bound the optimized noise-aware infidelities. For fixed $0<\epsilon<1/2$, every allowed protocol attaining expected fidelity loss at most $\epsilon$ must therefore satisfy
\begin{equation}
 M^*_{\mathcal A_{\rm log},\epsilon}
 \ge\frac{2}{\beta_{\rm ret}}
      [1-2\epsilon-2\RefInf{\mathsf F}-2\RefInf{\mathsf R}]_+.
 \label{eq:retained-prefactor-lower}
\end{equation}
This lower bound includes adaptive logical boundary controls, retained quantum calibration memory, and joint final processing with the data, under the prescribed fixed-budget and source-reset conventions.

For the upper bound, prepare a $\Zbar=+1$ state, run $T_\pi$ complete QEC cycles, and measure $\Zbar$. Assign outcome $+1$ to $\mathsf R$ and $-1$ to $\mathsf F$. The conditional identification errors are exactly $\RefInf{\mathsf R}$ and $\RefInf{\mathsf F}$. After the independent data run, apply $I$ for decision $\mathsf R$ and $\Xbar$ for decision $\mathsf F$. A correct decision has fixed-correction infidelity $\RefInf{h}$, and a wrong decision has infidelity at most one. The expected fidelity loss is therefore at most $2\RefInf{h}$ under either model. Since $\RefInf{h}=O_{\C}(\theta)$, one calibration run suffices whenever $2\max_h\RefInf{h}\le\epsilon$. For all sufficiently small $\theta$,
\begin{equation}
 M^*_{\mathcal A_{\rm log},\epsilon}\le T_\pi.
 \label{eq:retained-prefactor-upper}
\end{equation}

Finally, $A\to1$, $B\to0$, and the definition of $T_\pi$ give $T_\pi|B|\to\pi$. Equations \eqref{eq:retained-leading-ratio}, \eqref{eq:retained-prefactor-lower}, and \eqref{eq:retained-prefactor-upper} imply
\begin{align}
 \frac{2(1-2\epsilon)}{\pi}\frac{\binom{2r}{r}}{4^r}
 &\le\liminf_{\theta\to0^+}
       \frac{M^*_{\mathcal A_{\rm log},\epsilon}}{T_\pi}\notag\\
 &\le\limsup_{\theta\to0^+}
       \frac{M^*_{\mathcal A_{\rm log},\epsilon}}{T_\pi}
 \le1.
 \label{eq:retained-prefactor-bracket}
\end{align}
For $d_X=3$, the lower constant is $(1-2\epsilon)/\pi$. Equivalently,
\begin{equation}
 \left[\frac{1-2\epsilon}{\pi}-O_{\C,\epsilon}(\theta)\right]T_\pi
 \le M^*_{\mathcal A_{\rm log},\epsilon}\le T_\pi.
 \label{eq:retained-distance-three-bracket}
\end{equation}
The optimal prefactor lies in this explicit interval. All limits fix the code, the lookup table, and the target fidelity loss. Their finite-angle remainder constants can depend on the code.

\section{Calibration under a maximum extraction delay}
\label{app:cadence}

We prove \cref{thm:delay-cost} uniformly in the delay cap $\ell$. A protocol has $m$ complete QEC cycles and independent delayed-extraction queries of prescribed lengths $\lambda_1,\ldots,\lambda_N\in\{2,\ldots,\ell\}$, with total budget $M=m+\sum_j\lambda_j$.

\subsection{Uniform complete-output lower bound}

Consider a length-$\lambda$ query with physical channels $\mathcal U_\xi(x)\coloneqq \operatorname{Ad}_{e^{-i\xi xH}}$, where $H=\sum_jX_j$. Freeze all controls at their chosen design values and vary only an auxiliary amplitude $x$. With common CPTP maps $\mathcal B_j$ allowed between intervals, the channel conditional on the sign string is
\begin{equation}
 \mathcal V_{\boldsymbol \xi}(x)
 =\mathcal B_\lambda\circ\mathcal U_{\xi_\lambda}(x)\circ\cdots\circ
        \mathcal B_1\circ\mathcal U_{\xi_1}(x)\circ\mathcal B_0.
 \label{eq:delay-conditional-channel}
\end{equation}
The intermediate maps $\mathcal B_j$, $1\le j<\lambda$, act on the physical register together with all retained classical records and private quantum registers. Each noise channel acts as the identity on these ancillary registers. The boundary maps may encode the logical input and perform the final extraction and correction. Thus the enlarged query channel retains every allowed output register. Averaging $\mathcal V_{\boldsymbol \xi}(\theta)$ over the sign law of model $h$ gives the query channel $\mathcal Q_{h,\lambda}(\theta)$, which for $\lambda=2$ and no intermediate controls is \cref{eq:delayed-mixture}. Derivative bounds on each noise channel are $\|\mathcal U_\xi'(x)\|_\diamond\le2n$ and $\|\mathcal U_\xi''(x)\|_\diamond\le4n^2$. The product rule gives
\begin{equation}
 \|\mathcal V_{\boldsymbol \xi}''(x)\|_\diamond\le4n^2\lambda^2.
 \label{eq:delay-second-derivative}
\end{equation}
At $x=0$ all sign strings give the same channel, and the first derivative is a sum of common maps multiplied by the individual signs. Both models have zero mean for each sign, so their averaged channels agree to first order. Taylor's integral formula, together with the trivial bound of one, yields
\begin{equation}
 \begin{split}
 &\tfrac12\|\mathcal Q_{\mathsf F,\lambda}(\theta)
                  -\mathcal Q_{\mathsf R,\lambda}(\theta)\|_\diamond\\
 &\quad\le \delta_\lambda(\theta)\coloneqq \min\{1,2n^2\lambda^2\theta^2\},\quad \lambda\ge2.
 \end{split}
 \label{eq:delay-query-bound}
\end{equation}
At $\lambda=1$ the two channels agree exactly. The bound does not assume $\lambda\theta\ll1$, and permitting intermediate controls only enlarges the class against which the lower bound holds.

Replacing ordinary calls as in \cref{lem:hybrid}, and each fresh delayed-extraction query as a whole channel, shows that the complete calibration output satisfies
\begin{equation}
 \Dtr(\sigma_{\mathsf F},\sigma_{\mathsf R})
 \le\min\left\{1,\frac{m\beta_{\rm ret}}2
                         +\sum_j \delta_{\lambda_j}(\theta)\right\}.
 \label{eq:delay-hybrid}
\end{equation}
Fresh source resets justify the channel replacement even when the input is entangled with a quantum register shared with other runs. For $2\le \lambda\le\ell$,
\begin{align}
 \delta_\lambda/\lambda&\le2n^2\lambda\theta^2\le4n^2(\ell-1)\theta^2,\notag\\
 \delta_\lambda/\lambda&\le\min\{1/\lambda,2n^2\lambda\theta^2\}\le\sqrt2 n\theta.
\end{align}
Writing $\varsigma_\ell(\theta)\coloneqq \min\{(\ell-1)\theta^2,\theta\}$ gives
\begin{equation}
 \Dtr(\sigma_{\mathsf F},\sigma_{\mathsf R})
 \le M\max\{\beta_{\rm ret}/2,4n^2\varsigma_\ell(\theta)\}.
 \label{eq:delay-budget-bound}
\end{equation}
The delayed term is absent at $\ell=1$. Combining this with \cref{eq:cal-required-distance} and $\beta_{\rm ret}=O_{\C}(\theta^{d_X})$ proves the lower bound in \cref{eq:delay-cost}, with constants and an angle threshold independent of $\ell$.

\subsection{A uniform syndrome test}

For the allowed delayed-extraction queries, commutativity gives the exact nonzero-syndrome probabilities
\begin{align}
 p_{\mathsf F,\lambda}&=p_{\C}(\lambda\theta),\notag\\
 p_{\mathsf R,\lambda}&=\mathbb E p_{\C}(S_\lambda\theta)
 =2^{-\lambda}\sum_{j=0}^\lambda\binom{\lambda}{j}
                       p_{\C}((\lambda-2j)\theta),
 \label{eq:delay-exact-probabilities}
\end{align}
where $S_\lambda$ is a sum of independent fair signs. Choose a fixed $x_{\C}>0$ and $C_{\C}<\infty$ so that
\begin{equation}
 |p_{\C}(x)-\kappa_{\C}x^2|
 \le C_{\C}x^4\le\tfrac18\kappa_{\C}x^2
       \quad (|x|\le x_{\C}).
 \label{eq:delay-local-remainder}
\end{equation}
For $2\le \lambda\le x_{\C}/\theta$, the moments $\mathbb E S_\lambda^2=\lambda$ and $\mathbb E S_\lambda^4=3\lambda^2-2\lambda$ give
\begin{align}
 |p_{\mathsf F,\lambda}-\kappa_{\C}\lambda^2\theta^2|
   &\le C_{\C}\lambda^4\theta^4,\notag\\
 |p_{\mathsf R,\lambda}-\kappa_{\C}\lambda\theta^2|
   &\le C_{\C}(3\lambda^2-2\lambda)\theta^4.
 \label{eq:delay-remainders}
\end{align}
The relative form of \cref{eq:delay-local-remainder} gives uniformly
\begin{equation}
 p_{\mathsf F,\lambda}\ge\tfrac78\kappa_{\C}\lambda^2\theta^2,
 \qquad \frac{p_{\mathsf R,\lambda}}{p_{\mathsf F,\lambda}}
          \le\frac{9}{7\lambda}\le\frac9{14}.
 \label{eq:delay-contrast}
\end{equation}
The Bhattacharyya coefficient of the two Bernoulli distributions obeys
\begin{align}
 1-\operatorname{BC}_\lambda
 &\ge\tfrac12(\sqrt{p_{\mathsf F,\lambda}}-
                     \sqrt{p_{\mathsf R,\lambda}})^2\notag\\
 &\ge c_{\C}^{\rm test}\lambda^2\theta^2,\qquad
 c_{\C}^{\rm test}
   \coloneqq \frac{7\kappa_{\C}}{16}(1-\sqrt{9/14})^2>0.
 \label{eq:delay-test-rate}
\end{align}
After $N$ independent queries, each conditional error of the likelihood-ratio test is at most $\exp[-c_{\C}^{\rm test}N\lambda^2\theta^2]$. Thus a fixed error $\delta_{\rm err}>0$ requires at most
\begin{equation}
 N=\left\lceil\frac{\log(1/\delta_{\rm err})}
                   {c_{\C}^{\rm test}\lambda^2\theta^2}\right\rceil,
 \qquad N\lambda=O_{\C,\delta_{\rm err}}((\lambda\theta^2)^{-1}).
 \label{eq:delay-test-cost}
\end{equation}
The ceiling contributes at most $\lambda\le x_{\C}^2/(\lambda\theta^2)$.

For $\ell\ge2$, choose $\lambda=\min\{\ell,\lfloor x_{\C}/\theta\rfloor\}$ and $\theta\le x_{\C}/4$. This gives cost $O_{\C,\delta_{\rm err}}(\varsigma_\ell(\theta)^{-1})$, uniformly in $\ell$. Take $\delta_{\rm err}=\epsilon/2$ and decrease the angle threshold so that $\max_h\RefInf{h}\le\epsilon/2$. The induced recovery loss is then at most $\epsilon$. Since $\varsigma_\ell(\theta)\ge\theta^2\ge \theta^{d_X}$ for $\ell\ge2$ and small $\theta$, this upper bound has the form in \cref{eq:delay-cost}. At $\ell=1$, the logical endpoint experiment gives cost $T_\pi=O_{\C}(\theta^{-d_X})$. This completes the uniform proof.
\section{Ideal control sequences}\label{app:protection-comparison}
This appendix compares the original memory experiment with ideal physical randomization, logical Pauli twirling over $\{I,\Zbar\}$, and the midpoint echo sequence of logical dynamical decoupling. We record the one-cycle weights and randomized fidelity, then bound the conditional logical phases and derive the leading infidelities of these circuits. The final subsection proves the exact optimum within the logical-Pauli control class for the generalized Shor code.

For the $\mathsf F/\mathsf R$ pair, independent physical Pauli twirling gives independent physical $X$ errors with probability $\sin^2\theta$ per qubit and interval. Let $p_{\rm tw}$ be the one-cycle logical flip probability after the fixed lookup-table correction, and put $A_{\rm tw}\coloneqq1-2p_{\rm tw}$. The final identity correction achieves
\begin{equation}
 F_{\rm tw}(T)=\frac{1+A_{\rm tw}^T}{2}.
 \label{eq:twirl-fidelity}
\end{equation}
Under the code and minimum-weight-table assumptions of \cref{prop:protection-comparison}, the leading logical-flip contributions come from errors of weight exactly $r+1$. Their coherent amplitudes add before squaring, whereas physical randomization adds probabilities. With the counts $m_s$ defined after \cref{eq:retained-orders-main},
\begin{align}
 \frac{1-A}{2}&=\sum_s m_s^2\theta^{d_X+1}
                         +O_{\C}(\theta^{d_X+3}),\notag\\
 p_{\rm tw}&=\sum_s m_s\theta^{d_X+1}
                         +O_{\C}(\theta^{d_X+3}).
 \label{eq:twirl-leading}
\end{align}

\subsection{Conditional phases and moments}

\begin{lemma}[Small conditional phases]
\label{lem:minimum-weight-small-phase}
Assume \cref{def:family} and a minimum-weight lookup table in the sense of \cref{sec:regret}. Let $m_s$ be the counts defined in \cref{sec:visibility}, so that $\sum_sm_s^2>0$. For all sufficiently small $\theta>0$, the conditional rotation angles have representatives $\phi_s\in(-\pi,\pi)$ satisfying, for some constant $C'_{\C}<\infty$,
\begin{equation}
 \max_s|\phi_s|\le C'_{\C}\theta.
 \label{eq:small-phase-bound}
\end{equation}
Let $\phi_s$ take its value with probability $q_s=u_s^2+v_s^2$, and write $\mu_\theta\coloneqq \mathbb E\phi_s$ and $\sigma_\theta^2\coloneqq \operatorname{Var}(\phi_s)$. Then
\begin{align}
 \mu_\theta&=B+O_{\C}(\theta^{d_X+2}),\notag\\
 \mathbb E\phi_s^2
 &=4\sum_s m_s^2\theta^{d_X+1}+O_{\C}(\theta^{d_X+3}),\notag\\
 \sigma_\theta^2
 &=4\sum_s m_s^2\theta^{d_X+1}+O_{\C}(\theta^{d_X+3}),\notag\\
 \mathbb E|\phi_s-\mu_\theta|^3
 &=O_{\C}(\theta^{d_X+2}).
 \label{eq:small-phase-moments}
\end{align}
The physical representative chosen for a correction may be any member of the specified minimum-weight logical class.
\end{lemma}

\begin{proof}
For syndrome $s$, let $w_{s,0}$ and $w_{s,1}$ be the minimum physical weights in the classes corrected to $I$ and $\Xbar$, and let $c_{s,0},c_{s,1}$ count their minimum-weight representatives. The two classes have opposite weight parity. The hypothesis that the identity class contains a minimum-weight error therefore gives the strict inequality $w_{s,0}<w_{s,1}$. On the code space, each positive $X$ stabilizer acts as $+1$, so the amplitudes of equal-weight representatives add with the same phase. After a common branch phase is removed, their leading magnitudes are
\begin{equation}
 \begin{split}
 |u_s|&=c_{s,0}\theta^{w_{s,0}}
                  [1+O_{\C}(\theta^2)],\\
 |v_s|&=c_{s,1}\theta^{w_{s,1}}
                  [1+O_{\C}(\theta^2)].
 \end{split}
 \label{eq:small-phase-coset-orders}
\end{equation}
Thus $|v_s/u_s|=O_{\C}(\theta)$ for every syndrome. There are finitely many syndromes for the fixed code, so the bound is uniform. Choose $\phi_s=2\arctan(v_s/u_s)$, which differs from the original conditional angle by an integer multiple of $2\pi$. This proves \cref{eq:small-phase-bound} and preserves all conditional channels.

Because the table corrects every $X$ error of weight at most $r$, the logical class contains no error of weight at most $r$. Consequently $w_{s,1}\ge r+1$, and
\begin{equation}
 m_s=c_{s,1}\,\mathbf1_{\{w_{s,1}=r+1\}}.
 \label{eq:small-phase-count-convention}
\end{equation}
In particular, minimum-weight multiplicities with $w_{s,1}>r+1$ do not enter $\sum_s m_s^2$. The coherent expansion in \cref{eq:twirl-leading} gives
\begin{equation}
 \sum_s v_s^2=\sum_s m_s^2\theta^{d_X+1}
                        +O_{\C}(\theta^{d_X+3}).
 \label{eq:small-phase-flip-weight}
\end{equation}
The opposite coset parities exclude a contribution at degree $d_X+2$. Partitioning a minimum logical support into weights $r$ and $r+1$ shows that at least one $m_s$ is positive.

Set $x_s\coloneqq v_s/u_s$. Uniformly over the syndromes,
\begin{equation}
 \begin{aligned}
 q_s\phi_s^2&=4u_s^2(1+x_s^2)(\arctan x_s)^2\\
            &=4v_s^2[1+O_{\C}(\theta^2)].
\end{aligned}
 \label{eq:small-phase-second-moment}
\end{equation}
Summing proves the second-moment formula. Moreover, $\sum_sq_s\sin\phi_s=2\sum_su_sv_s=B$, and hence
\begin{equation}
 |\mu_\theta-B|
 \le\frac16\mathbb E|\phi_s|^3
 \le\frac16\max_s|\phi_s|\,\mathbb E\phi_s^2
 =O_{\C}(\theta^{d_X+2}).
 \label{eq:small-phase-mean-bound}
\end{equation}
Since $\mu_\theta=O_{\C}(\theta^{d_X})$ and $d_X\ge3$, $\mu_\theta^2=O_{\C}(\theta^{d_X+3})$, which gives the variance formula. Finally,
\begin{equation}
 \mathbb E|\phi_s-\mu_\theta|^3
 \le4\mathbb E|\phi_s|^3+4|\mu_\theta|^3
 =O_{\C}(\theta^{d_X+2}).
\end{equation}
\end{proof}

\subsection{Leading infidelities of the modified QEC circuits}

\begin{proof}[Proof of \cref{prop:protection-comparison}]
All comparisons use the same original storage time $T=T_\pi$ and the same fixed code. By \cref{eq:orders,eq:twirl-leading},
\begin{align}
 A&=1-2\sum_s m_s^2\theta^{d_X+1}+O_{\C}(\theta^{d_X+3}),\notag\\
 B&=b_{\C}\theta^{d_X}+O_{\C}(\theta^{d_X+2}),\notag\\
 T_\pi&=\frac{\pi}{|b_{\C}|\theta^{d_X}}
                   [1+O_{\C}(\theta^2)]+O(1).
 \label{eq:protection-endpoint-orders}
\end{align}

\paragraph{Full-record CPTP optima in the original experiment.} For the conditional $I/\Xbar$ mixtures considered here, \cref{eq:scaled-recovery-scores} gives the optimum over final CPTP recoveries, including ancillas. Applying it to each syndrome history yields
\begin{equation}
 \FC{h}(T)=\frac{1+\mathbb E_y|g_h(y)|}{2}.
 \label{eq:protection-full-record-optimum}
\end{equation}

For sufficiently small $\theta$, \cref{eq:small-phase-bound} implies $\cos\phi_s>0$ for every syndrome. Thus $g_{\mathsf R}(y)=\prod_j\cos\phi_{s_j}>0$ for every history, and
\begin{equation}
 \FC{\mathsf R}(T)=\frac{1+A^T}{2}
 \label{eq:protection-resampled-oracle}
\end{equation}
at every integer storage time. In particular, the final identity correction attains the full-record CPTP optimum under $\mathsf R$.

For $\mathsf F$ and a record $y=(s_1,\ldots,s_T)$ of length $T$, use the accumulated phase $\phi_{\rm tot}(y)=\sum_{j=1}^T\phi_{s_j}$. The syndrome law is the product $\prod_jq_{s_j}$, so these increments are independent even though the physical sign is fixed. The small-phase moment estimates and the definition of $T_\pi$ give
\begin{equation}
 \begin{split}
 T_\pi\mu_\theta
     &=\operatorname{sgn}(B)\pi+O_{\C}(\theta^2),\\
 \Sigma_\theta^2\coloneqq T_\pi\sigma_\theta^2&=O_{\C}(\theta).
 \end{split}
 \label{eq:protection-phase-concentration}
\end{equation}
For the following estimates, $y$ has length $T_\pi$. The centered increments are bounded by $L_\theta=O_{\C}(\theta)$. Bernstein's inequality gives, for fixed $x>0$,
\begin{equation}
 \Pr\{|\phi_{\rm tot}(y)-T_\pi\mu_\theta|\ge x\}
 \le2\exp\!\left[-\frac{x^2}{2(\Sigma_\theta^2+L_\theta x/3)}\right].
 \label{eq:protection-bernstein}
\end{equation}
For sufficiently small $\theta$, the event $\cos\phi_{\rm tot}(y)>0$ requires a deviation of at least $\pi/4$ from this mean. Hence, for some constant $c'_{\C}>0$,
\begin{equation}
 \begin{split}
 0\le\FC{\mathsf F}(T_\pi)
       -(1-\RefInf{\mathsf F})
 &=\mathbb E[(\cos\phi_{\rm tot}(y))_+]\\
 &\le2e^{-c'_{\C}/\theta}.
 \end{split}
 \label{eq:protection-frozen-oracle-gain}
\end{equation}
Thus the improvement of the full-record noise-aware optimum over fixed final $\Xbar$ is smaller than every power of $\theta$.

The endpoint expansion gives
\begin{equation}
 \begin{split}
 \frac{1-A^{T_\pi}}2
   &=\frac{\pi \sum_s m_s^2}{|b_{\C}|}\theta+O_{\C}(\theta^2),\\
 \frac{1+\operatorname{Re}[(A-iB)^{T_\pi}]}2
   &=\frac{\pi \sum_s m_s^2}{|b_{\C}|}\theta+O_{\C}(\theta^2).
 \end{split}
 \label{eq:protection-fixed-endpoint-losses}
\end{equation}
For the second line, $\lvert A+iB\rvert^2=A^2+B^2$ and $d_X\ge3$ imply $\lvert A+iB\rvert=1-2\sum_s m_s^2\theta^{d_X+1}+O_{\C}(\theta^{d_X+3})$. Rounding $T_\pi$ changes its accumulated phase from $\pi$ by $O_{\C}(\theta^{d_X})$. Combining these estimates proves
\begin{equation}
 1-\FC{h}(T_\pi)
 =\frac{\pi \sum_s m_s^2}{|b_{\C}|}\theta+o(\theta),
 \qquad h\in\{\mathsf F,\mathsf R\}.
 \label{eq:protection-oracle-leading}
\end{equation}

\paragraph{Active physical Pauli randomization.} By \cref{eq:twirl-fidelity,eq:twirl-leading}, the common final identity correction achieves
\begin{equation}
 1-F_{\rm tw}(T_\pi)
 =\frac{1-A_{\rm tw}^{T_\pi}}2
 =\frac{\pi \sum_s m_s}{|b_{\C}|}\theta+o(\theta).
 \label{eq:protection-physical-twirl-leading}
\end{equation}
Since each $m_s$ is a nonnegative integer, $\sum_sm_s\le\sum_sm_s^2$, with strict inequality exactly when some $m_s\ge2$. This compares the achieved infidelity of the randomized data circuit with the minimum infidelity attainable in the original circuit when the model and full syndrome record are available.

For completeness, the final identity correction is also optimal for the full syndrome record of this averaged randomized experiment at sufficiently small $\theta$. In each syndrome, the probabilities of its two corrected logical classes have leading orders $\theta^{2w_{s,0}}$ and $\theta^{2w_{s,1}}$. Since $w_{s,0}<w_{s,1}$, the conditional identity probability exceeds the logical-flip probability. Every composed conditional multiplier is consequently positive.

\paragraph{Logical Pauli twirling and a midpoint echo.} Independent uniform choices of conjugation by $I$ or $Z^{\otimes n}$ implement logical Pauli twirling over $\{I,\Zbar\}$ and change the sign in interval $j$ to $\eta_j\xi_j$, where $\eta_j=-1$ if $Z^{\otimes n}$ is applied and $\eta_j=+1$ otherwise. Averaging over the fresh independent $\eta_j$ gives the resampled sign law under either source model, including the syndrome output. After the known frame is compensated, the final identity correction has fidelity
\begin{equation}
 \Fid{\mathrm{log},h}(T)=\frac{1+A^T}{2}=\FC{\mathsf R}(T),
 \qquad h\in\{\mathsf F,\mathsf R\},
 \label{eq:protection-logical-average}
\end{equation}
for sufficiently small $\theta$. The protocol uses the record of applied Paulis for frame compensation and then discards it. The formula describes this averaged experiment. When the record is retained, the conditional channels need not equal the averaged one.

A midpoint logical $\Zbar$ pulse reverses the sign of the subsequent logical rotations. Compensating its final frame gives the mean readout under $\mathsf F$
\begin{equation}
 \langle\Zbar\rangle_{{\rm echo},\mathsf F}(T)
 =\begin{cases}
   \lvert A+iB\rvert^T,&T\text{ even},\\
   A\lvert A+iB\rvert^{T-1},&T\text{ odd},
  \end{cases}
 \label{eq:protection-echo-multiplier}
\end{equation}
where the two halves differ by one round in the odd case. Under $\mathsf R$ the multiplier remains $A^T$. For sufficiently small $\theta$, $0<A\le\lvert A+iB\rvert\le1$ and
\begin{equation}
 0\le\lvert A+iB\rvert-A=\frac{B^2}{\lvert A+iB\rvert+A}=O_{\C}(\theta^{2d_X}).
\end{equation}
The difference of two integer powers on $[0,1]$ is bounded by the exponent times the difference of their bases. Thus at $T=T_\pi$,
\begin{equation}
 \begin{aligned}
 0\le \langle\Zbar\rangle_{{\rm echo},\mathsf F}(T_\pi)-A^{T_\pi}
 &\le T_\pi(\lvert A+iB\rvert-A)\\
 &=O_{\C}(\theta^{d_X})=o(\theta).
 \end{aligned}
 \label{eq:protection-echo-error}
\end{equation}
The common final identity correction after either logical Pauli twirling or the midpoint echo therefore attains
\begin{equation}
 1-\Fid{\mathrm{control},h}(T_\pi)
 =\frac{\pi \sum_s m_s^2}{|b_{\C}|}\theta+o(\theta)
 \label{eq:protection-control-leading}
\end{equation}
under both models. All comparisons fix the code and use ideal additional controls at the same original storage time.
\end{proof}

\paragraph{Other $X$ correction tables.} For a table that chooses the opposite logical class in some syndromes, let $f(s)\in\{0,1\}$ record its difference from a table selecting the minimum-weight class. Its corrected branch is $\Xbar^{f(s)}K_s(\xi)$ up to a scalar phase. All the single-axis branches commute with $\Xbar$, so the final correction $\Xbar^{\sum_j f(s_j)}$ exactly removes the difference over a recorded history. The compensated echo and randomization protocols above preserve this identity: their logical-$Z$ pulses introduce only scalar commutation phases, and their effective noise branches remain single-axis. The same comparison is therefore achieved for another table by including this known final parity correction and using the counts of the minimum-weight logical classes. The raw weight-$(r+1)$ miscorrection counts of an arbitrary table can give a different coefficient when the final correction is fixed to identity.
\subsection{Exact midpoint echo for generalized Shor codes}
\label{app:logical-echo}

The midpoint echo is the specific logical dynamical-decoupling sequence considered here. Its exact optimality statement concerns the generalized Shor code of odd distance $d_X=2r+1$, with $d_X$ rows of $d_X$ qubits, $Z$ checks within rows, and adjacent products of full-row $X$ operators as $X$ checks. Use rowwise minimum-weight $X$ correction and the same common-sign physical rotations as in the main example. The allowed additional data controls are ideal, possibly syndrome-adaptive logical Pauli gates on $\mathcal H_L$ between complete QEC cycles, with compensation of their known final frame. The objective is worst-case absolute entanglement fidelity over $\mathsf F/\mathsf R$ and their mixtures, with a full-record final Pauli recovery.

For a row syndrome whose minimum correction has weight $w\le r$, the two corrected amplitudes, after removal of a common phase, are
\begin{equation}
 \begin{split}
 u_w&=(\cos\theta)^{d_X-w}(\sin\theta)^w,\\
 v_w&=(-1)^{r-w}(\cos\theta)^w(\sin\theta)^{d_X-w}.
 \end{split}
 \label{eq:echo-row-amplitudes}
\end{equation}
Its conditional phase is $\phi_w=2\arctan(v_w/u_w)$. If $0<\theta<\pi/(4d_X)$, then $|v_w/u_w|=(\tan\theta)^{d_X-2w}\le\tan\theta$, so $|\phi_w|\le2\theta$. The phases of the $d_X$ row branches add to the conditional phase $\phi_s$ of the round, with $|\phi_s|\le2d_X\theta<\pi/2$. The syndrome-averaged fixed-sign multiplier is $A-iB$, and its real part $A=\sum_sq_s\cos\phi_s>0$ is the coefficient of \cref{eq:AB} for this code and table.

Under $\mathsf R$, conditioning on any syndrome gives a real positive multiplier $\cos\phi_s$. A logical Pauli either preserves or reverses the phase axis and therefore leaves this averaged multiplier unchanged. Even adaptive logical-Pauli choices leave the compensated multiplier $\prod_j\cos\phi_{s_j}>0$. The final identity correction is optimal on every record, and the noise-aware $\mathsf R$ fidelity is exactly $(1+A^T)/2$. This bounds the worst-case performance of every allowed protocol.

For even $T$, a single logical $\Zbar$ pulse after round $T/2$ reverses the sign of the remaining accumulated phase. Under $\mathsf F$, averaging the independent syndrome outcomes gives
\begin{equation}
 (A-iB)^{T/2}(A+iB)^{T/2}=\lvert A+iB\rvert^T\ge A^T.
\end{equation}
Under $\mathsf R$ the multiplier remains $A^T$. Thus the midpoint echo and the final identity correction attain the upper bound, including all convex mixtures of $\mathsf F/\mathsf R$:
\begin{equation}
 \sup_{\mathfrak p}
       \inf_{0\le p\le1}\Fid{p\mathsf F+(1-p)\mathsf R}(\mathfrak p)
       =\frac{1+A^T}{2}.
 \label{eq:echo-exact-optimum}
\end{equation}
Here $\mathfrak p$ ranges over the allowed logical-Pauli protocols, $\Fid{h}(\mathfrak p)$ is the entanglement fidelity of $\mathfrak p$ under model $h$, including its controls and final Pauli recovery, and $p\mathsf F+(1-p)\mathsf R$ draws $\mathsf F$ with probability $p$.
Within this control class, calibration cannot improve the worst-case value. Physical Pauli dressing, general intermediate logical CPTP maps, and imperfect pulses lie outside the class, and the statement concerns worst-case absolute fidelity rather than the loss relative to each model's own optimal controlled circuit.
\section{Worst-case accuracy dependence}
\label{sec:accuracy-dependence}

We now keep the angle of the normalizer construction fixed and vary the desired accuracy. This asymptotic question is distinct from the fixed-target weak-angle limit of \cref{thm:costs}. The construction in \cref{eq:normalizer-construction} supplies an accuracy lower bound. Keep the same physical normalizer $P$, its logical representative $\bar P$, and fair sign $\xi$, but draw $c$ independently with
\begin{equation}
 \Pr_b(c=+1)=\tfrac12+b\nu,
 \qquad b\in\{+1,-1\},\quad 0<\nu<1/2.
 \label{eq:accuracy-pair}
\end{equation}
Use the identity data schedule without random dressing. The data channel $\mathcal N_b$ is
\begin{equation}
 \mathcal N_b=(\tfrac12-b\nu)\id
                   +(\tfrac12+b\nu)\operatorname{Ad}_{\bar P}.
 \label{eq:accuracy-data-channel}
\end{equation}
Its noise-aware optimal fidelities satisfy $\FC{b}=\FP{b}=1/2+\nu$. The two models prefer different corrections, and the wrong correction loses $2\nu$. The SPP has uniform Pauli-history probabilities on $\{I,P\}^2$, and every noise syndrome is zero.

Conditioned on $c$, the two-interval process, including its fair sign average, is independent of $b$. Any protocol using $N$ complete, independently reset two-interval queries is therefore a common CPTP postprocessing of $N$ hidden Bernoulli variables $c$. This simulation includes adaptive controls between the intervals, entangled inputs, and a retained private quantum register. Let $\sigma_b$ be the retained calibration state under model $b$. Relative-entropy data processing yields
\begin{equation}
 \begin{split}
 D_{\rm rel}(\sigma_+\Vert\sigma_-)
 &\le N D_{\rm B}(\nu),\\
 D_{\rm B}(\nu)&\coloneqq 2\nu\log\frac{1+2\nu}{1-2\nu}\\
              &=8\nu^2+O(\nu^4).
 \end{split}
 \label{eq:accuracy-kl}
\end{equation}

For completeness, the testing reduction also covers a joint quantum processor $\mathcal G$, whose record register is trivial in this example. Let $\Pi_I,\Pi_P$ be its positive program effects on $\mathcal H_{\rm cal}$ for the two ideal endpoint channels, so that $\Pi_I+\Pi_P\preceq I$. The effect $\Pi\coloneqq (I+\Pi_P-\Pi_I)/2$ defines a measurement-first recovery that applies $\bar P$ on outcome $\Pi$. Its fidelity is
\begin{equation}
 \tfrac12+b\nu\Tr[(\Pi_P-\Pi_I)\sigma_b],
\end{equation}
which is at least the original processor's fidelity $\tfrac12\Tr[(\Pi_I+\Pi_P)\sigma_b] +b\nu\Tr[(\Pi_P-\Pi_I)\sigma_b]$ under each model. The expected fidelity losses of this dominating recovery are $2\nu$ times its model-testing error probabilities. The optimal equal-prior testing error and the quantum Pinsker inequality therefore imply
\begin{equation}
 \begin{split}
 &\max_b\{\FC{b}-\Fid{b}(\mathcal G)\}\\
 &\quad\ge\nu[1-\Dtr(\sigma_+,\sigma_-)]\\
 &\quad\ge\nu\left[1-\sqrt{ND_{\rm B}(\nu)/2}\right].
 \end{split}
 \label{eq:accuracy-loss}
\end{equation}
Choose $\nu=2\epsilon$, with $0<\epsilon\le1/8$. Attaining expected loss at most $\epsilon$ requires
\begin{equation}
 N\ge\frac1{2D_{\rm B}(2\epsilon)}
        =\Omega(\epsilon^{-2}).
 \label{eq:accuracy-lower}
\end{equation}
A known logical eigenstate of a Pauli anticommuting with $\bar P$, followed by its final readout, observes $c$ without error. Bernoulli testing thus attains the same order in independent two-interval queries. The model pair varies with the target accuracy and uses a generally high-weight normalizer. The result is therefore a worst-case accuracy bound, separate from the weak local noise and fixed model pair of the main example.
\section{Finite-angle calculations}\label{app:numerics}
The accompanying enumeration routine specifies the Steane and Shor stabilizer generators, logical representatives, and fixed minimum-weight $X$ lookup tables. Enumerating physical $X$ supports gives the coherent branch amplitudes and, by squaring before grouping, the physical Pauli probabilities. Histories are grouped by identical conditional channels, and a sequential binomial decomposition evaluates their weights. Bernstein tail bounds control the omitted probability, while normalization and the exact means in \cref{eq:readout} check floating-point error.

The two-model minimax value is computed from its one-dimensional piecewise-linear dual over the prior on $\mathsf F$. A record-dependent randomized $I/\Xbar$ decoder attains the same value, checking the primal problem. The six rows in \cref{tab:finite-angle} have primal--dual gaps below $3\times10^{-9}$ and normalization and mean-readout residuals below $3\times10^{-8}$. The physical Pauli-randomization column is computed independently from physical error-string probabilities using \cref{eq:twirl-fidelity}. The same routine checks the retained-coefficient limit and the full-record continuity bounds at scaled storage times. \Cref{fig:separation} displays analytic weak-noise limits, and its panel (b) also includes a directly computed Steane curve at $\theta=0.01$. At $T=T_\theta(1/2)$ and $\theta=0.01$, the exact losses are 0.084596643 for Steane and 0.059386690 for Shor, compared with the leading-order predictions 0.150000000 and 0.086602540. These values show the slow approach to the asymptotic formula. The accompanying source package includes the code definitions, enumeration and optimization routines, and numerical data in the \texttt{anc/} directory, with all figure definitions and plotting data embedded in the manuscript source. The numerical checks use Python's standard library.

\section{List of notations}
\label{app:notation}

The lists below collect the symbols used in more than one place, in order of first appearance within each group. The bracket after each entry gives the defining equation or section. Symbols used only inside a single proof are defined there.

% Keep symbols in the text flow: REVTeX description-list labels can protrude
% into the margin or the neighboring column. Paragraph entries wrap and split
% naturally between columns/pages without a fixed-width label box.
\newenvironment{notationlist}
 {\par\small\setlength{\parindent}{0pt}\setlength{\parskip}{3pt}\raggedright}
 {\par}
\newcommand{\notationentry}[1]{\par\noindent#1\enspace---\enspace\ignorespaces}

\subsection*{Codes and circuits}
\begin{notationlist}
 \notationentry{$\mathcal H_L$, $\mathcal H_{\rm phys}$, $\mathcal H_E$} Logical, physical-data, and environment Hilbert spaces [\cref{sec:setting}]
 \notationentry{$\mathcal B(\mathcal H,\mathcal K)$, $\operatorname{Herm}(\mathcal H)$} Linear maps from $\mathcal H$ to $\mathcal K$, with $\mathcal B(\mathcal H)=\mathcal B(\mathcal H,\mathcal H)$, and Hermitian operators on $\mathcal H$ [\cref{sec:setting}]
 \notationentry{$n$, $k$, $D$} Physical qubits, logical qubits, and logical dimension $D=2^k$ [\cref{sec:setting}]
 \notationentry{$\rho_L$, $\rho'_L$} Input logical state, and output logical state after the final recovery, on $\mathcal H_L$ [\cref{fig:setting}]
 \notationentry{$V$} Encoding isometry $\mathcal H_L\to\mathcal H_{\rm phys}$ [\cref{sec:setting}]
 \notationentry{$T$, $j$} Storage time in noise intervals, and interval index [\cref{fig:setting}]
 \notationentry{$U_j$} Known unitary on $\mathcal H_{\rm phys}\otimes\mathcal H_E$ in interval $j$ [\cref{fig:setting}]
 \notationentry{$\mathcal M$, $\mathcal D$} Syndrome extraction, and the fixed map that returns the logical output [\cref{fig:setting}]
 \notationentry{$s$, $s_j$, $R_s$} Syndrome, syndrome of cycle $j$, and its physical lookup-table correction [\cref{fig:setting}]
 \notationentry{$P_{\rm fr}(y)$} Final Pauli frame of the actively applied corrections [\cref{eq:frame-record-map}]
 \notationentry{$\mathcal P_n$, $\mathcal P_k$} Physical Paulis, and phase-fixed Hermitian logical Paulis [\cref{eq:pauli-twirl,eq:branch}]
 \notationentry{$\C$} CSS code with a fixed $X$-correction lookup table [\cref{def:family}]
 \notationentry{$d_X$, $r$} $X$ distance $d_X=2r+1$ [\cref{def:family}]
 \notationentry{$\Xbar$, $\Zbar$} Logical Paulis, with $Z^{\otimes n}\simeq\Zbar$ meaning that $Z^{\otimes n}$ represents $\Zbar$ [\cref{eq:global-z-conjugation}]
 \notationentry{$N_{d_X}$} Number of minimum-weight supports of the nontrivial logical $X$ class [\cref{eq:orders}]
 \notationentry{$m_s$} Number of weight-$(r+1)$ $X$ errors with syndrome $s$ that the table corrects to $\Xbar$ [\cref{sec:family}]
 \notationentry{$n_s^{(1)}$} Number of weight-one $X$ errors with syndrome $s$ [\cref{eq:delayed-rate}]
\end{notationlist}

\subsection*{Noise models and environment}
\begin{notationlist}
 \notationentry{$h$} Label of the unknown noise model [\cref{eq:instrument}]
 \notationentry{$E$, $I_E$, $\rho_E(h)$} Environment, its identity, and its initial state under $h$ [\cref{fig:setting}]
 \notationentry{$\mathcal T$, $\operatorname{Ad}_Q$} Pauli twirl, and conjugation by $Q$ [\cref{eq:pauli-twirl}]
 \notationentry{SPP} Spatiotemporal Pauli process: the joint distribution of physical Pauli histories after an independent twirl in every interval [\cref{eq:pauli-twirl}]
 \notationentry{$U_\xi(\theta)$} Common-sign $X$ rotation in one interval [\cref{eq:physical}]
 \notationentry{$\theta$} Rotation parameter, half the Bloch-sphere angle [\cref{eq:physical}]
 \notationentry{$\xi$, $\boldsymbol\xi$} Rotation sign in one interval, and sign trajectory [\cref{eq:physical,eq:syndrome}]
 \notationentry{$\mathsf F$, $\mathsf R$} Fixed-sign and resampled-sign models [\cref{eq:physical}]
 \notationentry{$\rho_E^{\mathsf F}$, $\rho_E^{\mathsf R}$} Hidden sign registers of the two models [\cref{eq:classical-environment}]
 \notationentry{$U(\theta,\omega)$, $\omega$} Joint unitary with one environment qubit, and its precession parameter [\cref{eq:quantum-unitary}]
 \notationentry{$\mathrm{Ret}_\omega$, $\mathrm{Res}_\omega$} Environment qubit retained, or reset at every interval boundary [\cref{eq:quantum-unitary}]
 \notationentry{$p_X$, $p_{\rm rep}$} Physical $X$-fault and report-error rates [\cref{sec:equivalence}]
\end{notationlist}

\subsection*{Records, recovery, and fidelity}
\begin{notationlist}
 \notationentry{$y$, $\mathcal Y$} Complete record supplied to the decoder, and its alphabet, with $y=(s_1,\ldots,s_T)$ for $\mathsf F/\mathsf R$ [\cref{eq:instrument}]
 \notationentry{$\Phi_h$} Recorded logical instrument from $\mathcal B(\mathcal H_L)$ to $\mathcal B(\mathcal H_Y\otimes\mathcal H_L)$ [\cref{eq:instrument}]
 \notationentry{$\mathcal N_{h,y}$} Residual logical map for record $y$, before the final recovery: completely positive and trace nonincreasing on $\mathcal B(\mathcal H_L)$ [\cref{eq:instrument}]
 \notationentry{$K_{h,y,\mu}$} Kraus operators of $\mathcal N_{h,y}$, acting on $\mathcal H_L$ [\cref{eq:branch}]
 \notationentry{$w_h(y,L)$} Fidelity contribution of correction $L$ at record $y$ [\cref{eq:branch}]
 \notationentry{$p_h(y)$} Probability of record $y$ for a maximally mixed input [\cref{eq:fidelity}]
 \notationentry{$d(L\mid y)$} Pauli decoder: probability of applying $L$ at record $y$ [\cref{eq:fidelity}]
 \notationentry{$\Fid{h}(d)$, $\Fid{h}(\mathcal R)$} Entanglement fidelity of a Pauli decoder or of a final CPTP recovery $\mathcal R=\{\mathcal R_y\}$ under $h$ [\cref{eq:fidelity}]
 \notationentry{$\FP{h}$, $\FC{h}$} Noise-aware optima over Pauli decoders and over record-dependent final CPTP recoveries [\cref{eq:fidelity}]
 \notationentry{$\Floss{h}^{\mathrm P}$, $\Floss{h}^{\mathrm{CPTP}}$} Fidelity loss relative to $\FP{h}$ or $\FC{h}$ [\cref{eq:fidelity}]
 \notationentry{$|\Omega_D\rangle$, $F_{\mathrm e}(\mathcal M)$} Maximally entangled state, and entanglement fidelity of a map [\cref{eq:kraus-fidelity}]
 \notationentry{$z$, $Q_{\rm cal}$} Finite-sample classical calibration transcript, and retained quantum calibration output [\cref{fig:setting}]
\end{notationlist}

\subsection*{Decision framework}
\begin{notationlist}
 \notationentry{$\widetilde K_{y,\mu}$, $B_{y,L,\mu}$} Joint branches on logical system and environment, and environment components $\mathcal H_E\to\mathcal H_{E'}$ [\cref{eq:environment-effects}]
 \notationentry{$E_{y,L}$} Fidelity effect: positive operator on $\mathcal H_E$ whose expectation is $w_{\rho_E}(y,L)$ [\cref{eq:environment-effects,eq:score}]
 \notationentry{$\Delta E_{y;L,L'}$} Difference operator $E_{y,L}-E_{y,L'}$ [\cref{eq:contrast}]
 \notationentry{$C_i$, $\mathcal V_{\rm cal}$} Calibration effects on $\mathcal H_E$, normalized separately for each setting, and their real span with $I_E$ [\cref{eq:span}]
 \notationentry{$\mathbf p$} Exact calibration probabilities $(\Tr(\rho_EC_i))_i$ [\cref{eq:span}]
 \notationentry{$\Lambda_{\mathbf p}$} Models consistent with exact calibration probabilities $\mathbf p$ [\cref{eq:common-optimum}]
 \notationentry{$\Delta F^{\mathrm P,\star}(\Lambda_{\mathbf p})$} Minimax Pauli fidelity loss over $\Lambda_{\mathbf p}$ [\cref{eq:general-minimax}]
 \notationentry{$\sigma_h$, $\mathcal G$} Calibration state on $\mathcal H_{\rm cal}$, and a common CPTP processor from $\mathcal B(\mathcal H_{\rm cal}\otimes\mathcal H_Y\otimes\mathcal H_L)$ to $\mathcal B(\mathcal H_L)$, with fidelity $\Fid{h}(\mathcal G,\sigma)$ [\cref{eq:general-conflict}]
 \notationentry{$\Gamma^{\mathrm{CPTP}}_{01}$} Equal-prior conflict between two models without calibration [\cref{eq:general-conflict}]
 \notationentry{$L_{\rm sh}$, $\kappa$} Two-round Shor example: logical Pauli of $X_1X_2X_3$, and the sign correlation $\Tr(\rho_EZ_{12})$ [\cref{eq:shor-interactions,eq:shor-crossing}]
\end{notationlist}

\subsection*{Fixed and resampled signs}
\begin{notationlist}
 \notationentry{$K_s(\xi)$} Logical Kraus operator for syndrome $s$, acting on $\mathcal H_L$ [\cref{eq:branch-form}]
 \notationentry{$u_s$, $v_s$, $q_s$} Branch coefficients, with $q_s=u_s^2+v_s^2$ the probability of syndrome $s$ [\cref{eq:branch-form}]
 \notationentry{$A$, $B$} Syndrome-averaged multiplier $A-i\xi B$ of the perpendicular Bloch components [\cref{eq:AB}]
 \notationentry{$a_{\C}$, $b_{\C}$} Leading coefficients of $1-A$ and $B$ [\cref{eq:orders}]
 \notationentry{$\mathcal E_\xi$, $\mathcal E_\pm$} One-cycle instrument with the syndrome retained, for sign $\xi=\pm1$ [\cref{eq:beta}]
 \notationentry{$\beta_{\rm ret}$} Sign sensitivity $\|(\mathcal E_+-\mathcal E_-)/2\|_\diamond$ [\cref{eq:beta}]
 \notationentry{$\langle\Zbar\rangle_{\mathsf F}(T)$, $\langle\Zbar\rangle_{\mathsf R}(T)$} Mean $\Zbar$ readout before final recovery [\cref{eq:readout}]
 \notationentry{$\alpha$, $T_\pi$} Logical rotation angle per cycle under $\mathsf F$, and the storage time at which it accumulates $\pi$ [\cref{eq:horizon}]
 \notationentry{$\phi_{\rm tot}(y)$} Accumulated phase $\sum_j\phi_{s_j}$ of a syndrome history [\cref{eq:conditional}]
 \notationentry{$\phi_s$, $q_y$} Logical rotation angle for syndrome $s$, and history probability $\prod_jq_{s_j}$ [\cref{eq:conditional}]
 \notationentry{$g_h(y)$} Conditional logical multiplier of history $y$ [\cref{eq:conditional}]
 \notationentry{$\mathcal X$} Logical bit flip $\rho\mapsto\Xbar\rho\Xbar$ [\cref{eq:conditional}]
 \notationentry{$\Fmm(T)$} Full-record minimax fidelity loss for $\mathsf F/\mathsf R$, over final CPTP recoveries [\cref{eq:regret}]
 \notationentry{$x$, $T_\theta(x)$, $\chi(x)$} Scaled storage time, $\operatorname{round}(xT_\pi)$, and $[-\cos\pi x]_+$ [\cref{eq:horizon-curve}]
 \notationentry{$\RefInf{\mathsf F}$, $\RefInf{\mathsf R}$} Infidelities of the fixed corrections $\Xbar$ under $\mathsf F$ and $I$ under $\mathsf R$ at $T_\pi$ [\cref{eq:defects}]
\end{notationlist}

\subsection*{Calibration}
\begin{notationlist}
 \notationentry{$\mathcal A_{\rm syn}$, $\mathcal A_{\rm log}$, $\mathcal A_{\rm del}$} Access classes: syndromes only, all cycle outputs, and two-interval delayed extraction [\cref{sec:access}]
 \notationentry{$\mathcal A_{\rm del}(\ell)$, $\ell$} Access class with delayed-extraction queries of length at most $\ell$ [\cref{eq:delay-cost}]
 \notationentry{$M$, $m$, $N_{\rm del}$} Interval budget $M=m+2N_{\rm del}$, complete QEC cycles, and two-interval queries [\cref{sec:access}]
 \notationentry{$\epsilon$} Target fidelity loss, $0<\epsilon<1/2$ [\cref{eq:costs}]
 \notationentry{$M^*_{\mathcal A,\epsilon}$, $M^*_\epsilon(\theta,\ell)$} Smallest budget that reaches loss $\epsilon$ in class $\mathcal A$ or $\mathcal A_{\rm del}(\ell)$ [\cref{eq:costs,eq:delay-cost}]
 \notationentry{$\mathcal E_{\rm av}$, $\mathcal E_{\rm id}$} One extraction after a fair-sign rotation, and the zero-syndrome identity channel [\cref{eq:delayed-mixture}]
 \notationentry{$\mathcal Q_{\mathsf F,2}$, $\mathcal Q_{\mathsf R,2}$} Two-interval delayed-extraction channels [\cref{eq:delayed-mixture}]
 \notationentry{$p_{\rm nz}$, $\kappa_{\C}$} Nonzero-syndrome probability of a two-interval query under $\mathsf F$, and its leading coefficient [\cref{eq:delayed-rate}]
 \notationentry{$S_\lambda$} Sum of $\lambda$ independent fair signs in a length-$\lambda$ query [\cref{eq:delay-cost}]
\end{notationlist}

\subsection*{Modified QEC circuits}
\begin{notationlist}
 \notationentry{$p(y,L)$, $d^\star$} Joint distribution of record and residual logical Pauli after physical randomization, and its optimal decoder [\cref{eq:pauli-optimum}]
 \notationentry{$p_{\rm tw}$, $A_{\rm tw}$} One-cycle logical flip probability after physical randomization, and $A_{\rm tw}=1-2p_{\rm tw}$ [\cref{eq:twirl-fidelity}]
 \notationentry{$F_{\rm tw}(T)$} Fidelity of physical randomization with the final identity correction [\cref{eq:twirl-fidelity}]
 \notationentry{$P$, $\bar P$, $c$} Normalizer Pauli, its logical action, and the model sign [\cref{eq:normalizer-construction}]
\end{notationlist}

\subsection*{General conventions}
\begin{notationlist}
 \notationentry{$J(\mathcal M)$, $\widehat J_{h,y}$} Unnormalized Choi operator, and rescaled logical branch $J(\mathcal N_{h,y})/D$ with trace $p_h(y)$ [\cref{eq:choi-convention,eq:choi-recovery-weights}]
 \notationentry{$|P\rangle\!\rangle$} Unit Pauli--Choi vector $(P\otimes I)|\Omega_D\rangle$ for logical dimension $D$; the physical-noise witness uses dimension $2^n$ [\cref{app:choi-weights,app:quantum}]
 \notationentry{$\|\cdot\|_1$, $\|\cdot\|_\diamond$} Trace norm and diamond norm
 \notationentry{$\Dtr(\rho,\tau)$} Trace distance $\|\rho-\tau\|_1/2$ [\cref{eq:program-bound}]
 \notationentry{$[x]_+$} $\max\{x,0\}$ [\cref{eq:program-bound}]
 \notationentry{$O_{\C}$, $o_{\C}$, $\Theta_{\C,\epsilon}$} Asymptotics as $\theta\to0^+$, with implied constants that depend only on the subscripts: here the code $\C$ with its lookup table and the target $\epsilon$
\end{notationlist}

\end{document}